\documentclass[12pt,journal,onecolumn]{IEEEtran}
\usepackage{silence}
\usepackage{premable_doc}
\usepackage{premable_tikz}
\usepackage{cite}
\usepackage{etoolbox}
\makeatletter
\patchcmd{\@compress@cite}
    {\@citea\@cite@out{#3}}
    {\citedash\@cite@out{#3}}
    {}
    {\PackageError{Journal_v3_verified}{Citation range customization failed}{Check the cite package definition.}}
\makeatother

\usepackage{fancyhdr}  
\usepackage{array}

\newacronym{npd}{NPD}{neural polar decoder}
\newacronym{hy}{HY}{Honda-Yamamoto}
\newacronym{bms}{BMS}{binary memoryless symmetric}

\newcommand{\nnc}[3]{\cG_\mathsf{NN}^{\left(#1, #2, #3\right)}}
\newtheoremstyle{channelstyle} 
{3pt} 
{3pt} 
{\normalfont} 
{} 
{\bfseries} 
{.} 
{.5em} 
{} 

\theoremstyle{channelstyle}

\begin{document}

\title{Neural Polar Codes for Source Coding}
\author{%
\begin{tabular}{@{}>{\centering\arraybackslash}p{0.29\textwidth}
                    >{\centering\arraybackslash}p{0.29\textwidth}
                    >{\centering\arraybackslash}p{0.29\textwidth}@{}}
{\large Bashar Huleihel} & {\large Ziv Aharoni} & {\large Haim H. Permuter} \\
Ben-Gurion University & Duke University & Ben-Gurion University \\
basharh@post.bgu.ac.il & ziv.aharoni@duke.edu & haimp@bgu.ac.il
\end{tabular}}

\maketitle

\begin{abstract}
This paper presents a data-driven approach to lossless and lossy source compression using polar codes. Classical polar source coding achieves fundamental compression limits with low complexity for binary memoryless sources. For unknown $q$-ary sources with memory, however, successive-cancellation decoding requires an explicit source model, and its complexity grows rapidly with alphabet and state-space sizes. For comparison, we present analytical successive-cancellation recursions and show that exact implementations scale quadratically with alphabet size and cubically with state-space size. Motivated by these bottlenecks, we employ neural polar decoders trained directly from data, eliminating the need for an explicit source model and reducing decoding complexity. Once trained, the neural polar decoders approximate conditional distributions of the synthetic polar coordinates. These learned distributions are used both to estimate information-theoretic quantities, including entropy rates, mutual information, and rate--distortion functions, and to construct explicit polar compression schemes. For lossless compression, the neural polar decoder identifies high-entropy polarized coordinates and reconstructs the remaining coordinates through the learned successive-cancellation recursion. For lossy compression, we combine neural polar decoders with a neural test channel to estimate and optimize the mutual-information--distortion tradeoff and construct an explicit lossy polar code. We further extend the framework to the Wyner--Ziv setting with decoder side information by learning conditional polarized distributions and a reconstruction mapping. Experiments on memoryless sources, sources with memory, and real genomic data show that the learned decoders closely match analytical successive-cancellation decoders when available and remain applicable when classical polar methods are computationally infeasible or the required source model is unavailable.
\end{abstract}


\glsresetall
\section{Introduction}

Source coding studies the efficient compression of data generated by a source, requiring exact reconstruction in the lossless setting and reconstruction within a prescribed distortion in the lossy setting.  In lossless compression, the entropy characterizes the fundamental limit for memoryless sources, and the entropy rate plays the corresponding role for stationary sources with memory.  Classical model-based schemes, such as Huffman coding~\cite{huffman1952method} and arithmetic coding~\cite{rissanen1979arithmetic}, approach these limits when supplied with an accurate source model, for example through block coding or context-dependent conditional probabilities, whereas universal schemes such as Lempel--Ziv coding~\cite{ziv1977universal} do not require prior knowledge of the source distribution and are asymptotically optimal for finite-alphabet stationary ergodic sources.  The latter guarantees are asymptotic, however, and finite-block performance can depend strongly on the available sequence length and on the source structure.  For lossy compression, rate--distortion theory characterizes the minimum rate required to reproduce a source within a specified average distortion~\cite{gallager1968information,cover2006elements}.  Despite these foundations, translating information-theoretic limits into tractable coding algorithms remains challenging for many real sources, which may be non-binary, statistically dependent, and structured across time or space, as in text, genomic sequences, images, and scientific measurements.  In such regimes the source law is often unknown, and even when a finite-state or high-order model is available, the inference required by exact model-based coding can be computationally prohibitive.

For binary memoryless sources, polar codes provide constructive coding schemes that achieve the information-theoretic limits with low-complexity encoding and decoding.  Following Ar\i kan's channel-polarization construction~\cite{arikan2009channel}, polar methods have been developed for both lossless and lossy source coding~\cite{cronie2010lossless,korada2010polar,sasoglu2011polarization,mori2014source}.  The polar transform maps a length-$N$ source block to synthetic coordinates whose conditional entropies polarize: most coordinates become nearly deterministic given the preceding ones, while the remaining coordinates carry essentially all of the source uncertainty.  Compression is performed by storing the high-entropy coordinates, and reconstruction is carried out by successive-cancellation (SC) decoding using the source distribution.  For sources with memory, SC decoding can be implemented through successive-cancellation trellis (SCT) recursions, which augment the polar recursion with trellis state tracking~\cite{wang2015construction}.  For binary memoryless sources the SC updates reduce to simple scalar recursions.  For non-binary finite-field alphabets, however, the likelihood recursions become vector-valued, and for finite-state sources they require state tracking, so that the complexity grows rapidly with the alphabet size and the number of source states.

Neural networks (NNs) provide a complementary approach by learning expressive statistical models directly from data.  Autoregressive models, variational autoencoders (VAEs), normalizing flows, and generative-adversarial networks (GANs) have been used to model complex distributions and, in several compression settings, to parameterize encoders, decoders, or test channels~\cite{van2016pixel,theis2015generative,kingma2013auto,rezende2015variational,larochelle2011neural,goodfellow2014generative}.  Neural decoders for polar and related codes, as well as learned distributed compressors, have also been studied, including neural belief-propagation, graph-neural-network, neural successive-cancellation, and learned-binning approaches~\cite{nachmani2019hyper,cammerer2017scaling,doan2018neural,doan2019neural,song2019graph,teng2019low,hebbar2023crisp,makkuva2021ko,xu2020deep,ebada2019deep,kshirsagar2025polar,ozyilkan2024neural}.  These works primarily address channel coding or settings in which the underlying channel or source model is relatively simple.  The present work instead develops a data-driven polar source-coding architecture for finite-alphabet sources, sources with memory, and real-world sources for which analytical SC operations are unavailable or computationally impractical.

In this paper, we develop a neural--polar framework for source coding that combines the structure of polar compression with data-driven estimation.  The approach preserves the polar transform and the SC computation graph, while replacing the analytical core operations of SC decoding, namely the check-node and bit-node updates, with neural modules trained directly from samples.  Building on the neural polar decoding framework originally developed for binary-input channel coding~\cite{aharoni2023polar,aharoni2026optimized,Aharoni2025dna}, we extend neural polar decoders (NPDs) to finite-alphabet source coding.  The framework therefore does not require an explicit probabilistic source model, and its recursive computational cost is governed by the chosen neural architecture and the depth of the polar recursion rather than by the quadratic alphabet enumeration of analytical SC recursions or the cubic state-space dependence of analytical SCT recursions. We focus on single-encoder, single-decoder problems, including Wyner--Ziv coding with decoder-only side information. Extending the framework to multiterminal source coding introduces additional challenges. For example, Slepian--Wolf coding separately encodes correlated sources and jointly decodes them. Its objective is a rate region rather than a single rate, requiring encoder-specific transmitted sets and a joint decoding rule for the desired rate tuple~\cite{abbe2015randomness,korada2010itw_wynerziv}. The main contributions are summarized as follows.
\begin{itemize}
    \item \emph{Lossless compression.} We extend the NPD methodology from binary-input channel decoding to finite-alphabet source coding over memoryless and finite-state sources.  The learned decoder estimates the polarized conditional distributions required for entropy estimation, synthetic-coordinate selection, and SC reconstruction, so that the encoder can select the high-entropy coordinates and the decoder can reconstruct the remaining ones through the learned SC recursion.  In the experiments, the NPD closely reproduces the behavior of the analytical SC or SCT decoder whenever the latter is tractable.  On the \texttt{E. coli} benchmark~\cite{canterbury_large_corpus,canterbury_ecoli_description}, supplying the sequential NPD probabilities to an arithmetic coder yields a lossless rate of $1.883$ bits per base, comparable to the NNCP rate reported in~\cite{yomjinda2020lossless_dna}. This experiment is included to demonstrate applicability to real-world data.
    \item \emph{Lossy compression.} We develop a data-driven variational procedure for lossy source coding in which a neural test channel is optimized against a multi-letter rate--distortion objective~\cite{gallager1968information,cover2006elements,korada2010polar}.  For a fixed test channel, two NPDs estimate the entropy of the polarized reconstruction sequence and its conditional entropy given the source sequence. Their difference is a sample-based estimate of the induced mutual information, following the NPD-based estimator of~\cite{aharoni2026optimized}.  The test channel is then updated with a distortion penalty through a score-function gradient.  The procedure produces estimated rate--distortion operating points, and the same learned polarized statistics specify an explicit polar lossy coding scheme.
    \item \emph{Wyner--Ziv coding.} We extend this estimation--optimization principle to source coding with decoder side information in the Wyner--Ziv setting~\cite{wyner1976rate}.  Polar codes have previously been applied to this problem through nested constructions for binary memoryless models, where the direct construction attains points on the auxiliary Wyner--Ziv functional and the full tradeoff is obtained by convexification, e.g., through time sharing~\cite{korada2010itw_wynerziv,korada2010polar}.  Our formulation jointly learns a multi-letter auxiliary test channel, a reconstruction mapping, and the conditional polarized distributions that define the conditional mutual-information term. The rate-carrying set consists of the synthetic coordinates that are difficult to infer from the side information but predictable from the source through the learned test channel.  To motivate our approach, we establish a nested polar coding theorem in the setting where the induced joint processes admit finite-state, aperiodic, irreducible Markov (FAIM) representations, building on the conditional polarization results of~\cite{shuval2019fast}.
\end{itemize}
Together, these contributions yield a unified neural--polar approach for estimating source-coding limits and constructing lossless, lossy, and Wyner--Ziv codes in regimes where the source law is unknown or the corresponding analytical SC/SCT recursions are computationally prohibitive.

The remainder of the paper is organized as follows.  Section~\ref{sec:pre} introduces the notation and the polar source-coding preliminaries used throughout the paper.  Section~\ref{sec:polar_alphabet_memory} states the analytical SC and SCT recursions for non-binary alphabets and finite-state sources and identifies their complexity with respect to the alphabet and state-space sizes.  Section~\ref{sec: methodology} presents the NPD-based source-coding architecture together with the proposed lossless compression scheme.  Section~\ref{sec: RD} develops the lossy compression framework based on NPD mutual-information estimation and neural test-channel optimization.  Section~\ref{sec:WZ} extends the construction to the Wyner--Ziv setting with decoder side information.  Section~\ref{sec:conc} concludes the paper.  To preserve the flow of the presentation, the proofs and implementation details are given in the appendices.

\section{Notation and Background} \label{sec:pre}

Throughout this paper, random variables (RVs) are denoted by capital letters and their realizations by lower-case letters, e.g., $X$ and $x$, respectively. We write $X_i^j$ for the random vector $(X_i, X_{i+1}, \dots, X_j)$ and $x_i^j$ for its realization when $i < j$. For a sequence $u_1^n$ we use $u_{1,o}^{n}$ and $u_{1,e}^{n}$ to denote the subsequences of odd and even indices, respectively. The probability $\Pr[X=x]$ is denoted by $P_X(x)$. The notation $[N]$ denotes the set $\{1,\dots,N\}$. For a vector $v$, $[v]_i$ denotes its $i$-th component, with multi-index versions (e.g., $[v]_{i,j,s}$) defined analogously when $v$ is an array. The mutual information between two RVs $X$ and $Y$ is denoted by $\f{\sI}{X;Y}$, and the entropy of $X$ is denoted by $\f{\sH}{X}$. Sets are denoted by calligraphic letters, e.g., $\mathcal{X}$. The notation $\mathbb{GF}(q)$ refers to the Galois field with $q$ elements, where $q$ is a prime power, and $\ominus_q$ denotes subtraction in $\mathbb{GF}(q)$, applied element-wise when used on vectors. The notation $A \otimes B$ denotes the Kronecker product when $A$ and $B$ are matrices, and the tensor product when $A$ and $B$ are distributions. We use the conventions $\log 0=-\infty$ for log-probability expressions and $0\log 0=0$ in entropy expressions. Unless stated otherwise, $\log$ denotes the base-2 logarithm, so that entropies, mutual informations, and rates are expressed in bits, and $\ln$ denotes the natural logarithm. 

The class of shallow NNs with fixed input and output dimensions is defined as follows \cite{schafer2006recurrent}.
\begin{definition}[NN function class]\label{def:NN_function_class}
For the ReLU activation function
$\sigma_\mathsf{R}(x) = \max(x,0)$ and $d_i,d_o \in\bN$, define the class of NNs with $k\in\bN$ neurons as:
\begin{align}
        \cG_\mathsf{NN}^{(d_i,k,d_o)} &:= \left\{g:\bR^{d_i}\to\bR^{d_o}: g(x)=\sum_{j=1}^k\beta_j \sigma_\sR( w_j^\top x+b_j),\ x\in\bR^{d_i} \right\},\label{eq:NN_def}
\end{align}
where $\beta_j\in\bR^{d_o}$, $w_j\in\bR^{d_i}$, and $b_j\in\bR$ are the parameters of $g\in \cG_\mathsf{NN}^{(d_i,k,d_o)}$.
\end{definition}

\subsection{Binary Polar Codes for Source Coding}\label{sec:sc_binary}
Polar codes, introduced for channel coding in \cite{arikan2009channel}, also provide a framework for data compression \cite{korada2010polar}, in which polarization segregates the transformed source symbols into those that are highly predictable and those that are not.
Formally, consider a source distribution $P_{X}$ with binary alphabet $\mathcal{X}=\{0,1\}$, and let $X^N$ represent a sequence of $N=2^n$ ($n\in\mathbb{N}$) independent and identically distributed (i.i.d.) source symbols drawn from $P_X$. The process of polarization begins by transforming the input sequence $X^N$ into $U^N$ through multiplication with a polarizing generator matrix $G_N$, such that $U^N = X^NG_N$, where $G_N$ is defined as
\begin{align}\label{eq:gen_mat}
    G_N=B_NG_2^{\otimes n},
\end{align}
with $G_2$ the Arıkan kernel, $G_2^{\otimes n}$ its $n$th Kronecker power, and $B_N$ the bit-reversal matrix \cite{arikan2009channel}. As the block length $N$ increases, the conditional entropy of each bit $u_i$,  given the preceding bits $u^{i-1}$, tends to either $0$ or $1$. The encoder then selects the least predictable bits, defined by the set $A$, as
\begin{align}
    A = A_\tau=\{i\in[N]\;|\;\mathsf{H}(U_i|U^{i-1})>\tau\},
\end{align}
for some chosen threshold $\tau$, which determines the compression rate $R$. Consequently, the encoder outputs $u_A=(x^N \cdot G_N)_{A}$, reducing the number of bits from $N$ to $|A|$, with compression rate $R = \frac{|A|}{N}$.

Following the encoding process, we turn to the SC decoder for polar codes. The SC decoder is defined by the tuple $(\cX, \cY, \cE, E, F, G, H)$, where $\cX, \cY$ are the alphabets of the source and the side information (if available). The term $\cE$ is the embedding space of the SC decoder. The function $E: \cY \to \cE$ represents the input embedding. When side information is not available, it is a constant function. The functions $F$, $G$, and $H$ are the core operations of the SC decoder. Specifically, $F: \cE \times \cE \to \cE$ and $G: \cE \times \cE \times \cX \to \cE$ denote the check-node and bit-node operations, while $H: \cE \to [0,1]$ converts an embedding into a probability for soft decision output. For binary input memoryless sources, suitable choices for $E$, $F$, $G$, and $H$ are:
\begin{align}\label{eqn:sc_ops}
    E &= \ln\frac{\f{P_{X}}{0}}{\f{P_{X}}{1}}, \\
    F(e_1, e_2) &=  2\tanh^{-1}\left(\tanh{\frac{e_1}{2}}\tanh{\frac{e_2}{2}}\right)\label{eq:CN}, \\
    G(e_1, e_2, u) &=  e_2 + (-1)^{u}e_1, \label{eq:BN}\\
    H(e_1) &= \sigma(-e_1),
\end{align}
where $\sigma(\cdot)$ is the sigmoid function, $e_1,e_2\in\bR$, and $u \in \cX$. Applying SC decoding yields an estimate of the polarized bits and their corresponding posterior distribution \cite{arikan2009channel}.
Specifically, the SC decoding performs the map $(u_{A},A) \mapsto \left\{\f{P_{U_i|U_1^{i-1}}}{1| \hat{u}_1^{i-1}},\hat{u}_i\right\}_{i\in[N]}$. For more details on SC decoding, the reader may refer to \cite[Section VIII]{arikan2009channel}. For non-binary source alphabets and sources with memory, these functions are not suitable, and Section \ref{sec:polar_alphabet_memory} presents the appropriate formulas for these scenarios.

\subsection{Lossless compression using polar codes}\label{sec:lossless_correction}
The SC scheme described above may incur decoding errors at finite block lengths. To guarantee exact reconstruction for every source block, it can be augmented with correction information, as introduced in \cite{cronie2010lossless}. The decoder receives not only the compressed message $u_A$ but also a set of indices $T$, defined as
\begin{align}
    T = \{i: i\in A^c, \hat{u}_i\neq u_i\}.
\end{align}
The set $T$ contains the indices in $A^c\triangleq[N]\setminus A$ where the SC estimate $\hat{u}_i$ differs from the true bit $u_i$. Since the encoder knows $u^N$, it determines $T$ by simulating SC decoding from $u_A$, comparing each decision with $u_i$, and correcting any mismatch before proceeding. The encoder outputs $(u_A,T)$. With $A$ known to both encoder and decoder, the rate in bits per source symbol is
\begin{equation}
    R=\frac{|A|+\ell(T)}{N},
    \label{eq:lossless_correction_rate}
\end{equation}
where $\ell(T)$ denotes the number of bits used to encode $T$. For fixed $A$, the expected
rate is $(|A|+\mathbb E[\ell(T)])/N$.

When the decoder receives $(u_A,T)$, it estimates $u_i$ for all $i\in A^c$ using the SC algorithm. However, during the decoding process, if an index $i$ is included in $T$, the decoder corrects the estimated bit. Once the sequence $u_{A^c}$ is reconstructed, it is combined with $u_A$ to recover $u^N$, enabling accurate reconstruction of the original sequence $x^N$ and achieving lossless compression.

\section{Polar Coding Beyond Binary Memoryless Sources}\label{sec:polar_alphabet_memory}
This section presents the analytical SC recursions used as baselines beyond the binary memoryless setting. We first state the non-binary memoryless recursion in the log-likelihood-ratio (LLR) domain and then specialize the trellis-based polar recursion to finite-state sources over a non-binary alphabet. Stating these formulas explicitly exposes the dependence of a straightforward exact implementation on the alphabet size and the number of source states. This computational dependence, together with the need for an explicit source model, motivates the data-driven NPD formulation in Section~\ref{sec: methodology}.

For $\mathcal{X}=\mathbb{GF}(q)$, where $q=p^m$ with $p$ prime, we consider the $2\times2$ kernel over $\mathbb{GF}(q)$~\cite{mori2014source,sasoglu2011polardiscrete}
\begin{align}\label{eq:kernel2}
    G_2 =
    \begin{bmatrix}
        1 & 0 \\
        \alpha & 1
    \end{bmatrix},
\end{align}
where $\alpha \in \mathbb{GF}(q)\setminus\{0\}$. For this kernel, polarization for every memoryless $q$-ary source holds if and only if $\mathbb{GF}(p)(\alpha)=\mathbb{GF}(q)$~\cite[Corollary~15]{mori2014source}, a condition satisfied when $\alpha$ is a primitive element of $\mathbb{GF}(q)$.

\subsection{Memoryless sources}\label{sec:sc}
For prime alphabet sizes, polarization with $\alpha=1$ was established in~\cite{sasoglu2011polardiscrete}. Mori and Tanaka~\cite{mori2014source} characterized the polarizing kernels over general finite fields as stated above. Recursive SC decoding for $q$-ary sources and channels has been studied previously~\cite{feng2020simplified,farsiabi2024fast}. For completeness, we state the LLR-domain source-decoding updates for the kernel in \eqref{eq:kernel2} using the notation adopted in this paper. The LLR representation improves numerical stability.

Let $P_N(u^N)=P_{X^N}(u^NG_N^{-1})$,
and define
\begin{align}\label{eq: pn^i}
    P_N^{(i)}(u^i)=\sum_{u_{i+1}^N}P_N(u^N).
\end{align}
To express the SC updates in the log domain, define the LLR vector
\begin{align}\label{eq:llr_vec}
\mathbf{L}_N^{(i)}(\hat{u}^{i-1})=[L_N^{(i,0)}(\hat{u}^{i-1}),\dots,L_N^{(i,q-1)}(\hat{u}^{i-1})],
\end{align}
where $L_N^{(i,k)}(\hat{u}^{i-1})\triangleq \ln\left(\frac{P_N^{(i)}(\hat{u}^{i-1},u_i=k)}{P_N^{(i)}(\hat{u}^{i-1},u_i=0)}\right)$,
for $k\in\mathcal{X}$. At the first stage of the SC decoding recursion, the LLR vector is given by
\begin{align}\label{eq:source_llr}
    \mathbf{L}_1^{(1)} &= \left[0,\ln\left(\frac{P_X(1)}{P_X(0)}\right),\dots, \ln\left(\frac{P_X(q-1)}{P_X(0)}\right)\right],
\end{align}
which depends only on the source distribution $P_X$.
The following lemma states explicit LLR-domain formulas for the core SC operations over a non-binary alphabet using the finite-field kernel $G_2$ in~\eqref{eq:kernel2}.

\begin{lemma}[LLR-domain non-binary SC]\label{lem:memoryless_sc}
    For a memoryless source $P_{X}$ with $\mathcal{X}=\mathbb{GF}(q)$ and $P_X(x)>0$ for all $x\in\mathcal{X}$, and any $n\geq 0$, $N=2^n$, $k\in\mathcal{X}$, and $1\leq i\leq N$, we have
    \begin{align}
        L_{2N}^{(2i-1,k)}(\hat{u}^{2i-2})&=f^{+}_k\left(\textbf{L}_{N}^{(i)}(\hat{u}_{1,o}^{2i-2}\ominus_q  \alpha\hat{u}_{1,e}^{2i-2}),\textbf{L}_{N}^{(i)}(\hat{u}_{1,e}^{2i-2})\right),\label{eq:cn_memoryless}\\
        L_{2N}^{(2i,k)}(\hat{u}^{2i-1})&=f^{-}_k\left(\textbf{L}_{N}^{(i)}(\hat{u}_{1,o}^{2i-2}\ominus_q  \alpha\hat{u}_{1,e}^{2i-2}),\textbf{L}_{N}^{(i)}(\hat{u}_{1,e}^{2i-2}),\hat{u}_{2i-1}\right)\label{eq:bn_memoryless},
    \end{align}
    where $f^{+}_k:\mathbb{R}^{2q}\to\mathbb{R}$ and $f^{-}_k:\mathbb{R}^{2q}\times\mathcal{X}\to\mathbb{R}$ are defined as
    \begin{align}
        f^{+}_k(l_1,l_2) &=\ln\left(\frac{\sum_{v\in\mathcal{X}}e^{l_1(k\ominus_q \alpha v)}e^{l_2(v)}}{\sum_{v\in\mathcal{X}}e^{l_1(0\ominus_q \alpha v)}e^{l_2(v)}}\right),\label{eq:sc1}\\
        f^{-}_k(l_1,l_2,u) &= l_2(k)+l_1(u\ominus_q \alpha k)-l_1(u),\label{eq:sc2}
    \end{align}
    for $l_1,l_2\in\mathbb{R}^q$ and $u\in\mathcal{X}$.
\end{lemma}
The proof of Lemma~\ref{lem:memoryless_sc} is given in
Appendix~\ref{app:sc_alphbet}. It uses the polar kernel to factor the distribution
of a source block of length $2N$ into two components corresponding to blocks of length $N$.
Marginalizing the undecoded symbols gives the probability recursions for
odd and even indices, whose log ratios yield $f_k^{+}$ and $f_k^{-}$.
These scalar maps collectively define the vector-valued check-node and bit-node operations, which are applied recursively along the standard SC computation graph. At each leaf, the LLR vector is converted into a posterior distribution and a hard decision according to
\begin{align}
    &P_{U_i|U^{i-1}}(\cdot|\hat{u}^{i-1}) = \text{softmax}\left(\textbf{L}_N^{(i)}(\hat{u}^{i-1})\right),\label{eq:hard_dec1}\\
    &\hat{u}_i=\argmax_{u\in\mathcal{X}}P_{U_i|U^{i-1}}(u|\hat{u}^{i-1}).\label{eq:hard_dec2}
\end{align}

In the binary case, each node propagates a single LLR. In the non-binary case, the check-node operation produces $|\mathcal{X}|$ output metrics, each requiring a summation over $|\mathcal{X}|$ symbols in \eqref{eq:sc1}. Thus, a straightforward exact implementation requires $O\left(|\mathcal{X}|^2N\log N\right)$ operations for the full SC recursion, compared with $O(N\log N)$ in the binary case.
\subsection{Sources with memory}\label{sec:sct}
In this section, we consider a finite-state stochastic source characterized by the conditional probability $P_{X_t,S_t|S_{t-1}}$, where $S_{t-1}\in\mathcal{S}$ is a state variable with initial state $S_0$. We assume that the transition law is time-homogeneous and that, conditioned on $S_{t-1}$, the pair $(X_t,S_t)$ is independent of the preceding source symbols and states. Since the memoryless SC operations of Section~\ref{sec:sc} no longer apply directly to such sources, we adopt the trellis-based polar recursion of~\cite{wang2015construction} and state its log-domain specialization to finite-state sources over a non-binary alphabet.

The construction of the polar coding scheme requires the probability $P_{U^i}(u^{i})$, which is expressed as
\begin{align}\label{eq:joint_memory}
    P_{U^{i}}(u^{i})=\sum_{s_0,s_N'}P_{S_0}(s_0)P_{U^i,S_N|S_0}(u^{i},s_N'|s_0).
\end{align}
Let
\begin{align}\label{eq:p_N}
P_N^{(i)}\left(u^i,s_N|s_0\right)=\sum_{u_{i+1}^N\in\mathcal{X}^{N-i}}P_{U^N,S_N|S_0}(u^N,s_N|s_0).
\end{align}
For numerical stability, we state the SCT recursion directly in the log-probability domain, with the convention $\ln 0=-\infty$. Let
\begin{align}
    L_N^{(i)}(u^i,s_N|s_0)\triangleq \ln\left(P_N^{(i)}(u^i,s_N|s_0)\right).
\end{align}
The recursion is initialized at $N=1$ by
\begin{equation}
    L_1^{(1)}(u_1,s_1|s_0)=\ln P_{X_1,S_1|S_0}(u_1,s_1|s_0).
    \label{eq:sct_log_base}
\end{equation}
The following lemma gives the direct log-domain recursion used by the analytical SCT decoder. It is a log-domain source-coding specialization of the probability-domain trellis-based polar recursion in~\cite{wang2015construction} for the finite-field kernel $G_2$ in~\eqref{eq:kernel2}, with the summation over the non-binary source alphabet made explicit.
\begin{lemma}[Log-domain SCT decoding]\label{lem:sct}
    For the time-homogeneous finite-state source described above with source alphabet $\mathcal{X}=\mathbb{GF}(q)$, and any $n\geq0$, $N=2^n$, and $1\leq i\leq N$, the metrics satisfy
    \begin{align}
        L_{2N}^{(2i-1)}(u_1^{2i-1},s_{2N}|s_0)
        &=\ln\Bigg(\sum_{s_N'}\sum_{u_{2i}}e^{L_N^{(i)}(u_{1,o}^{2i}\ominus_q  \alpha u_{1,e}^{2i},s_N'|s_0)} e^{L_N^{(i)}(u_{1,e}^{2i},s_{2N}|s_N')}\Bigg),\label{eq:sct_log_odd}\\
        L_{2N}^{(2i)}(u_1^{2i},s_{2N}|s_0)
        &=\ln\Bigg(\sum_{s_N'}
        e^{L_N^{(i)}(u_{1,o}^{2i}\ominus_q  \alpha u_{1,e}^{2i},s_N'|s_0)} e^{L_N^{(i)}(u_{1,e}^{2i},s_{2N}|s_N')}\Bigg).\label{eq:sct_log_even}
    \end{align}
\end{lemma}
The proof of Lemma~\ref{lem:sct} is given in
Appendix~\ref{app:sct_alphabet}. It splits the source trellis at its midpoint
and sums over the intermediate state. Combining the
two trellis segments through the polar kernel gives the probability-domain
recursions, and taking logarithms yields the stated formulas. For prime alphabets, the standard polar transform guarantees polarization when the state process is stationary, irreducible, and aperiodic~\cite{shuval2019fast}.

Following \eqref{eq:joint_memory}, \(P_{U^{i}}(u^{i})\) can be computed by marginalizing over \(S_0\) and \(S_N\) while using the recursive formulas in Lemma~\ref{lem:sct}. As a result, the complexity of SCT decoding increases significantly compared with the memoryless case. Specifically, each recursion step requires computing the log-probability metrics for all \(S_0\), \(S_N\), and $U_i\in\mathcal{X}$, and summing over $|\mathcal{S}||\mathcal{X}|$ combinations in \eqref{eq:sct_log_odd} and over $|\mathcal{S}|$ intermediate states in \eqref{eq:sct_log_even}. Thus, the computational complexity of SCT decoding is $O\left(|\mathcal{S}|^3|\mathcal{X}|^2N\log N\right)$.

\section{Data-Driven Lossless Polar Compression}\label{sec: methodology}
In this section, we present the NPD, originally introduced for binary-input channel coding in~\cite{aharoni2023polar}, and describe its adaptation to source coding over finite alphabets. We consider a finite-state stochastic source with alphabet $\mathcal{X}=\mathbb{GF}(q)$. Our objective is to construct practical polar compression schemes without requiring an explicit analytical model of the source distribution. Instead, the posterior quantities required for compression are learned directly from source samples.

\subsection{Neural polar decoder}
The NPD integrates NNs into the traditional architecture of SC polar decoders. Let $x^N$ represent the sequence of source symbols, and $u^N = x^N G_N$ be the corresponding polarized symbols. The NPD is designed to reconstruct $u^N$ from the compressed message by modifying the SC decoding process. Unlike conventional SC decoders that rely on predefined functions $E$, $F$, $G$, and $H$, the NPD replaces these with NNs: $E_{\theta_E} \in \nnc{1}{k}{d}$, $F_{\theta_F} \in \nnc{2d}{k}{d}$, $G_{\theta_G} \in \nnc{2d+1}{k}{d}$, and $H_{\theta_H} \in \nnc{d}{k}{|\mathcal{X}|}$, each parameterized by parameters in $\Theta \subset \mathbb{R}^p$, where $\Theta$ is a compact parameter space. Here, $k$ denotes the number of hidden units and $d$ denotes the dimension of the embedding space $\cE\subset\mathbb{R}^d$, and we collectively write $\theta=\{\theta_E,\theta_F,\theta_G,\theta_H\}$. In the experiments, multilayer versions of these modules are used (Appendix~\ref{app:implementation_details}).

The recursive structure of SC decoding is preserved in the NPD, and the parameters $\theta$ are adjusted during a training phase so that the NPD imitates the SC decoder of Section~\ref{sec:sc} for memoryless sources and the SCT decoder of Section~\ref{sec:sct} for sources with memory. The training objective is to minimize the cross-entropy (CE) between the true conditional distributions $P_{U_i \mid U^{i-1}}$ and their approximations $P^{\theta}_{U_i \mid U^{i-1}}$, which the NPD generates by following the SC decoding recursion implemented with the NNs. This approximation enables the estimation of the conditional entropies $\mathsf{H}(U_i \mid U^{i-1})$, which are then used to identify the set $A$ of unpredictable symbols. Since the true distributions are unknown and only finite samples are available, the CE loss is computed as the empirical negative log-likelihood under $P^{\theta}_{U_i \mid U^{i-1}}$. In particular, let $\cD_{M,N} \sim P_{X^{MN}}$ be a sample of $M$ consecutive length-$N$ blocks. At each training iteration, a block $x_{1}^N$ is sampled from $\mathcal{D}_{M,N}$, and the corresponding polarized symbols $u_{1}^N = x_{1}^N G_N$ are computed.
The leaf-wise component of the training loss is
\begin{align}\label{eq: loss}
\cL_{\mathsf{leaf}}(x_{1}^N) = -\frac{1}{N}\sum_{i=1}^N\log P^{\theta}_{U_i|U^{i-1}}(u_{i}|u_{1}^{i-1}),
\end{align}
where each prediction is conditioned on the true preceding symbols.

The computation of $P^{\theta}_{U_i|U^{i-1}}$ proceeds as follows. Let $e_{l,i}$ denote the embedding of the $i$-th symbol at depth $l$ of the SC decoding tree, and $e_{l,1}^N$ denote the collection of embeddings at depth $l$. Since $N=2^n$, the tree has $n+1$ levels indexed by $l=0,\ldots,n$. In particular, $e_{0,i}$ corresponds to the embedding of $X_i$, while $e_{n,i}$ corresponds to that of $U_i$. Correspondingly, let $v_{l,1}^N$ denote the symbols at depth $l$ of the polar computation graph, i.e., the intermediate values of the recursive polar transform, so that $v_{0,1}^N=x_1^N$ and $v_{n,1}^N=u_1^N$. The initial embeddings are constructed as
\[
e_{0,1}^N = E_{\theta}(0^N) \in \mathbb{R}^{N \times d},
\]
where the embedding network \( E_{\theta} \) is applied element-wise to a constant input.
Consequently, each position \( i \in [N] \) is assigned the same \( d \)-dimensional vector \( e_{0,i} \in \mathbb{R}^d \), which encodes the learned prior belief on the source symbols before the polar transformation is applied. The embeddings are then updated recursively along the SC decoding tree. At each depth $j\in\{0,\dots,n-1\}$, the nodes are paired according to the polar computation graph, and each pair of embeddings $(e_{j,i},e_{j,i'})$ produces the two child embeddings at depth $j+1$,
\begin{align}
    e_{j+1,m} &= \f{F_\theta}{e_{j,i},e_{j,i'}}, \\
    e_{j+1,m'} &= \f{G_\theta}{e_{j,i},e_{j,i'}, v_{j+1,m}},
\end{align}
where $m$ and $m'$ index the check-node and bit-node children of the pair, and the bit-node update uses the symbol $v_{j+1,m}$ already available at the sibling node. The rules for pairing the input embeddings and assigning the corresponding child indices follow the standard SC recursion, as specified in Algorithm~\ref{alg:training-nsc} in Appendix~\ref{app:implementation_details}.
Each time a leaf node is reached, the embedding $e_{n,i}$ is passed through the output network $H_{\theta}$, which produces a vector of logits in $\mathbb{R}^{|\mathcal{X}|}$, and a probability mass function over $\mathbb{GF}(q)$ is obtained by applying the softmax function to these logits. That is,
\begin{equation}
    P^{\theta}_{U_i \mid U^{i-1}}(\cdot \mid u^{i-1})
    = \text{softmax}\!\left(H_{\theta}(e_{n,i})\right).
\end{equation}
In the recursive implementation, we also add auxiliary CE terms at the internal levels of the polar decoding tree. At each depth, $H_\theta$ maps each embedding to $|\mathcal{X}|$ logits, and softmax converts these logits to the probability vector supplied to the CE loss. The full per-block loss is
\begin{equation}\label{eq:npd_full_loss}
    \cL(x_1^N)
    = \frac{1}{N}\sum_{l=0}^{n}\sum_{i=1}^{N}
    L^\theta_{\mathsf{ce}}\!\left(\text{softmax}\!\left(H_\theta(e_{l,i})\right),v_{l,i}\right),
\end{equation}
where $L^\theta_{\mathsf{ce}}\!\left(P,u\right)=-\log P(u)$.
The depth-$n$ contribution is $\cL_{\mathsf{leaf}}(x_1^N)$, while depths $0,\ldots,n-1$ provide the auxiliary terms, which supply additional supervision to the shared modules $F_\theta$ and $G_\theta$ across recursion depths and are used only to improve training. The recursive computation of \eqref{eq:npd_full_loss} along the SC decoding tree is given in Algorithm~\ref{alg:training-nsc} in Appendix~\ref{app:implementation_details}. Further implementation details are provided in~\cite{aharoni2023polar, aharoni2026optimized}.

The complete training procedure for the NPD is summarized in Algorithm \ref{alg:memory}, where the weights of $\theta$ are randomly initialized, the block length is set as $N = 2^{n_\mathsf{t}}$ for a fixed predetermined $n_\mathsf{t}$, and the loss $\cL(x_1^N)$ in \eqref{eq:npd_full_loss} of each sampled block is used for an SGD update of $\theta$ until convergence or until the CE loss no longer improves. Note that Algorithm \ref{alg:memory} uses both $u^N$ and $x^N$ to tune the parameters $\theta$. This is true only for the training phase, for the exclusive purpose of learning the parameters $\theta$. After the training phase, the trained parameters $\theta^\ast$ are used for decoding following the procedure in \cite{arikan2009channel,wang2015construction}.

\begin{algorithm}[t]
    \caption{Data-driven polar compressor}
    \label{alg:memory}
    \textbf{input:}
    Dataset $\cD_{M,N}$, block length $n_\mathsf{t}$, \# of information symbols $k$ \\
    \textbf{output:} Unpredictable set $\cA$, $E_\theta, G_\theta, F_\theta,H_\theta$
    \algrule
    \begin{algorithmic}
    \State Initialize the weights of $E_\theta, F_\theta, G_\theta, H_\theta$
    \State Set $N=2^{n_\mathsf{t}}$
    \For{$t = 1$ to $\mathsf{N_{iters}}$}
        \State Sample $x_1^N\sim\cD_{M,N}$
        \State Compute $u_1^N = x_1^NG_N$
        \State Compute $e_{0,1}^N = E_{\theta}(0^N)$
        \State Compute $\cL(x_{1}^N)=\f{\mathsf{Loss}}{e_{0,1}^N, u_1^N, 0}$ in \eqref{eq:npd_full_loss} using Algorithm~\ref{alg:training-nsc}
        \State Minimize $\cL(x_{1}^N)$ w.r.t. $\theta$ using SGD
    \EndFor
    \State Evaluate $\mathbf{L_i}$ for all $i\in[N]$ using \eqref{eq: loss2}
    \State $\cA = \arg\text{sort}\left\{\mathbf{L_1^N}\right\}_{N-k}^N$
    \end{algorithmic}
\end{algorithm}
Importantly, the set $\cA$ of unpredictable symbols is not part of the training procedure itself. Instead, after training terminates, the empirical leaf-wise cross-entropy values are used as estimates of the conditional entropies $\mathsf{H}(U_i \mid U^{i-1})$ to identify the set $A$.
Specifically, for each $i=1,\dots, N$, we evaluate
\begin{align}\label{eq: loss2}
&\mathbf{L}_i = -\frac{1}{M}\sum_{x_1^N\sim \cD_{M,N}}\log P^{\theta}_{U_i|U^{i-1}}(u_i|u_{1}^{i-1}),
\end{align}
and the set $\cA$ is obtained by selecting the top $k$ indices with largest values in $\mathbf{L}_1^N$, denoted by
\begin{align}
 \cA = \arg\text{sort}(\mathbf{L}_1^N)_{N-k}^N.
\end{align}
For a polar message comprising $|\mathcal A|$ transmitted coordinates over an alphabet of cardinality $q$, the rate is $R=|\mathcal A|\log q/N$ bits per source symbol, with $q=|\mathcal X|$ in the present lossless setting.
\begin{remark}\label{rem:npd_complexity}
In contrast to analytical SC-based polar decoders, which require explicit knowledge of the source model to compute the relevant conditional distributions, the NPD learns the decoding operations directly from samples. Moreover, as shown in \cite[Thm.~$4$]{aharoni2023polar}, for a fixed internal neural architecture with hidden width $k$ and embedding dimension $d$, the recursive part of its computational complexity is $O\left(k d N \log N\right)$, with no dependence on $|\mathcal S|$. In the non-binary setting, the softmax head $H_\theta$ adds a linear dependence on $|\mathcal{X}|$ through its output logits when evaluated at the leaves, while the NPD still avoids the explicit $|\mathcal{X}|^2$ alphabet enumeration and $|\mathcal{S}|^3$ state-space enumeration required by the analytical SC/SCT recursions of Section~\ref{sec:polar_alphabet_memory}.
\end{remark}

\subsection{Experiments}
In this section, we evaluate the performance of the NPDs for source coding. The experiments demonstrate two main points. First, the NPD accurately learns to imitate the operations of the SC and SCT decoders across different alphabet sizes and memory orders. Second, the learned decoder remains applicable in regimes where the analytical recursions of Sections~\ref{sec:sc} and~\ref{sec:sct} become prohibitive even for moderate alphabet sizes or memory orders (Remark~\ref{rem:npd_complexity}), as well as on real-world data, for which the classical SC/SCT algorithms cannot be directly applied since the underlying source model and the corresponding analytical decoding operations are unknown.

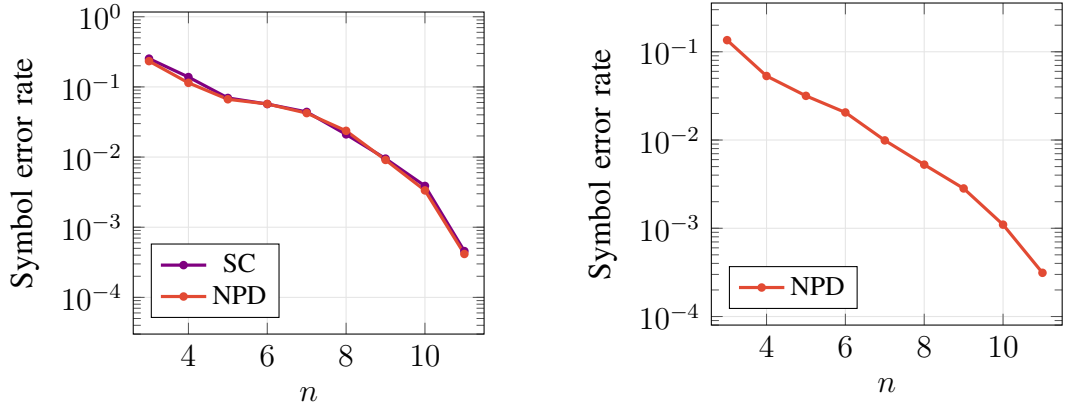
\begin{figure}[t]
     \centering
     \hspace{-3.0cm}
     \begin{subfigure}[]{0.24\textwidth}
        \centering
\begin{tikzpicture}

\definecolor{chocolate2267451}{RGB}{226,74,51}
\definecolor{dimgray85}{RGB}{85,85,85}
\definecolor{gainsboro229}{RGB}{229,229,229}

\begin{axis}[
width=2.45in,
height=2.3in,
log basis y={10},
x grid style={gainsboro229},
xlabel={\(\displaystyle n\)},
xlabel style={font=\normalsize},
xmajorgrids,
xmin=2.6, xmax=11.5,
xtick style={color=dimgray85},
y grid style={gainsboro229},
ylabel={Symbol error rate},
ylabel style={font=\normalsize},
ymajorgrids,
ymin=3e-05, ymax=1.168600008718341,
ymode=log,
ytick style={color=dimgray85},
legend style={at={(0.05,0.05)},anchor=south west, nodes={scale=0.85, transform shape} }
]
\addplot [very thick, violet, mark=*, mark size=1, mark options={solid}]
table {%
3  2.531e-01
4  1.382e-01
5  6.964e-02
6  5.697e-02
7  4.402e-02
8  2.106e-02
9  9.502e-03
10 3.878e-03
11 4.568e-04
};
\addlegendentry{SC}
\addplot [very thick, chocolate2267451, mark=*, mark size=1, mark options={solid}]
table {%
3  2.331e-01
4  1.143e-01
5  6.664e-02
6  5.717e-02
7  4.232e-02
8  2.366e-02
9  9.102e-03
10 3.348e-03
11 4.158e-04
};
\addlegendentry{NPD}
\end{axis}

\end{tikzpicture}
        \label{fig:5_source}
     \end{subfigure}
    \hspace{3cm}
     \begin{subfigure}[]{0.24\textwidth}
        \centering
\begin{tikzpicture}

\definecolor{chocolate2267451}{RGB}{226,74,51}
\definecolor{dimgray85}{RGB}{85,85,85}
\definecolor{gainsboro229}{RGB}{229,229,229}

\begin{axis}[
width=2.45in,
height=2.3in,
log basis y={10},
x grid style={gainsboro229},
xlabel={\(\displaystyle n\)},
xlabel style={font=\normalsize},
xmajorgrids,
xmin=2.6, xmax=11.5,
xtick style={color=dimgray85},
y grid style={gainsboro229},
ylabel={Symbol error rate},
ylabel style={font=\normalsize},
ymajorgrids,
ymin=8e-05, ymax=3.568600008718341e-01,
ymode=log,
ytick style={color=dimgray85},
legend style={at={(0.05,0.05)},anchor=south west, nodes={scale=0.85, transform shape} }
]
\addplot [very thick, chocolate2267451, mark=*, mark size=1, mark options={solid}]
table {%
3  1.354e-01
4  5.314e-02
5  3.166e-02
6  2.053e-02
7  9.905e-03
8  5.261e-03
9  2.827e-03
10 1.099e-03
11 3.128e-04
};
\addlegendentry{NPD}
\end{axis}

\end{tikzpicture}
        \label{fig:23}
     \end{subfigure}
     \caption{The left figure compares symbol error rates (SERs) achieved by the SC decoder and the NPD for a source with alphabet $|\mathcal{X}|=5$ and probability vector $P_X=[0.15,0.15,0.05,0.05,0.6]$ at rate $R=1.95$. The right figure shows the SERs achieved by the NPD for a memoryless source with alphabet $|\mathcal{X}|=23$ and entropy $\mathsf{H}(X)=2.0176$ at rate $R=3.2$.}
     \label{fig:alphabet}
\end{figure}

\subsubsection{Memoryless Sources with Large Alphabet Sizes}

We begin with experiments on memoryless sources to assess the behavior of the NPDs for increasing alphabet sizes. Figure \ref{fig:alphabet} shows results for a source with alphabet size $|\mathcal{X}|=5$ and probability vector $$P_X = [0.15,\,0.15,\,0.05,\,0.05,\,0.6].$$
For different block lengths, the symbol error rates (SERs) obtained by the NPD closely follow those achieved by the SC decoder described in Section~\ref{sec:polar_alphabet_memory}. This shows that the NPD learns the decoding operations of the SC algorithm directly from data, without requiring analytical knowledge of the source model.

Figure~\ref{fig:alphabet} also presents results for a memoryless source with alphabet size $|\mathcal{X}|=23$ and entropy $\mathsf{H}(X)=2.0176$. For this alphabet size, classical SC decoding is computationally infeasible with our implementation because of the explicit $|\mathcal{X}|^2$ alphabet enumeration in \eqref{eq:sc1}. We therefore report only the NPD results. The NPD uses fixed neural update modules together with a softmax head whose output cost is linear in $|\mathcal{X}|$, and exhibits the expected polarization behavior as the SER decreases with increasing block length. This demonstrates that the NPD enables polar-based source coding in large-alphabet settings.

\subsubsection{Sources With Memory}
The left subfigure of Figure~\ref{fig:memory} illustrates the bit error rates (BERs) for a binary second-order Markov source. The performance of the NPD closely follows that of the analytical SCT decoder described in Section~\ref{sec:sct}, confirming that the NPD successfully learns the trellis-based recursion in Lemma~\ref{lem:sct}, without requiring access to the underlying Markov transition model.

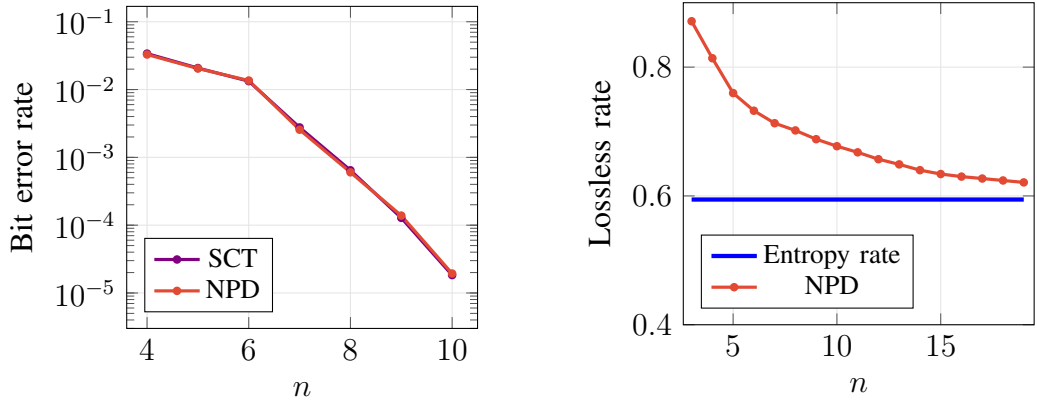
\begin{figure}[t]
     \centering
     \hspace{-3.0cm}
     \begin{subfigure}[]{0.24\textwidth}
        \centering
\begin{tikzpicture}

\definecolor{chocolate2267451}{RGB}{226,74,51}
\definecolor{dimgray85}{RGB}{85,85,85}
\definecolor{gainsboro229}{RGB}{229,229,229}

\begin{axis}[
width=2.45in,
height=2.3in,
log basis y={10},
x grid style={gainsboro229},
xlabel={\(\displaystyle n\)},
xlabel style={font=\normalsize},
xmajorgrids,
xmin=3.6, xmax=10.5,
xtick style={color=dimgray85},
y grid style={gainsboro229},
ylabel={Bit error rate},
ylabel style={font=\normalsize},
ymajorgrids,
ymin=3e-06, ymax=0.168600008718341,
ymode=log,
ytick style={color=dimgray85},
legend style={at={(0.05,0.05)},anchor=south west, nodes={scale=0.85, transform shape} }
]
\addplot [very thick, violet, mark=*, mark size=1, mark options={solid}]
table {%
4  3.390e-02
5  2.076e-02
6  1.331e-02
7  2.762e-03
8  6.447e-04
9  1.282e-04
10 1.835e-05
};
\addlegendentry{SCT}

\addplot [very thick, chocolate2267451, mark=*, mark size=1, mark options={solid}]
table {%
4  3.287e-02
5  2.044e-02
6  1.359e-02
7  2.559e-03
8  6.016e-04
9  1.385e-04
10 1.935e-05
};
\addlegendentry{NPD}
\end{axis}

\end{tikzpicture}
        \label{fig:Markov}
     \end{subfigure}
    \hspace{3cm}
     \begin{subfigure}[]{0.24\textwidth}
        \centering
\begin{tikzpicture}

\definecolor{chocolate2267451}{RGB}{226,74,51}
\definecolor{dimgray85}{RGB}{85,85,85}
\definecolor{gainsboro229}{RGB}{229,229,229}

\begin{axis}[
width=2.45in,
height=2.3in,
log basis y={10},
x grid style={gainsboro229},
xlabel={\(\displaystyle n\)},
xlabel style={font=\normalsize},
xmajorgrids,
xmin=2.6, xmax=19.5,
xtick style={color=dimgray85},
y grid style={gainsboro229},
ylabel={Lossless rate},
ylabel style={font=\normalsize},
ymajorgrids,
ymin=0.4, ymax=0.9,
ytick style={color=dimgray85},
legend style={at={(0.05,0.05)},anchor=south west, nodes={scale=0.85, transform shape} }
]
\addplot [ultra thick, blue]
table {%
3  0.5944
4  0.5944
5  0.5944
6  0.5944
7  0.5944
8  0.5944
9  0.5944
10 0.5944
11 0.5944
12 0.5944
13 0.5944
14 0.5944
15 0.5944
16 0.5944
17 0.5944
18 0.5944
19 0.5944
};
\addlegendentry{Entropy rate}
\addplot [very thick, chocolate2267451, mark=*, mark size=1, mark options={solid}]
table {%
3  0.8712
4  0.8138
5  0.7597
6  0.7323
7  0.7129
8  0.7017
9  0.6881
10 0.6772
11 0.6677
12 0.657
13 0.649
14 0.64
15 0.634
16 0.63
17 0.627
18 0.624
19 0.621
};
\addlegendentry{NPD}
\end{axis}

\end{tikzpicture}
        \label{fig:markov6}
     \end{subfigure}
     \caption{The left subfigure compares the bit error rates (BERs) achieved by the SCT decoder and the NPD on a second-order Markov source. The right subfigure shows the average zero-error lossless compression rate achieved by the NPD, including the bits used to encode $T$ in \eqref{eq:lossless_correction_rate}, as a function of block size parameter $n$ for a 6th-order Markov source model with an entropy rate of $0.595$.}
     \label{fig:memory}
\end{figure}
To evaluate scalability to larger memory orders, the right subfigure of Figure~\ref{fig:memory} presents results for a sixth-order Markov source using the correction scheme of Section~\ref{sec:lossless_correction}, which guarantees exact reconstruction at every block length. In this case, the analytical SCT recursion has $|\mathcal{S}|=|\mathcal{X}|^{6}$ states, and its $|\mathcal{S}|^{3}$ dependence makes it computationally impractical. We therefore do not run this decoder. Instead, we compare the average lossless compression rate $(|A|+\mathbb E[\ell(T)])/N$, including the correction overhead, with the theoretical entropy rate. As the block length increases, the compression rate achieved by the NPD approaches the entropy rate, demonstrating that the learned polar scheme captures long-range dependencies even in high-order Markov processes while eliminating the complexity bottleneck of SCT decoding.

\subsubsection{Real-world data}
In addition to the synthetic experiments, we further evaluate the proposed approach on real-world genomic data using the \texttt{E. coli} file from the Large Corpus of the Canterbury Corpus collection~\cite{canterbury_large_corpus}. The file contains the complete genome of the \emph{Escherichia coli} bacterium and is listed in the Canterbury Corpus with size $4{,}638{,}690$ bytes.

For the \texttt{E. coli} experiment, each source block is transformed as
$U^N=X^NG_N$. The trained NPD provides the sequential probabilities
$P_\theta(U_i|U^{i-1})$ through its SC computation graph, and these
probabilities drive an arithmetic coder that losslessly encodes $U^N$.
Because $G_N$ is invertible, decoding $U^N$ gives exact recovery of $X^N$.
Combining the NPD probabilities with arithmetic coding yields a lossless
compression rate of $1.883$ bpb, which is comparable to the $1.885$ bpb NNCP result
in~\cite{yomjinda2020lossless_dna}. Arithmetic coding is also used by neural DNA compressors that
combine a learned sequential probability model with an arithmetic
coder~\cite{yomjinda2020lossless_dna}. The purpose of this experiment is to
demonstrate applicability to real-world genomic data rather than to claim a
performance improvement. Specifically, Table~\ref{tab:ecoli_comparison} lists
representative final compression rates for the same \texttt{E. coli} file. These rates are not directly comparable because our evaluation measures the coding rate on held-out validation data using an NPD trained on a separate portion of the sequence, whereas the published baselines in~\cite{yomjinda2020lossless_dna} use adaptive whole-file compression. The implementation details for this experiment, including the NPD architecture, training parameters, DNA preprocessing, and rate accounting, are summarized in Appendix~\ref{app:implementation_details}.

\begin{table}[t]
\centering
\caption{Representative final bits-per-base (bpb) compression rates on the
\texttt{E. coli} benchmark. The proposed rate is obtained by arithmetic coding
driven by the NPD probabilities. The remaining rates are taken from~\cite[Table~11]{yomjinda2020lossless_dna}, where the two NNCP configurations differ in batch size and in (time steps, segment length).}
\label{tab:ecoli_comparison}
\begin{tabular}{lc}
\hline
Method & Rate (bpb) \\
\hline
Proposed NPD & $1.883$ \\
LSTM (NNCP), $4\times352$, batch $8$, $(24,24)$ & $1.885$ \\
LSTM-Stateful, $4\times352$ & $1.893$ \\
LSTM (NNCP), $4\times352$, batch $16$, $(20,20)$ & $1.895$ \\
Bi-LSTM, $4\times32$ & $1.900$ \\
PAQ8l & $1.900$ \\
PAQ8a & $1.902$ \\
TRFCP, $4\times256$ & $1.913$ \\
LSTM+2DCNN & $1.933$ \\
PPM(3) & $1.964$ \\
\hline
\end{tabular}
\end{table}

\section{Lossy Compression via NPD}
\label{sec: RD}
Building on the lossless NPD construction in Section~\ref{sec: methodology},
we now turn to lossy compression, whose fundamental
limits are described by rate--distortion theory. Rather than reconstructing $X^N$ exactly, the compressed
representation is allowed to produce an approximation $\hat X^N$. We
measure the discrepancy between the source and reconstruction by a
single-letter distortion measure
$d:\mathcal X\times \mathcal X\to \mathbb{R}_{\ge 0}$, with
$d_{\max}\triangleq\max_{x,\hat x\in\mathcal X}d(x,\hat x)<\infty$.
The normalized additive block distortion is
\begin{equation}
    d_N(x^N,\hat x^N)
    \triangleq
    \frac{1}{N}
    \sum_{i=1}^{N}
    d(x_i,\hat x_i).
    \label{eq:normalized_block_distortion}
\end{equation}
The compression scheme is required to satisfy an average per-letter
distortion constraint $D\in\mathbb{R}_{\ge 0}$.

For stationary ergodic sources, Gallager~\cite{gallager1968information}
showed that the rate--distortion function is characterized by a
multi-letter optimization. In particular, the $N$-th order
rate--distortion function is given by
\begin{equation}
    R_N(D)
    =
    \frac{1}{N}
    \min_{P_{\hat X^N|X^N}}
    \mathsf{I}(X^N;\hat X^N)
    \quad
    \text{subject to}
    \quad
    \mathbb{E}
    \left[
        d_N(X^N,\hat X^N)
    \right]
    \le D,
    \label{eq:rd_problem}
\end{equation}
and the operational rate--distortion function $R(D)$ is obtained as
$N\to\infty$.

In the following, we propose a practical framework that jointly estimates
the rate--distortion curve and then constructs an explicit lossy coding scheme.
The proposed method operates directly on source samples, applies to
stationary ergodic sources with or without memory, and learns both the
test channel and the corresponding coding rule.
Section~\ref{sec:RD_pmf} presents the neural parameterization of the test
channel $P_{\hat X^N|X^N}^{\phi}$, Section~\ref{sec:RD_estimation}
describes how NPDs estimate the mutual information for a fixed test channel, and Section~\ref{sec:RD_opt} presents the test-channel optimization step, the resulting alternating scheme, and the lossy polar code construction.

\subsection{Neural Test Channel}
\label{sec:RD_pmf}

To model the joint distribution between the source and its reconstruction,
we use a parametric test channel $P_{\hat X^N|X^N}^{\phi}$, where $\phi$
denotes the collection of its parameters. We do not restrict attention to
memoryless test channels. Instead, we allow
$P_{\hat X^N|X^N}^{\phi}$ to capture temporal dependencies across the
block. This generality is essential for stationary ergodic sources beyond
the i.i.d.\ setting.

The test channel is parameterized by an encoder--decoder transformer
operating on length-$N$ blocks~\cite{vaswani2017attention}. The encoder
processes the complete source block $x^N$, including positional embeddings.
At index $i$, the decoder uses causal self-attention over the embeddings of
the preceding reconstruction symbols $\hat x^{i-1}$ and cross-attention to
all encoder outputs. Denoting the resulting transformer mapping by
$T_\phi$, we define
$o_i^\phi(x^N,\hat x^{i-1})=T_\phi(x^N,\hat x^{i-1},i)
\in\mathbb R^{|\mathcal X|}$, where $\hat x^0$ is the empty sequence.
Applying the softmax function to $o_i^\phi$ gives
\[
    P_{\hat X_i|X^N,\hat X^{i-1}}^{\phi}
    (a|x^N,\hat x^{i-1})
    =
    \frac{\exp\!\left(o_i^\phi(x^N,\hat x^{i-1})_a\right)}
    {\sum_{b\in\mathcal X}
    \exp\!\left(o_i^\phi(x^N,\hat x^{i-1})_b\right)},
    \qquad a\in\mathcal X.
\]
The resulting block test channel has the chain-rule factorization
\begin{equation}
    P_{\hat X^N|X^N}^{\phi}(\hat x^N|x^N)
    =
    \prod_{i=1}^{N}
    P_{\hat X_i|X^N,\hat X^{i-1}}^{\phi}
    (\hat x_i|x^N,\hat x^{i-1}).
    \label{eq:test_channel_factorization}
\end{equation}
The normalized factors in \eqref{eq:test_channel_factorization} define a
valid conditional distribution, and every $P_{\hat X^N|X^N}$ admits this
factorization by the chain rule. The transformer restricts
\eqref{eq:rd_problem} to the family
$\{P_{\hat X^N|X^N}^{\phi}\}_{\phi}$, with each symbol conditioned on the
complete source block.

\subsection{Mutual Information Estimation}
\label{sec:RD_estimation}

For lossy compression, the fundamental quantity of interest is the mutual information induced by the
test channel. Following the NPD-based mutual information estimator
in~\cite{aharoni2026optimized}, we next describe how this quantity is
estimated using the NPD of Section~\ref{sec: methodology} for a fixed neural test channel.

For a given $\phi$, the source distribution and the neural test channel
induce the joint distribution
\[
    P_{X^N,\hat X^N}^{\phi}(x^N,\hat x^N)
    =
    P_{X^N}(x^N)
    P_{\hat X^N|X^N}^{\phi}(\hat x^N|x^N).
\]
Let $U^N=\hat X^N G_N$. Since the mapping between $\hat X^N$ and $U^N$
is bijective, the mutual information can be written in the polar domain as
\begin{align}
    \mathsf{I}_\phi(X^N;\hat X^N)
    &=
    \mathsf{I}_\phi(X^N;U^N) \nonumber\\
    &=
    \mathsf{H}_\phi(U^N)-\mathsf{H}_\phi(U^N|X^N).
    \label{eq:MI_dist}
\end{align}
Thus, the mutual information can be computed by estimating the
unconditional entropy of $U^N$ and the conditional entropy of $U^N$ given $X^N$.

Both entropy terms are estimated using an NPD, but with different embedding
networks. To this end, we use two parameter sets,
\begin{align*}
    \theta^{\mathrm{const}}
    &=
    \{E_{\theta_{E_1}},F_{\theta_F},G_{\theta_G},H_{\theta_H}\},
    \\
    \theta^{\mathrm{cond}}
    &=
    \{E_{\theta_{E_2}},F_{\theta_F},G_{\theta_G},H_{\theta_H}\},
\end{align*}
where the core components $F_{\theta_F}$, $G_{\theta_G}$, and
$H_{\theta_H}$ are shared, while the embedding networks are different.
The entropies in \eqref{eq:MI_dist} are determined by the test-channel
distribution parameterized by $\phi$. The NPD parameters $\theta$ affect
only their estimates. The parameter set
$\theta^{\mathrm{const}}$ is used to estimate
$\mathsf{H}_\phi(U^N)$ and employs the constant embedding
\[
    \mathbf e_0^{\mathrm{const}}
    =
    E_{\theta_{E_1}}(0^N)
    \in \mathbb{R}^{N\times d},
\]
following the procedure described in Section~\ref{sec: methodology}. In
contrast, $\theta^{\mathrm{cond}}$ is used to estimate
$\mathsf{H}_\phi(U^N|X^N)$ and employs the source-dependent embedding
\[
    \mathbf e_0^{\mathrm{cond}}
    =
    E_{\theta_{E_2}}(x^N)
    \in \mathbb{R}^{N\times d},
\]
which provides the NPD with the source realization and enables it
to approximate the conditional distributions
$P_{U_i|U^{i-1},X^N}^{\phi}$.

Let
\[
    \mathcal D_{M,N}^{\phi}
    =
    \{(x_j^N,\hat x_j^N)\}_{j=1}^{M}
    \sim
    P_{X^N,\hat X^N}^{\phi}
\]
be a sample set generated from the source and the fixed test channel. For
each pair $(x_j^N,\hat x_j^N)$, we compute
$u_j^N=\hat x_j^N G_N$. The NPD losses associated with the unconditional
and conditional entropy estimates are
\begin{align}
    \mathcal L_{\mathrm{const}}(\hat x_j^N)
    &=
    -\frac{1}{N}
    \sum_{i=1}^{N}
    \log
    P_{U_i|U^{i-1}}^{\theta^{\mathrm{const}}}
    (u_{j,i}|u_j^{i-1}),
    \label{eq:loss_const_mi}
    \\
    \mathcal L_{\mathrm{cond}}(x_j^N,\hat x_j^N)
    &=
    -\frac{1}{N}
    \sum_{i=1}^{N}
    \log
    P_{U_i|U^{i-1},X^N}^{\theta^{\mathrm{cond}}}
    (u_{j,i}|u_j^{i-1},x_j^N).
    \label{eq:loss_cond_mi}
\end{align}
To distinguish the learned estimates from the exact entropy terms, define
\begin{align*}
    \widehat{\mathsf H}_{\theta^{\mathrm{const}}}^{\phi}(U^N)
    &:=
    \frac{N}{M}
    \sum_{j=1}^{M}
    \mathcal L_{\mathrm{const}}(\hat x_j^N),
    \\
    \widehat{\mathsf H}_{\theta^{\mathrm{cond}}}^{\phi}(U^N|X^N)
    &:=
    \frac{N}{M}
    \sum_{j=1}^{M}
    \mathcal L_{\mathrm{cond}}(x_j^N,\hat x_j^N).
\end{align*}
Here, the subscript identifies the NPD parameters, while the superscript
$\phi$ indicates that the empirical averages are evaluated on
$\mathcal D_{M,N}^{\phi}$. These quantities estimate
$\mathsf H_\phi(U^N)$ and
$\mathsf H_\phi(U^N|X^N)$, respectively. Accordingly, the normalized
mutual information is estimated by
\begin{equation}
    \begin{aligned}
        \frac{1}{N}
        \widehat{\mathsf{I}}_{\theta}(\mathcal D_{M,N}^{\phi})
        &:=
        \frac{1}{N}
        \left[
            \widehat{\mathsf H}_{\theta^{\mathrm{const}}}^{\phi}(U^N)
            -
            \widehat{\mathsf H}_{\theta^{\mathrm{cond}}}^{\phi}(U^N|X^N)
        \right]
        \\
        &=
        \frac{1}{M}
        \sum_{j=1}^{M}
        \left[
            \mathcal L_{\mathrm{const}}(\hat x_j^N)
            -
            \mathcal L_{\mathrm{cond}}(x_j^N,\hat x_j^N)
        \right],
    \end{aligned}
    \label{eq:npd_mi_estimator}
\end{equation}
where $\theta=\{\theta^{\mathrm{const}},\theta^{\mathrm{cond}}\}$.

We use \eqref{eq:npd_mi_estimator} as a plug-in mutual-information
estimate. Consistency of this construction is established
in~\cite[Theorem~4]{aharoni2026optimized}
for binary-input indecomposable finite-state channels under the stated
neural-approximation and loss-minimization conditions. 

\subsection{Distortion-Constrained Mutual Information Optimization}
\label{sec:RD_opt}

We now describe the update of the test-channel parameters $\phi$, which is performed with the NPD parameters held fixed and used only to evaluate the mutual-information term.
Following the $N$-th order rate--distortion formulation in
\eqref{eq:rd_problem}, we consider the Lagrangian
\begin{equation}
    J_\lambda(\phi)
    \triangleq
    \frac{1}{N} \mathsf{I}_\phi(X^N;\hat X^N)
    +
    \lambda
    \mathbb{E}_{P_{X^N,\hat X^N}^{\phi}}
    \left[
        d_N(X^N,\hat X^N)
    \right],
    \label{eq:rd_lagrangian}
\end{equation}
where $\lambda\ge 0$ controls the rate--distortion tradeoff.

Let
\begin{equation}
    \imath_\phi(x^N,\hat x^N)
    \triangleq
    \log
    \frac{
        P_{U^N|X^N}^{\phi}(u^N|x^N)
    }{
        P_{U^N}^{\phi}(u^N)
    },
    \qquad
    u^N = \hat x^N G_N ,
    \label{eq:rd_information_density}
\end{equation}
denote the information density in bits induced by
$P_{X^N,\hat X^N}^{\phi}$. The minimization of
\eqref{eq:rd_lagrangian} requires gradients with respect to the
test-channel parameters. Since $\hat X^N$ is sampled from the discrete
test channel $P_{\hat X^N|X^N}^{\phi}$, the sampling operation is not
directly differentiable. To estimate the gradient $\nabla_\phi J_\lambda(\phi)$, we employ the score-function (REINFORCE) estimator \cite{williams1992reinforce}. Specifically, the following theorem gives a score-function
representation, which enables
Monte-Carlo gradient estimation from samples drawn according to
$P_{X^N,\hat X^N}^{\phi}$.

\begin{theorem}\label{th:pg_rd}
Under the softmax test-channel parameterization in
Section~\ref{sec:RD_pmf}, assume that the test-channel probabilities are
differentiable with respect to $\phi$. Then the gradient of
$J_\lambda(\phi)$ is
\begin{equation}
    \nabla_\phi J_\lambda(\phi)
    =
    \mathbb{E}_{P_{X^N,\hat X^N}^{\phi}}
    \left[
        \nabla_\phi
        \ln
        P_{\hat X^N|X^N}^{\phi}(\hat X^N|X^N)
        \left(
            \frac{1}{N}\imath_\phi(X^N,\hat X^N)
            +
            \lambda d_N(X^N,\hat X^N)
        \right)
    \right].
    \label{eq:rd_exact_gradient}
\end{equation}
\end{theorem}
The proof of Theorem~\ref{th:pg_rd} is given in
Appendix~\ref{app:PG}. It differentiates the mutual-information and distortion
terms separately using the score-function identity. For the
mutual-information term, the expected derivative of the information density
vanishes by normalization, leaving the expression in
\eqref{eq:rd_exact_gradient}.

In the data-driven setting, the exact information density in
\eqref{eq:rd_information_density} is unavailable, and we replace it
by the NPD-based plug-in estimate obtained in the estimation step, in which the
conditional and unconditional distributions are approximated by
$P_{U^N|X^N}^{\theta^{\mathrm{cond}}}$ and
$P_{U^N}^{\theta^{\mathrm{const}}}$, respectively. Thus, for fixed NPD
parameters
$\theta=\{\theta^{\mathrm{const}},\theta^{\mathrm{cond}}\}$ during the
test-channel update, we define
\begin{align}
    \hat \imath_{\theta}(x^N,\hat x^N)
    &\triangleq
    \log
    \frac{
        P_{U^N|X^N}^{\theta^{\mathrm{cond}}}(u^N|x^N)
    }{
        P_{U^N}^{\theta^{\mathrm{const}}}(u^N)
    }
    \nonumber\\
    &=
    \sum_{i=1}^{N}
    \log
    \frac{
        P_{U_i|U^{i-1},X^N}^{\theta^{\mathrm{cond}}}
        (u_i|u^{i-1},x^N)
    }{
        P_{U_i|U^{i-1}}^{\theta^{\mathrm{const}}}
        (u_i|u^{i-1})
    },
    \qquad
    u^N=\hat x^N G_N .
    \label{eq:rd_npd_information_density}
\end{align}

For a mini-batch $\mathcal D_{M,N}^{\phi}$ as in Section~\ref{sec:RD_estimation},
drawn independently at the current $\phi$ with the NPD parameters held fixed, substituting
\eqref{eq:rd_npd_information_density} for the exact information density in
\eqref{eq:rd_exact_gradient} gives the plug-in
score-function gradient estimate
\begin{equation}
    \widehat{\nabla_\phi J_\lambda}
    (\theta,\phi;\mathcal D_{M,N}^{\phi})
    =
    \frac{1}{M}
    \sum_{j=1}^{M}
    \nabla_\phi
    \ln
    P_{\hat X^N|X^N}^{\phi}(\hat x_j^N|x_j^N)
    \left[
        \frac{1}{N}
        \hat \imath_{\theta}(x_j^N,\hat x_j^N)
        +
        \lambda d_N(x_j^N,\hat x_j^N)
    \right].
    \label{eq:rd_empirical_score_gradient}
\end{equation}
For fixed $\theta$, \eqref{eq:rd_empirical_score_gradient} is an unbiased
estimate of
$\nabla_\phi\,\mathbb{E}_{P_{X^N,\hat X^N}^{\phi}}
[N^{-1}\hat\imath_{\theta}(X^N,\hat X^N)+\lambda d_N(X^N,\hat X^N)]$,
the gradient of the learned surrogate objective in which $\hat\imath_\theta$
is held fixed, but not necessarily of $\nabla_\phi J_\lambda(\phi)$. When $P_{U^N|X^N}^{\theta^{\mathrm{cond}}}$ and $P_{U^N}^{\theta^{\mathrm{const}}}$
coincide with $P_{U^N|X^N}^{\phi}$ and $P_{U^N}^{\phi}$ at the current
$\phi$, then $\hat\imath_\theta=\imath_\phi$ and, by
Theorem~\ref{th:pg_rd}, the estimate is unbiased for
$\nabla_\phi J_\lambda(\phi)$.

The test-channel parameters are then updated according to
\begin{equation}
    \phi
    \leftarrow
    \phi
    -
    \gamma
    \widehat{\nabla_\phi J_\lambda}
    (\theta,\phi;\mathcal D_{M,N}^{\phi}).
    \label{eq:rd_phi_update}
\end{equation}

This optimization step is alternated with the NPD estimation
step of Section~\ref{sec:RD_estimation}: for the
current test channel, the NPD parameters are updated to estimate the
mutual-information term, and then, with the NPD parameters fixed, the test
channel is updated using \eqref{eq:rd_empirical_score_gradient}. Before
the first update of $\phi$, a short warm-up stage updates only
$\theta=\{\theta^{\mathrm{const}},\theta^{\mathrm{cond}}\}$
using $\mathcal L_{\mathrm{const}}$ and $\mathcal L_{\mathrm{cond}}$ with $\phi$ fixed,
which provides the initial estimate of the information density
needed for a meaningful first update of the test-channel parameters.
Repeating these steps jointly estimates the rate--distortion tradeoff and
learns the polarized statistics from which an explicit lossy polar coding scheme is constructed below. The resulting alternating optimization procedure is
illustrated in Fig.~\ref{fig:rd_framework_diagram} and summarized in
Algorithm~\ref{alg:RD}.

\begin{figure}[t]
    \centering
    \resizebox{0.98\textwidth}{!}{\begin{tikzpicture}[
  data/.style={
      ellipse,
      draw=blue!65!black,
      fill=blue!6,
      align=center,
      minimum height=7mm,
      minimum width=16mm,
      font=\scriptsize,
      inner sep=1.4pt
  },
  process/.style={
      rectangle,
      draw=black!60,
      rounded corners=1.5pt,
      fill=gray!8,
      align=center,
      minimum height=9mm,
      text width=23mm,
      font=\scriptsize,
      inner sep=1.5pt
  },
  npd/.style={
      rectangle,
      draw=green!55!black,
      rounded corners=1.5pt,
      fill=green!8,
      align=center,
      minimum height=11mm,
      text width=28mm,
      font=\scriptsize,
      inner sep=1.5pt
  },
  metric/.style={
      rectangle,
      draw=black!60,
      rounded corners=1.5pt,
      fill=gray!8,
      align=center,
      minimum height=10mm,
      text width=25mm,
      font=\scriptsize,
      inner sep=1.5pt
  },
  ratebox/.style={
      rectangle,
      draw=black!60,
      rounded corners=1.5pt,
      fill=gray!8,
      align=center,
      minimum height=12mm,
      text width=34mm,
      font=\scriptsize,
      inner sep=1.8pt
  },
  stepbox/.style={
      rectangle,
      draw=orange!80!black,
      rounded corners=1.5pt,
      fill=orange!10,
      align=center,
      minimum height=12mm,
      text width=30mm,
      font=\scriptsize,
      inner sep=1.7pt
  },
  arrow/.style={
      -{Stealth[length=1.6mm,width=1.1mm]},
      line width=0.48pt
  },
  feedback/.style={
      -{Stealth[length=1.6mm,width=1.1mm]},
      dashed,
      line width=0.48pt
  },
  trainarrow/.style={
      -{Stealth[length=1.6mm,width=1.1mm]},
      dashed,
      line width=0.48pt,
      green!55!black
  },
  ctx/.style={
      font=\scriptsize,
      fill=white,
      inner sep=0.8pt
  }
]


\node[data] (x) at (0,0)
    {Source\\$X^N$};

\node[process] (tc) at (3.0,0)
    {Neural test channel\\
     $P^{\phi}_{\hat X^N|X^N}$};

\node[data] (xh) at (6.0,0)
    {Sample\\$\hat X^N$};

\node[metric, text width=30mm] (dist) at (10.45,0)
    {Samplewise distortion\\[0.5mm]
     $\displaystyle d_N(x_j^N,\hat x_j^N)$};

\draw[arrow] (x.east)--(tc.west);
\draw[arrow] (tc.east)--(xh.west);
\draw[arrow] (xh.east)--(dist.west);


\node[process,text width=30mm] (polar) at (6.0,-1.45)
    {Polar transform $U^N=\hat X^N G_N$};

\draw[arrow] (xh.south)--(polar.north);


\node[npd] (npdx) at (4.45,-3.15)
    {Conditional NPD\\[0.7mm]
     $N^{-1}\widehat{\mathsf H}_{\theta^{\mathrm{cond}}}^{\phi}
      (U^N|X^N)$};

\node[npd] (npdc) at (7.55,-3.15)
    {Unconditional NPD\\[0.7mm]
     $N^{-1}\widehat{\mathsf H}_{\theta^{\mathrm{const}}}^{\phi}
      (U^N)$};

\begin{scope}[on background layer]
\node[
    draw=green!50!black,
    dashed,
    rounded corners=2pt,
    fit=(npdx)(npdc),
    inner sep=1.8mm
] (npdgroup) {};
\end{scope}

\coordinate (usplit) at ($(polar.south)+(0,-2.6mm)$);
\draw[arrow] (polar.south)--(usplit);
\draw[arrow] (usplit)-|(npdx.north);
\draw[arrow] (usplit)-|(npdc.north);

\coordinate (ctxY) at (0,-3.15);
\draw[arrow]
    (x.south)--(ctxY)
    -- node[ctx,pos=0.55,above,yshift=1mm] {$X^N$ context}
       (npdx.west);


\node[ratebox] (rate) at (6.0,-5.15)
    {NPD information-density estimate\\[0.8mm]
     $\displaystyle \hat\imath_\theta(x_j^N,\hat x_j^N)$};

\draw[arrow]
    (npdx.south)--(npdx.south|-rate.north);

\draw[arrow]
    (npdc.south)--(npdc.south|-rate.north);


\node[stepbox] (train) at (1.65,-5.15)
    {NPD estimation step\\[0.45mm]
     $\displaystyle
      \min_{\theta}\;
      \mathcal L_{\mathrm{const}}
      +\mathcal L_{\mathrm{cond}}$\\[-0.15mm]
     update $\theta$ with $\phi$ fixed};

\coordinate (trainlevel) at ($(npdgroup.south)+(0,-3mm)$);
\coordinate (traintarget) at ($(npdgroup.south)+(-12mm,0)$);

\draw[trainarrow]
    (train.north)
    -- (train.north|-trainlevel)
    -- (traintarget|-trainlevel)
    -- (traintarget);


\node[stepbox] (phistep) at (10.45,-5.15)
    {Score-function test-channel update\\[0.45mm]
     $\displaystyle
      \phi\leftarrow\phi-\gamma
      \widehat{\nabla_\phi J_\lambda}$\\[-0.15mm]
     using \eqref{eq:rd_empirical_score_gradient} with $\theta$ fixed};

\draw[arrow] (rate.east)--(phistep.west);
\draw[arrow] (dist.south)--(phistep.north);

\coordinate (fbR) at ($(phistep.east)+(8mm,0)$);
\coordinate (fbTop) at ($(tc.north)+(0,8mm)$);

\draw[feedback]
    (phistep.east)
    -- (fbR)
    -- (fbR|-fbTop)
    -- (tc.north|-fbTop)
    -- (tc.north);

\end{tikzpicture}}
    \caption{Alternating NPD-based optimization of the rate--distortion
    objective. For the current neural test channel
    $P_{\hat X^N|X^N}^{\phi}$, each source block generates a
    reconstruction sample, its samplewise distortion, and the polarized sequence $U^N=\hat X^NG_N$. The conditional and
    unconditional NPDs provide the plug-in information-density estimate
    $\hat\imath_\theta(x_j^N,\hat x_j^N)$, which enters the score-function gradient
    \eqref{eq:rd_empirical_score_gradient} together with the distortion term. The NPD parameters $\theta$ are updated with $\phi$ fixed (dashed green arrow), and the updated test-channel parameters $\phi$ are fed back to the test channel (dashed black arrow).}
    \label{fig:rd_framework_diagram}
\end{figure}
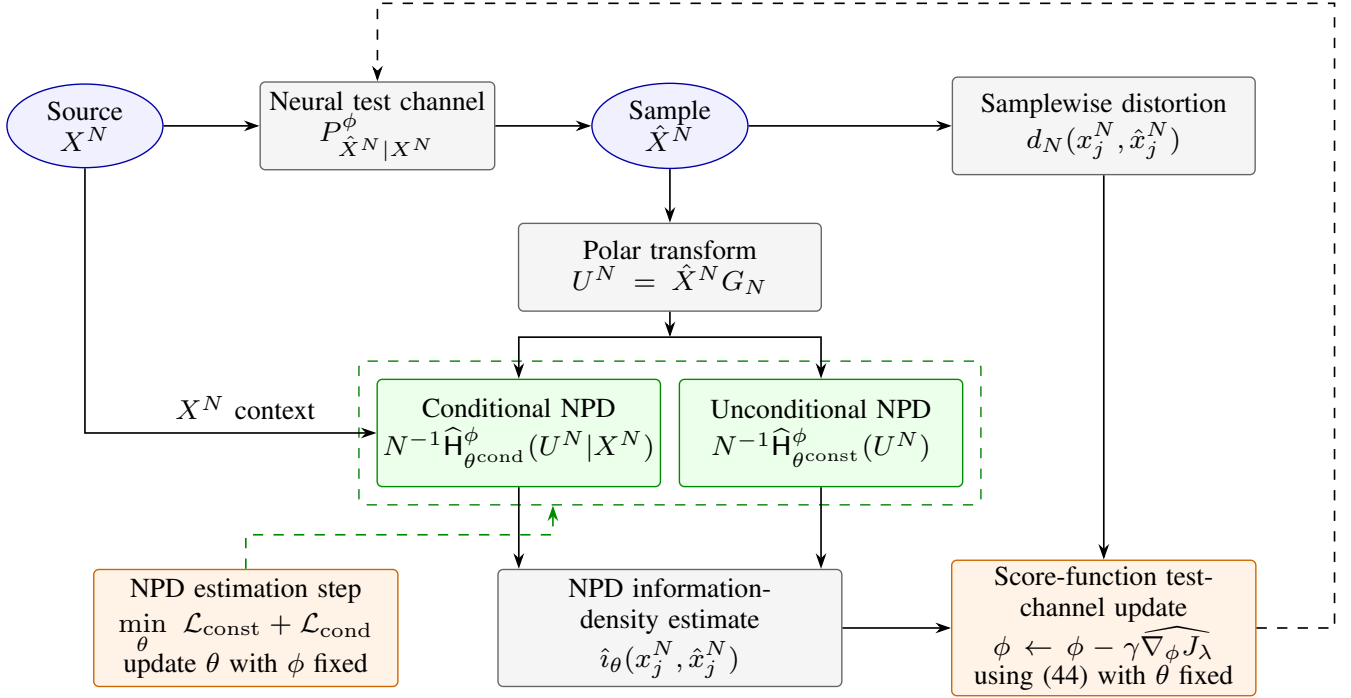

\begin{algorithm}[]
    \caption{Rate--distortion via NPDs}
    \label{alg:RD}
    \textbf{input:}
    Block length $N$, dataset size $M$,
    \# of iterations $\mathsf{N_{iters}}$, learning rates
    $\gamma_\theta,\gamma_\phi$, Lagrange multiplier $\lambda$\\
    \textbf{output:} Optimized $\theta^\ast,\phi^\ast$
    \algrule
    \begin{algorithmic}
    \State Initialize the weights of
    $\theta=\{\theta^{\mathrm{const}},\theta^{\mathrm{cond}}\}$ and $\phi$
    \State Generate
    $\mathcal D_{M,N}^{\phi}\sim P_{X^N,\hat X^N}^{\phi}$
    \State \underline{\textbf{Warm-up:}} repeat the estimation step below with $\phi$ fixed

    \For{$\mathsf{N_{iters}}$ iterations}

        \State \underline{\textbf{Estimation step:}}
        \State Sample $(x^N,\hat x^N)\sim \mathcal D_{M,N}^{\phi}$
        \State Compute $u^N=\hat x^N G_N$
        \State Compute
        $\mathbf e_0^{\mathrm{const}}=E_{\theta_{E_1}}(0^N)$
        and
        $\mathbf e_0^{\mathrm{cond}}=E_{\theta_{E_2}}(x^N)$
        \State Compute
        $\mathcal L_{\mathrm{const}}(\hat x^N)$ using
        \eqref{eq:loss_const_mi}
        and
        $\mathcal L_{\mathrm{cond}}(x^N,\hat x^N)$ using
        \eqref{eq:loss_cond_mi}
        \State Update
        \[
            \theta
            :=
            \theta
            -
            \gamma_\theta
            \nabla_\theta
            \left[
                \mathcal L_{\mathrm{const}}(\hat x^N)
                +
                \mathcal L_{\mathrm{cond}}(x^N,\hat x^N)
            \right]
        \]

        \State \underline{\textbf{Optimization step:}}
        \State Sample $(x^N,\hat x^N)\sim \mathcal D_{M,N}^{\phi}$
        \State Compute $u^N=\hat x^N G_N$
        \State Compute
        $\hat\imath_{\theta}(x^N,\hat x^N)$ using
        \eqref{eq:rd_npd_information_density}
        \State Compute
        $\widehat{\nabla_\phi J_\lambda}
        (\theta,\phi;\mathcal D_{M,N}^{\phi})$ using
        \eqref{eq:rd_empirical_score_gradient}
        \State Update
        \[
            \phi
            :=
            \phi
            -
            \gamma_\phi
            \widehat{\nabla_\phi J_\lambda}
            (\theta,\phi;\mathcal D_{M,N}^{\phi})
        \]
        \State Generate
        $\mathcal D_{M,N}^{\phi}\sim P_{X^N,\hat X^N}^{\phi}$

    \EndFor \\

    \Return $\theta^\ast,\phi^\ast$
    \end{algorithmic}
\end{algorithm}

\subsubsection*{Lossy Polar Code Construction}
\label{sec:RD_construction}

After the alternating procedure converges, let $\phi^\ast$ denote the
optimized test-channel parameters and let
$\theta^\ast=\{\theta^{\mathrm{const},\ast},
\theta^{\mathrm{cond},\ast}\}$ denote the final trained NPD parameters.
These models are converted into an operational finite-length lossy source
code. The construction follows the Honda--Yamamoto polar source-coding
principle~\cite{honda2013polar}, with the model-based synthetic distributions
$P_{U_i|U^{i-1}}$ and $P_{U_i|U^{i-1},X^N}$ of the polar sequence $U^N=\hat X^NG_N$ replaced by their
learned counterparts
$P_{U_i|U^{i-1}}^{\theta^{\mathrm{const},\ast}}$ and
$P_{U_i|U^{i-1},X^N}^{\theta^{\mathrm{cond},\ast}}$.

The index sets are chosen using a design set generated by the optimized
test channel. We draw
$(X^N,\hat X^N)\sim P_{X^N}P_{\hat X^N|X^N}^{\phi^\ast}$, form
$U^N=\hat X^NG_N$, and evaluate the NPD probabilities under teacher
forcing using the known prefixes $u^{i-1}$ of each sampled block.
Since the entire block $u^N$ is available, these evaluations require no
sequential SC decisions, and all
nodes at each depth of the SC computation graph can be evaluated in parallel before proceeding to the next
depth~\cite{aharoni2026optimized}. Averaging the
coordinatewise negative log-probabilities over the design set gives the
empirical entropy estimates
\begin{equation}
    \widehat h_i
    \triangleq
    \widehat{\mathsf H}_{\theta^{\mathrm{const},\ast}}^{\phi^\ast}
    (U_i|U^{i-1}),
    \qquad
    \widehat h_i^X
    \triangleq
    \widehat{\mathsf H}_{\theta^{\mathrm{cond},\ast}}^{\phi^\ast}
    (U_i|U^{i-1},X^N),
    \qquad i=1,\ldots,N .
\end{equation}
Let $q=|\mathcal X|$ and choose a common threshold parameter
$0<\delta<1/2$. Define the nested sets
\[
\begin{aligned}
    \mathcal F
    &\triangleq \{i\in[N]:\widehat h_i>\delta\log q\},\\
    \mathcal F_x
    &\triangleq \mathcal F\cap
    \{i\in[N]:\widehat h_i^X\ge(1-\delta)\log q\}.
\end{aligned}
\]
The intersection enforces $\mathcal F_x\subseteq\mathcal F$ even when
estimation errors violate the ordering of the true conditional and
unconditional entropies. The transmitted coordinates are $\mathcal A=\mathcal F\setminus\mathcal F_x$. Thus, $\mathcal F_x$, $\mathcal A$, and $[N]\setminus\mathcal F$
partition all coordinates into shared frozen symbols, transmitted symbols,
and symbols reconstructed by the decoder, respectively. In particular, intermediate-entropy coordinates,
which are neither frozen nor classified as low entropy at the decoder, are
transmitted.

Once the sets are fixed, the encoder samples the polarized symbols using
the conditional NPD rather than sampling directly from the neural test
channel. Given $x^N$, it runs the learned SC recursion with
the source-dependent embedding $E_{\theta_{E_2}^\ast}(x^N)$ and the learned
conditional distributions
$P_{U_i|U^{i-1},X^N}^{\theta^{\mathrm{cond},\ast}}
(\cdot|\hat u^{i-1},x^N)$, where $\hat u^{i-1}$ contains the symbols
selected in the preceding SC steps. For $i\in\mathcal F_x$, the encoder
sets $\hat u_i$ to a shared symbol drawn independently and uniformly over
$\mathbb{GF}(q)$, independently of the source. For every
$i\notin\mathcal F_x$, it samples $\hat u_i$ from the learned conditional
distribution above. It transmits the sampled symbols on $\mathcal A$ and
leaves those outside $\mathcal F$ untransmitted. The message is therefore
$\hat u_{\mathcal A}$.

The decoder performs SC or successive-cancellation list (SCL) decoding,
using the received symbols $\hat u_{\mathcal A}$ on $\mathcal A$ and the
regenerated shared symbols on $\mathcal F_x$ as fixed coordinates. It
estimates only the symbols in
$[N]\setminus\mathcal F$ using the learned unconditional distributions
$P_{U_i|U^{i-1}}^{\theta^{\mathrm{const},\ast}}(\cdot|\hat u^{i-1})$.
Finally, it applies the inverse polar transform to obtain
$\hat x^N=\hat u^NG_N^{-1}$.
\begin{figure}[t]
    \centering

    \begin{subfigure}{0.44\textwidth}
        \centering
        \resizebox{\textwidth}{!}{

\pgfmathdeclarefunction{hbin}{1}{%
  \pgfmathparse{(#1<=0 || #1>=1) ? 0 : -(#1*ln(#1) + (1-#1)*ln(1-#1))/ln(2)}%
}

\begin{tikzpicture}
\begin{axis}[
    width=3in,
    height=2.6in,
    grid=both,
    xlabel={Rate},
    ylabel={Distortion},
    xmin=0, xmax=1,
    ymin=0, ymax=0.5,
    title={Binary symmetric source},
    legend style={
        at={(0.98,0.98)},
        anchor=north east,
        nodes={scale=0.85, transform shape}
    }
]

\addplot[
    very thick,
    black,
    domain=0.0001:0.4999,
    samples=1000
]
({1 - hbin(x)}, {x});
\addlegendentry{Rate--Distortion}

\addplot[
    only marks,
    mark=star,
    mark size=4.0,             
    line width=1.0,            
    mark options={fill=yellow!90!white, draw=black}
]
coordinates {(0.7115, 0.051)};
\addlegendentry{Estimated pairs}

\addplot[
    only marks,
    mark=star,
    mark size=4.0,
    line width=1.0,
    mark options={fill=yellow!90!white, draw=black},
    forget plot
]
coordinates {
    (0.5330, 0.103)
    (0.3852, 0.153)
    (0.2791, 0.205)
    (0.18717, 0.254)
    (0.0662, 0.353)
};


\addplot[
    only marks,
    mark=o,
    mark size=1.2,   
    line width=0.8,
    draw=red,
    mark options={fill=white}
]
coordinates {
    (0.7136, 0.0738)
    (0.5310, 0.1238)
    (0.3902, 0.1738)
    (0.2781, 0.2238)
    (0.1887, 0.2738)
    (0.0659, 0.3738)
};
\addlegendentry{Block length $n = 10$}

\addplot[
    only marks,
    mark=o,
    mark size=1.2,
    line width=0.8,
    draw=green,
    mark options={fill=white}
]
coordinates {
    (0.7136, 0.0675)
    (0.5310, 0.1175)
    (0.3902, 0.1675)
    (0.2781, 0.2175)
    (0.1887, 0.2675)
    (0.0659, 0.3675)
};
\addlegendentry{Block length $n = 11$}

\addplot[
    only marks,
    mark=o,
    mark size=1.2,
    line width=0.8,
    draw=blue,
    mark options={fill=white}
]
coordinates {
    (0.7136, 0.0652)
    (0.5310, 0.1152)
    (0.3902, 0.1652)
    (0.2781, 0.2152)
    (0.1887, 0.2652)
    (0.0659, 0.3652)
};
\addlegendentry{Block length $n = 12$}

\addplot[
    only marks,
    mark=o,
    mark size=1.2,
    line width=0.8,
    draw=magenta,
    mark options={fill=white}
]
coordinates {
    (0.7136, 0.0614)
    (0.5310, 0.1114)
    (0.3902, 0.1614)
    (0.2781, 0.2114)
    (0.1887, 0.2614)
    (0.0659, 0.3614)
};
\addlegendentry{Block length $n = 13$}

\end{axis}
\end{tikzpicture}}
        \label{fig:RD_BSS}
    \end{subfigure}
    \hspace{0.015\textwidth}
    \begin{subfigure}{0.42\textwidth}
        \centering
        \resizebox{\textwidth}{!}{

\begin{tikzpicture}
\begin{axis}[
    width=3in,
    height=2.6in,
    grid=both,
    xlabel={Rate},
    ylabel={Distortion},
    xlabel style={font=\normalsize},
    ylabel style={font=\normalsize},
    xmin=0, xmax=0.85,
    ymin=0, ymax=0.5,
    title={First-order Markov source},
    title style={font=\normalsize},
    tick label style={font=\small},
    legend style={
        at={(0.98,0.98)},
        anchor=north east,
        nodes={scale=0.8, transform shape}
    }
]

\addplot[
    thick,
    blue,
    mark=o,
    mark size=1.8,
    line width=1.0,
    mark options={fill=blue!30!white, draw=blue}
]
coordinates {
    (0.0000000, 0.5000000)
    (0.00001676302, 0.4840919)
    (0.0001899161, 0.4607729)
    (0.001041784, 0.4485878)
    (0.007805088, 0.4081276)
    (0.01891936, 0.3698409)
    (0.03615485, 0.3305780)
    (0.05935751, 0.2923670)
    (0.08777879, 0.2573168)
    (0.1207477, 0.2256095)
    (0.1585711, 0.1964401)
    (0.1976844, 0.1717921)
    (0.2414934, 0.1487717)
    (0.2947852, 0.1250935)
    (0.3519401, 0.1032953)
    (0.4117202, 0.08352986)
    (0.4700492, 0.06664051)
    (0.5276350, 0.05194228)
    (0.5906104, 0.03763278)
    (0.6453422, 0.02650760)
    (0.6911638, 0.01811840)
    (0.7274454, 0.01210727)
    (0.7552706, 0.007915685)
    (0.7759869, 0.005066018)
    (0.7909818, 0.003175081)
};
\addlegendentry{Blahut-Arimoto}

\addplot[
    only marks,
    mark=star,
    mark size=3.8,
    line width=1.0,
    mark options={fill=red, draw=black}
]
coordinates {
    (0.162, 0.2172)
    (0.28241, 0.146)
    (0.0222, 0.391)
    (0.4620, 0.079)
    (0.587771, 0.04211)
};
\addlegendentry{Estimated Pairs}

\addplot[
    only marks,
    mark=o,
    mark size=2.2,
    line width=1.0,
    mark options={fill=green!70!black, draw=black}
]
coordinates {
    (0.039, 0.427)
    (0.1870, 0.2519)
    (0.3237, 0.1740)
    (0.66328, 0.051511)
    (0.5194, 0.095)
};
\addlegendentry{Operating Pairs}

\end{axis}
\end{tikzpicture}}
        \label{fig:RD_Markov}
    \end{subfigure}

\caption{Estimated rate--distortion pairs and corresponding $(R,D)$ operating points of the learned lossy coding scheme. Star markers denote estimated pairs, and circular markers denote achieved performance. In both figures, the NPD and neural test channel are trained using block length $n=10$. The left figure corresponds to a binary symmetric source, while the right figure corresponds to a first-order Markov source. The operating points in the latter are obtained using block length $n=13$ and list size $L=16$ for decoding.}

    \label{fig:RD_combined}
\end{figure}
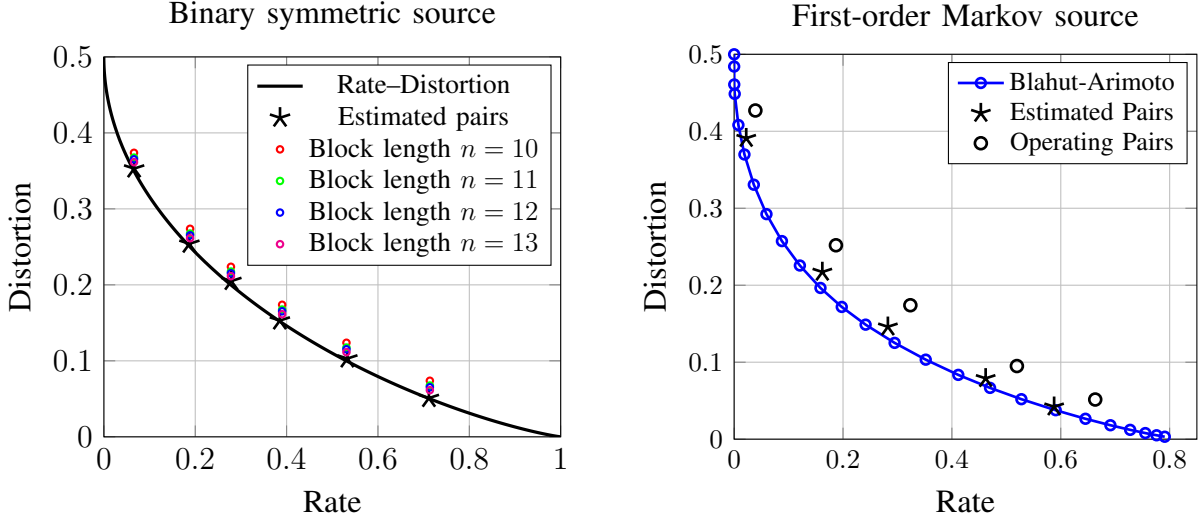

\begin{figure}[t]
    \centering
    \resizebox{0.52\textwidth}{!}{\pgfmathdeclarefunction{hbinecoli}{1}{%
  \pgfmathparse{(#1<=0 || #1>=1) ? 0 : -(#1*ln(#1) + (1-#1)*ln(1-#1))/ln(2)}%
}

\begin{tikzpicture}
\begin{axis}[
    width=3.4in,
    height=2.7in,
    grid=both,
    xlabel={Rate (bits per base)},
    ylabel={Distortion},
    xmin=0, xmax=2.0,
    ymin=0, ymax=0.65,
    title={\texttt{E.coli} source},
    title style={font=\normalsize},
    tick label style={font=\small},
    legend style={
        at={(0.97,0.97)},
        anchor=north east,
        nodes={scale=0.72, transform shape}
    }
]

\addplot[
    black,
    dashed,
    thick,
    domain=0.0001:0.65,
    samples=500
]
({2-hbinecoli(x)-x*ln(3)/ln(2)}, {x});
\addlegendentry{Uniform i.i.d. $q=4$}

\addplot[
    blue,
    thick,
    smooth,
    mark=star,
    mark size=3.8,
    line width=1.0,
    mark options={fill=yellow!90!white, draw=blue}
]
coordinates {
    (0.072, 0.566)
    (0.250, 0.445)
    (0.6753, 0.266)
    (0.988, 0.175)
    (1.4758, 0.063)
    (1.885, 0.0001)
};
\addlegendentry{Estimated pairs}

\end{axis}
\end{tikzpicture}}
    \caption{Estimated rate--distortion tradeoff for the \texttt{E. coli} source under symbolwise Hamming distortion. Blue stars denote the estimated $(R,D)$ pairs obtained after convergence of the NPDs and neural test channel for different values of the Lagrange multiplier. The dashed black curve is the rate--distortion function of a uniform memoryless quaternary source.}
    \label{fig:RD_EColi}
\end{figure}
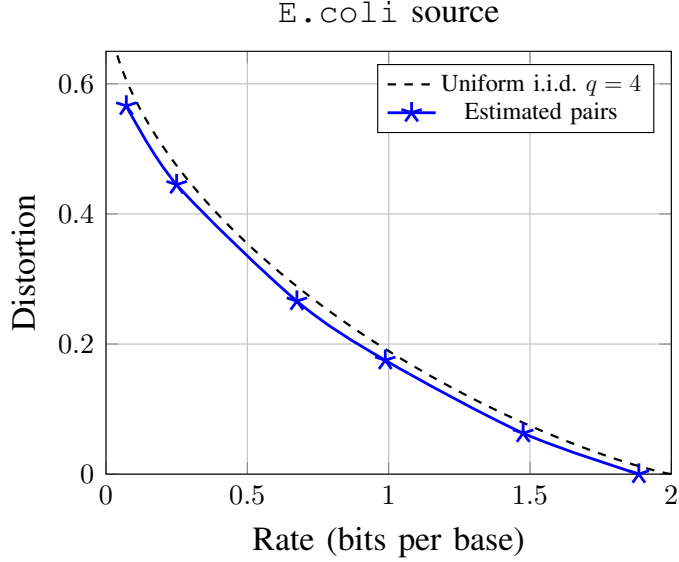

\subsection{Experiments}
We evaluate the lossy source-coding performance of the proposed rate--distortion estimation and optimization scheme. The architecture, training parameters, and alternating optimization schedule used for the experiments are given in Table~\ref{tab:lossy_reproducibility} in Appendix~\ref{app:implementation_details}. Figure~\ref{fig:RD_combined}
(left) shows the results for a binary symmetric source under Hamming
distortion. The star markers correspond to the rate--distortion estimates
obtained while training the NPDs and optimizing the neural test channel, as described in
Algorithm~\ref{alg:RD}. These estimates closely follow the analytical
rate--distortion curve, demonstrating the accuracy of the proposed
estimation--optimization procedure. The remaining markers represent the
$(R,D)$ operating points achieved by the resulting lossy polar coding
scheme for different block lengths. As the block length increases, these
operating points approach the estimated rate--distortion curve.

Figure~\ref{fig:RD_combined} (right) presents the corresponding results for a
non-symmetric binary first-order Markov source with transition matrix
\[
P_{X_t|X_{t-1}}
=
\begin{bmatrix}
0.3 & 0.7\\
0.8 & 0.2
\end{bmatrix}.
\]
Since the rate--distortion function of this source is not available in
closed form, we approximate it numerically using a
Blahut--Arimoto-type algorithm applied to source blocks of length
$N=2^{11}$ \cite{naiss2012computable}. The star markers show the rate--distortion
estimates produced by the proposed procedure in Algorithm~\ref{alg:RD}.
These estimates lie close to, and slightly above, the numerical Blahut--Arimoto curve,
indicating that the learned test channel and the NPD-based mutual
information estimator capture the multi-letter rate--distortion behavior
of the Markov source. The circular markers correspond to the $(R,D)$
operating points achieved by the Honda--Yamamoto construction using the
learned NPD and neural test channel. As expected for finite-block lossy
coding, these operating points lie above the estimated
rate--distortion curve.

Finally, we evaluate the proposed rate--distortion estimation procedure on the real \texttt{E. coli} genomic sequence used in the lossless experiments. The source is represented over the four-symbol alphabet using the same preprocessing and chronological training--validation split described in Appendix~\ref{app:implementation_details}. We train the NPDs and neural test channel under symbolwise Hamming distortion for several values of the Lagrange multiplier. This experiment focuses on evaluating the learned rate--distortion estimator on real genomic data. The estimated $(R,D)$ pairs are shown in Fig.~\ref{fig:RD_EColi}. At near-zero distortion, the estimated rate closely matches the lossless compression rate of $1.883$ bpb reported in Section~\ref{sec: methodology}. For comparison, Fig.~\ref{fig:RD_EColi} includes the rate--distortion function of a uniform memoryless quaternary source under Hamming distortion,
\[
R_{\mathrm{iid}}(D)=
\begin{cases}
    2-h_2(D)-D\log 3, & 0\le D\le \frac{3}{4},\\
    0, & D>\frac{3}{4},
\end{cases}
\]
where $h_2(\cdot)$ denotes the binary entropy function.
This provides a baseline that ignores temporal dependence in the genomic sequence.

\section{Wyner--Ziv via NPD}\label{sec:WZ}
In this section, we extend the NPD-based rate--distortion estimation optimization framework of Section~\ref{sec: RD} to the Wyner--Ziv setting, where side information is available only at the decoder. Unlike the ordinary rate--distortion problem, the Wyner--Ziv formulation involves an auxiliary random variable, the rate is characterized by a conditional mutual information, and the reconstruction is formed at the decoder from the auxiliary description and the side information. Consequently, the learning procedure must optimize both a test channel that generates the auxiliary representation and a reconstruction mapping, while the NPDs estimate the synthetic conditional entropy terms that define the conditional mutual information. The scheme applies to stationary ergodic source pairs and is based on the multi-letter Wyner--Ziv rate--distortion functional, whereas its polar-coding guarantees are established for the FAIM setting under the assumptions of Lemma~\ref{lem:wz_polarization} and Theorem~\ref{thm:wz_faim_achievability}. Throughout this section, we define explicitly the components that differ from Section~\ref{sec: RD} and refer to the corresponding definitions there for the common structure.

\subsection{Problem Formulation}
\label{sec:WZ_problem}

Let $\{(X_i,Y_i)\}_{i\ge 1}$ be a jointly stationary ergodic source pair with
finite alphabets $\mathcal X$ and $\mathcal Y$. The sequence $X^N$ is the
source to be compressed and $Y^N$ denotes correlated side information
available only at the decoder. An $(N,|\mathcal M_N|)$ Wyner--Ziv code
consists of an encoder $f_N:\mathcal X^N\to\mathcal M_N$ that maps
$X^N$ to a message $M=f_N(X^N)$, and a decoder
$\psi_N:\mathcal M_N\times\mathcal Y^N\to\mathcal X^N$ that
produces the reconstruction $\hat X^N=\psi_N(M,Y^N)$. We use the bounded
distortion measure $d$ and normalized additive block distortion $d_N$
from~\eqref{eq:normalized_block_distortion}.
Let $R_{\mathrm{WZ}}(D)$ denote the infimum of asymptotically achievable
rates under an expected per-symbol distortion constraint of $D$.

For i.i.d. source pairs, the classical Wyner--Ziv theorem gives a
single-letter auxiliary-variable characterization~\cite{wyner1976rate}.
Kieffer's sliding-block coding framework extends Wyner--Ziv coding to
jointly stationary ergodic finite-alphabet processes~\cite{kieffer1981method}.
For each block length $N$, define
\begin{equation}
    \mathcal R_{\mathrm{WZ}}^{(N)}(D)
    \triangleq
    \frac{1}{N}
    \inf_{P_{U^N|X^N},\,g_N}
    \mathsf{I}(X^N;U^N|Y^N),
    \label{eq:wz_multiletter}
\end{equation}
where the infimum ranges over all finite block auxiliary alphabets
$\mathcal U_N$, all conditional distributions $P_{U^N|X^N}$, and all
reconstruction mappings
\[
    g_N:\mathcal U_N\times\mathcal Y^N\to\mathcal X^N
\]
such that the Markov condition
\begin{equation}
    U^N - X^N - Y^N
    \label{eq:wz_block_markov}
\end{equation}
holds and
\begin{equation}
    \mathbb E
    \left[
        d_N
        \left(
            X^N,
            g_N(U^N,Y^N)
        \right)
    \right]
    \le D.
    \label{eq:wz_multiletter_distortion}
\end{equation}
The operational Wyner--Ziv rate--distortion function
admits the multiletter characterization~\cite{kieffer1981method,verdu2008information}
\begin{equation}
    R_{\mathrm{WZ}}(D)
    =
    \lim_{N\to\infty}
    \mathcal R_{\mathrm{WZ}}^{(N)}(D).
    \label{eq:wz_limit}
\end{equation}
The neural formulation developed below restricts this finite-block
optimization to a prescribed auxiliary alphabet and parameterized
test-channel and reconstruction families.

\subsection{Conditional Polarization and Nested Index Sets}
\label{sec:WZ_polarization}

Before introducing the NPD estimation and coding scheme, we first
identify the conditional polarization and nesting properties that determine which
auxiliary coordinates must be communicated in the Wyner--Ziv setting.
For the polar construction analyzed below, let $q$ be prime and consider a
stationary auxiliary process $\{U_t\}_{t\in\mathbb Z}$ over
$\mathbb{GF}(q)$ such
that the joint process
$\{(U_t,X_t,Y_t)\}_{t\in\mathbb Z}$ satisfies, for every $N$, the blockwise
Markov relation
\begin{equation}
    U^N-X^N-Y^N.
    \label{eq:wz_markov}
\end{equation}
For $N=2^n$, let $\tilde U^N=U^NG_N$, where $G_N$ is generated by the
standard $q$-ary Ar\i kan kernel with operations over $\mathbb{GF}(q)$. For
each of the two choices $Z=X$ and $Z=Y$, define
\begin{equation}
    h_{i,N}^{Z}
    \triangleq
    \frac{\mathsf H(\tilde U_i|\tilde U^{i-1},Z^N)}{\log q},
    \qquad i=1,\ldots,N.
    \label{eq:wz_conditional_synthetic_entropy}
\end{equation}
Since the polar transform is bijective and the Markov relation implies
$\mathsf H(U^N|X^N,Y^N)=\mathsf H(U^N|X^N)$, we have
\begin{align}
    \mathsf I(X^N;U^N|Y^N)
    &=
    \mathsf H(\tilde U^N|Y^N)
    -\mathsf H(\tilde U^N|X^N)
    \nonumber\\
    &=
    (\log q)\sum_{i=1}^{N}
    \left(h_{i,N}^{Y}-h_{i,N}^{X}\right).
    \label{eq:wz_synthetic_decomposition}
\end{align}
To translate the decomposition in~\eqref{eq:wz_synthetic_decomposition}
into a polar rate set, we apply the conditional-polarization result for
FAIM source--observation pairs to each of the processes
$\{(U_t,X_t)\}_{t\in\mathbb Z}$ and $\{(U_t,Y_t)\}_{t\in\mathbb Z}$.
This result, including its prime-alphabet extension, is established
in~\cite[Cor.~8, Th.~13, App.~B]{shuval2019fast}. The lemma below combines
the resulting conditional polarization with the block Markov relation to
establish nested high-entropy sets whose asymptotic size difference equals
the conditional mutual-information rate.

For $\beta\in(0,1/2)$, let $\delta_N\triangleq2^{-N^\beta}$. For
each of the two choices $Z=X$ and $Z=Y$, define
\begin{equation}
    \mathcal H_{U|Z}^{(N)}
    \triangleq
    \{i\in[N]:h_{i,N}^{Z}\ge 1-\delta_N\}.
    \label{eq:wz_high_entropy_set}
\end{equation}

\begin{lemma}[Fast conditional polarization and nested rate set]
\label{lem:wz_polarization}
Let $q$ be prime, and let $\{(U_t,X_t,Y_t)\}_{t\in\mathbb Z}$ be a
stationary finite-alphabet process with $U_t\in\mathbb{GF}(q)$. Suppose
that the joint processes $\{(U_t,X_t)\}_{t\in\mathbb Z}$ and
$\{(U_t,Y_t)\}_{t\in\mathbb Z}$ each admit a FAIM representation,
and that $U^N-X^N-Y^N$ forms a Markov
chain for every $N$. For $N=2^n$, every
$\beta\in(0,1/2)$, and both $Z=X$ and $Z=Y$, we have
\begin{equation}
    \lim_{N\to\infty}
    \frac{|\mathcal H_{U|Z}^{(N)}|}{N}\log q
    =\overline{\mathsf H}(U|Z),
    \label{eq:wz_cor_high_limit}
\end{equation}
where
$\overline{\mathsf H}(U|Z)
\triangleq\lim_{N\to\infty}N^{-1}\mathsf H(U^N|Z^N)$. Furthermore, for
every $i\in[N]$,
\begin{align}
    h_{i,N}^{X} &\le h_{i,N}^{Y},
    \label{eq:wz_cor_entropy_ordering}
    \\
    \mathcal H_{U|X}^{(N)}
    &\subseteq
    \mathcal H_{U|Y}^{(N)}.
    \label{eq:wz_cor_set_nesting}
\end{align}
Hence, the rate-carrying set
$\mathcal A_N\triangleq\mathcal H_{U|Y}^{(N)}\setminus\mathcal H_{U|X}^{(N)}$
satisfies
\begin{equation}
    \lim_{N\to\infty}
    \frac{|\mathcal A_N|}{N}\log q
    =
    \overline{\mathsf I}(X;U|Y),
    \label{eq:wz_cor_rate_limit}
\end{equation}
where
$\overline{\mathsf I}(X;U|Y)
\triangleq\lim_{N\to\infty}N^{-1}\mathsf I(X^N;U^N|Y^N)$.
\end{lemma}

The proof of Lemma~\ref{lem:wz_polarization} is given in
Appendix~\ref{app:wz_polarization}. It applies conditional FAIM polarization
separately with observations $X^N$ and $Y^N$. The block
Markov relation orders the coordinatewise conditional entropies and produces
nested fast high-entropy sets, whose asymptotic size difference equals the
conditional mutual-information rate.

The following theorem converts this nested polarization structure into an
operational coding statement for a prescribed auxiliary process and
reconstruction sequence.

\begin{theorem}[Nested polar coding for Wyner--Ziv compression]
\label{thm:wz_faim_achievability}
Let $q$ be prime and $\{(U_t,X_t,Y_t)\}_{t\in\mathbb Z}$ be a stationary
finite-alphabet process with $U_t\in\mathbb{GF}(q)$. Assume that
$\{(U_t,X_t)\}$ and $\{(U_t,Y_t)\}$ each admit a FAIM representation
and that $U^N-X^N-Y^N$ is a Markov chain for every $N$.
Assume that the single-letter distortion satisfies
$0\le d(x,\hat x)\le d_{\max}<\infty$ for all
$x,\hat x\in\mathcal X$.
For block lengths $N=2^n$ and $\beta\in(0,1/2)$, let
$\delta_N\triangleq 2^{-N^\beta}$ and define
\begin{align}
    \mathcal F_N
    &\triangleq
    \{i\in[N]:h_{i,N}^{X}\ge 1-\delta_N\},
    \label{eq:wz_ach_frozen_set}\\
    \mathcal L_N
    &\triangleq
    \{i\in[N]:h_{i,N}^{Y}\le \delta_N\},
    \label{eq:wz_ach_low_set}\\
    \mathcal I_N
    &\triangleq
    [N]\setminus(\mathcal F_N\cup\mathcal L_N).
    \label{eq:wz_ach_information_set}
\end{align}
Let
$g_N:\mathcal U^N\times\mathcal Y^N\to\mathcal X^N$
be a sequence of deterministic reconstruction mappings. If
\begin{equation}
    \limsup_{N\to\infty}
    \mathbb E_{P_{X^NY^NU^N}}
    \left[d_N\bigl(X^N,g_N(U^N,Y^N)\bigr)\right]
    \le D,
    \label{eq:wz_target_distortion}
\end{equation}
then there exists a sequence of deterministic nested polar codes with rate
$R_N=|\mathcal I_N|\log q/N$ bits per source symbol and reconstruction
$\hat X^N$ satisfying
\begin{align}
    \lim_{N\to\infty}R_N
    &=\overline{\mathsf I}(X;U|Y),
    \label{eq:wz_ach_rate}\\
    \limsup_{N\to\infty}
    \mathbb E\left[d_N(X^N,\hat X^N)\right]
    &\le D.
    \label{eq:wz_ach_distortion}
\end{align}
where the expectation in~\eqref{eq:wz_ach_distortion} is under the joint law
induced by the source and the deterministic code.
Moreover, the decoder recovers the encoder-generated auxiliary sequence with
probability tending to one.
\end{theorem}

The proof of Theorem~\ref{thm:wz_faim_achievability} is given in
Appendix~\ref{app:wz_faim_achievability}. It first analyzes a randomized
successive encoder whose induced law approaches the target law in total
variation. SC decoding and fast polarization establish reliability
and rate, the target distortion is transferred to the actual reconstruction,
and a selection argument removes the encoder randomness. Hence, if the prescribed auxiliary
process and reconstruction sequence attain the Wyner--Ziv rate--distortion
boundary, so do the resulting polar codes.

For memoryless sources, the corresponding nested set is used in the
polar Wyner--Ziv construction of~\cite{korada2010itw_wynerziv}, and Theorem~\ref{thm:wz_faim_achievability} gives the corresponding
operational statement for a prescribed auxiliary process in the FAIM case. The
finite-block neural test channel introduced next is not constrained to induce
processes that admit FAIM state representations or to attain the optimal
target law.
Thus, the theorem provides the information-theoretic basis for the FAIM case, while the general learned finite-length construction is
evaluated empirically.

\subsection{Neural Test Channel and Reconstruction Models}
\label{sec:WZ_test_channel}

Whereas the neural test channel of Section~\ref{sec:RD_pmf} directly generates the reconstruction sequence, the Wyner--Ziv formulation requires two learned components: a stochastic test
channel $P_{U^N|X^N}^{\phi}$ that generates the auxiliary description $U^N$ of $X^N$, and a reconstruction mapping that forms $\hat X^N$ at the decoder from
$(U^N,Y^N)$.

For the neural parameterization, we fix a finite auxiliary alphabet
$\mathcal U$ and use the same length-$N$ encoder--decoder transformer as in
Section~\ref{sec:RD_pmf}, with the reconstruction symbol $\hat X_i$ replaced by
the auxiliary symbol $U_i$. The encoder processes the complete source block
$X^N$, and the causally masked decoder generates $U_i$ from $X^N$ and
$U^{i-1}$, which induces the factorization
\begin{equation}
    P_{U^N|X^N}^{\phi}(u^N|x^N)
    =
    \prod_{i=1}^{N}
    P_{U_i|X^N,U^{i-1}}^{\phi}
    (u_i|x^N,u^{i-1}).
    \label{eq:wz_test_channel_factorization}
\end{equation}
Together with the source law $P_{X^N,Y^N}$, this test channel induces the
joint distribution
\[
    P_{X^N,Y^N,U^N}^{\phi}
    =
    P_{X^N,Y^N}
    P_{U^N|X^N}^{\phi},
\]
so that the Markov relation $U^N-X^N-Y^N$ is satisfied by construction.

The reconstruction is modeled separately and is allowed to
depend on both the auxiliary sequence and the side-information sequence available at the decoder.
For deterministic evaluation and coding, we use the reconstruction mapping
\begin{equation}
    \hat X^N
    =
    g_\psi(U^N,Y^N),
    \label{eq:wz_reconstruction_mapping}
\end{equation}
where $\psi$ denotes its parameters. The mapping $g_\psi$ may process
$(U^N,Y^N)$ noncausally across the block, reflecting the multi-letter
nature of the Wyner--Ziv reconstruction rule.

For gradient-based training with discrete reconstructions, the network
also defines a stochastic reconstruction distribution
\[
    P_{\hat X^N|U^N,Y^N}^{\psi}(\hat x^N|u^N,y^N)
    =\prod_{i=1}^{N}P_{\hat X_i|U^N,Y^N}^{\psi}(\hat x_i|u^N,y^N),
\]
where each $P_{\hat X_i|U^N,Y^N}^{\psi}$ is a softmax distribution with differentiable
logits. Reconstructions are sampled from $P_{\hat X^N|U^N,Y^N}^{\psi}$ during training.
For deterministic evaluation and coding, the mapping in
\eqref{eq:wz_reconstruction_mapping} selects
$[g_\psi(u^N,y^N)]_i=\arg\max_{\hat x\in\mathcal X}
P_{\hat X_i|U^N,Y^N}^{\psi}(\hat x|u^N,y^N)$, with ties resolved by a fixed rule.
The test channel $P_{U^N|X^N}^{\phi}$ thus determines the auxiliary
representation and the conditional mutual-information term
$\mathsf{I}_\phi(X^N;U^N|Y^N)$, while the reconstruction network defines
$P_{\hat X^N|U^N,Y^N}^{\psi}$ for training and $g_\psi$ for deterministic evaluation and coding. The two components are optimized jointly as described next.

\subsection{Optimization and Estimation Scheme}
\label{sec:WZ_opt}
The estimation--optimization procedure follows the alternating structure of Section~\ref{sec: RD}.
For a fixed test channel, the NPDs are trained to estimate the relevant information-theoretic objective.
With the NPD parameters fixed, the learned stochastic models are then updated using the resulting plug-in estimate.
There are three differences from the rate--distortion setting.
First, the rate term is the conditional mutual information $\mathsf{I}_\phi(X^N;U^N|Y^N)$.
Accordingly, the unconditional NPD of Section~\ref{sec:RD_estimation} is replaced by an NPD conditioned on the decoder side information $Y^N$.
Second, the training distortion is averaged over the reconstruction distribution
$P_{\hat X^N|U^N,Y^N}^{\psi}$, which uses the decoder side information.
Third, the reconstruction parameters $\psi$ are updated together with $\phi$.
We define these modified components explicitly below.

\subsubsection*{Estimation Stage}

For a fixed test channel $P_{U^N|X^N}^{\phi}$, the induced joint distribution $P_{X^N,Y^N,U^N}^{\phi}$ satisfies $U^N-X^N-Y^N$ by construction, so the Wyner--Ziv rate term can be
decomposed according to \eqref{eq:wz_synthetic_decomposition}. We estimate the two entropy terms in
\eqref{eq:wz_synthetic_decomposition} using
two NPD parameter sets,
\begin{align*}
    \theta^{Y}
    &=
    \{E_{\theta_{E_Y}},F_{\theta_F},G_{\theta_G},H_{\theta_H}\},
    \\
    \theta^{X}
    &=
    \{E_{\theta_{E_X}},F_{\theta_F},G_{\theta_G},H_{\theta_H}\},
\end{align*}
which, as in Section~\ref{sec:RD_estimation}, share the core components and differ in their embedding networks. The parameter set $\theta^Y$ uses the side-information-dependent
embedding
$\mathbf e_0^{Y}=E_{\theta_{E_Y}}(y^N)\in\mathbb R^{N\times d}$
and estimates the distributions
$P_{\tilde U_i|\tilde U^{i-1},Y^N}^{\theta^Y}$, whereas
$\theta^X$ uses the source-dependent embedding
$\mathbf e_0^{X}=E_{\theta_{E_X}}(x^N)\in\mathbb R^{N\times d}$
and estimates the distributions
$P_{\tilde U_i|\tilde U^{i-1},X^N}^{\theta^X}$.

Let $\mathcal D_{M,N}^{\phi}
    =\{(x_j^N,y_j^N,u_j^N)\}_{j=1}^{M}\sim P_{X^N,Y^N,U^N}^{\phi}$. 
    For each triplet, we compute $\tilde u_j^N=u_j^N G_N$. The corresponding NPD losses are
\begin{align}
    \mathcal L_Y(y_j^N,u_j^N)
    &=
    -\frac{1}{N}
    \sum_{i=1}^{N}
    \log
    P_{\tilde U_i|\tilde U^{i-1},Y^N}^{\theta^Y}
    (\tilde u_{j,i}|\tilde u_j^{i-1},y_j^N),
    \label{eq:wz_loss_y}
    \\
    \mathcal L_X(x_j^N,u_j^N)
    &=
    -\frac{1}{N}
    \sum_{i=1}^{N}
    \log
    P_{\tilde U_i|\tilde U^{i-1},X^N}^{\theta^X}
    (\tilde u_{j,i}|\tilde u_j^{i-1},x_j^N).
    \label{eq:wz_loss_x}
\end{align}
Averaging these losses gives estimates of
$\frac{1}{N}\mathsf{H}_\phi(\tilde U^N|Y^N)$ and
$\frac{1}{N}\mathsf{H}_\phi(\tilde U^N|X^N)$, respectively, and the
normalized conditional mutual information is estimated by the plug-in estimate
\begin{equation}
    \frac{1}{N}
    \widehat{\mathsf{I}}_{\theta}(\mathcal D_{M,N}^{\phi})
    =
    \frac{1}{M}
    \sum_{j=1}^{M}
    \left[
        \mathcal L_Y(y_j^N,u_j^N)
        -
        \mathcal L_X(x_j^N,u_j^N)
    \right],
    \label{eq:wz_cmi_estimator}
\end{equation}
where $\theta=\{\theta^Y,\theta^X\}$.

\subsubsection*{Optimization Stage}
With the NPD parameters fixed, we update the test-channel parameters
$\phi$ and the reconstruction parameters $\psi$. Following the
multi-letter formulation in \eqref{eq:wz_multiletter}, we use the
Lagrangian with randomized reconstruction
\begin{equation}
    J_\lambda(\phi,\psi)
    \triangleq
    \frac{1}{N}
    \mathsf{I}_\phi(X^N;U^N|Y^N)
    +
    \lambda
    \mathbb E_{P_{X^N,Y^N}P_{U^N|X^N}^{\phi}P_{\hat X^N|U^N,Y^N}^{\psi}}
    \left[
        d_N
        \left(
            X^N,
            \hat X^N
        \right)
    \right],
    \label{eq:wz_lagrangian}
\end{equation}
where $\lambda\ge 0$ controls the rate--distortion tradeoff.
The conditional information density associated with the first term is
\begin{equation}
    \imath_\phi(x^N,y^N,u^N)
    \triangleq
    \log
    \frac{
        P_{U^N|X^N}^{\phi}(u^N|x^N)
    }{
        P_{U^N|Y^N}^{\phi}(u^N|y^N)
    }.
    \label{eq:wz_information_density}
\end{equation}
The following theorem gives the corresponding score-function gradient for
the Wyner--Ziv objective.

\begin{theorem}[Score-function gradient for the Wyner--Ziv objective]
\label{th:pg_wz}
Under the softmax test-channel parameterization in
Section~\ref{sec:WZ_test_channel}, assume that the test-channel probabilities
are differentiable with respect to $\phi$. For fixed source--side-information
distribution $P_{X^N,Y^N}$ and reconstruction parameters $\psi$, the gradient
of $J_\lambda(\phi,\psi)$ in \eqref{eq:wz_lagrangian} with respect to $\phi$ is
\begin{equation}
    \nabla_\phi J_\lambda(\phi,\psi)
    =
    \mathbb E
    \left[
        \nabla_\phi
        \ln
        P_{U^N|X^N}^{\phi}(U^N|X^N)
        \left(
            \frac{1}{N}
            \imath_\phi(X^N,Y^N,U^N)
            +
            \lambda
            d_N
            \left(
                X^N,
                \hat X^N
            \right)
        \right)
    \right],
    \label{eq:wz_exact_phi_gradient}
\end{equation}
where the expectation is with respect to
$P_{X^N,Y^N}P_{U^N|X^N}^{\phi}P_{\hat X^N|U^N,Y^N}^{\psi}$.
\end{theorem}
The proof is given in Appendix~\ref{app:PG_WZ}. It extends the argument
of Theorem~\ref{th:pg_rd} to conditional mutual information.

As in Section~\ref{sec:RD_opt}, the exact information density in
\eqref{eq:wz_information_density} is unavailable and is replaced by an NPD-based plug-in estimate.
For fixed NPD parameters $\theta=\{\theta^{Y},\theta^{X}\}$,
define
\begin{align}
    \hat \imath_{\theta}(x^N,y^N,u^N)
    &\triangleq
    \log
    \frac{
        P_{\tilde U^N|X^N}^{\theta^{X}}(\tilde u^N|x^N)
    }{
        P_{\tilde U^N|Y^N}^{\theta^{Y}}(\tilde u^N|y^N)
    }
    \nonumber\\
    &=
    \sum_{i=1}^{N}
    \log
    \frac{
        P_{\tilde U_i|\tilde U^{i-1},X^N}^{\theta^{X}}
        (\tilde u_i|\tilde u^{i-1},x^N)
    }{
        P_{\tilde U_i|\tilde U^{i-1},Y^N}^{\theta^{Y}}
        (\tilde u_i|\tilde u^{i-1},y^N)
    },
    \qquad
    \tilde u^N=u^N G_N ,
    \label{eq:wz_npd_information_density}
\end{align}
which is the counterpart of \eqref{eq:rd_npd_information_density} with the unconditional distribution
replaced by the distribution conditioned on the side information $Y^N$.
For a mini-batch
\[
    \mathcal D_{M,N}^{\phi,\psi}
    =
    \{(x_j^N,y_j^N,u_j^N,\hat x_j^N)\}_{j=1}^{M}
    \sim P_{X^N,Y^N}P_{U^N|X^N}^{\phi}P_{\hat X^N|U^N,Y^N}^{\psi},
\]
drawn independently at the current parameter values, we hold $\theta$
and $\psi$ fixed for the test-channel update and substitute
\eqref{eq:wz_npd_information_density} for the exact information density in
\eqref{eq:wz_exact_phi_gradient}, which gives the plug-in Monte Carlo gradient estimate
\begin{align}
    \widehat{\nabla_\phi J_\lambda}
    (\theta,\phi,\psi;\mathcal D_{M,N}^{\phi,\psi})
    =
    \frac{1}{M}
    \sum_{j=1}^{M}
    &
    \nabla_\phi
    \ln
    P_{U^N|X^N}^{\phi}(u_j^N|x_j^N)
    \nonumber\\
    &\times
    \left[
        \frac{1}{N}
        \hat \imath_{\theta}(x_j^N,y_j^N,u_j^N)
        +
        \lambda
        d_N
        \left(
            x_j^N,
            \hat x_j^N
        \right)
    \right].
    \label{eq:wz_empirical_phi_gradient}
\end{align}
As in the rate--distortion case, for fixed $\theta$ and $\psi$, \eqref{eq:wz_empirical_phi_gradient} is
an unbiased estimate of the gradient with respect to $\phi$ of the learned surrogate
$\mathbb E[N^{-1}\hat\imath_{\theta}(X^N,Y^N,U^N)+\lambda d_N(X^N,\hat X^N)]$,
in which $\hat\imath_\theta$ is held fixed. If both learned
conditional distributions coincide with their true counterparts at the current $\phi$, the estimate is unbiased for $\nabla_\phi J_\lambda(\phi,\psi)$ by Theorem~\ref{th:pg_wz}.

The reconstruction parameters $\psi$ influence only the distortion term.
With $\phi$ fixed, applying the same score-function argument as in
Theorem~\ref{th:pg_wz} to the reconstruction distribution
$P_{\hat X^N|U^N,Y^N}^{\psi}$ gives
\begin{equation}
    \nabla_\psi J_\lambda(\phi,\psi)
    =
    \lambda
    \mathbb E_{P_{X^N,Y^N}P_{U^N|X^N}^{\phi}P_{\hat X^N|U^N,Y^N}^{\psi}}
    \left[
        d_N(X^N,\hat X^N)
        \nabla_\psi\ln P_{\hat X^N|U^N,Y^N}^{\psi}(\hat X^N|U^N,Y^N)
    \right].
    \label{eq:wz_psi_gradient}
\end{equation}
The derivative acts on the reconstruction log-probability rather than
on $d_N$, so \eqref{eq:wz_psi_gradient} can be applied directly to Hamming
distortion. An unbiased mini-batch
estimate is
\[
    \widehat{\nabla_\psi J_\lambda}
    =\frac{\lambda}{M}\sum_{j=1}^{M}
    d_N(x_j^N,\hat x_j^N)
    \nabla_\psi\ln P_{\hat X^N|U^N,Y^N}^{\psi}(\hat x_j^N|u_j^N,y_j^N).
\]
The parameters are updated according to
\begin{align}
    \phi
    &\leftarrow
    \phi
    -
    \gamma_\phi
    \widehat{\nabla_\phi J_\lambda}
    (\theta,\phi,\psi;\mathcal D_{M,N}^{\phi,\psi}),
    \label{eq:wz_phi_update}
    \\
    \psi
    &\leftarrow
    \psi
    -
    \gamma_\psi
    \widehat{\nabla_\psi J_\lambda}.
    \label{eq:wz_psi_update}
\end{align}

The optimization step is alternated with the estimation step above exactly as in Section~\ref{sec:RD_opt}, including the warm-up stage in which the NPDs are trained with $P_{U^N|X^N}^{\phi}$ fixed before the first update of $(\phi,\psi)$. The complete procedure, which differs from Algorithm~\ref{alg:RD} only in the estimation losses \eqref{eq:wz_loss_y}--\eqref{eq:wz_loss_x}, the additional sampling of $\hat x^N$ from $P_{\hat X^N|U^N,Y^N}^{\psi}$, and the update \eqref{eq:wz_psi_update} of $\psi$, is summarized in Algorithm~\ref{alg:WZ} in Appendix~\ref{app:implementation_details}. Repeating the two steps estimates the
Wyner--Ziv rate--distortion tradeoff and learns a lossy polar coding
scheme with decoder side information, which we describe next.

\subsubsection*{Polar Code Construction}

The lossy construction of Section~\ref{sec:RD_opt} is adapted to the Wyner--Ziv setting as follows.
Let $\phi^\ast$ and $\psi^\ast$ denote the optimized test-channel and
reconstruction parameters, respectively, and let
$\theta^\ast=\{\theta^{X,\ast},\theta^{Y,\ast}\}$ denote the final trained
NPD parameters.
The main change is that the decoder has access to $Y^N$ rather than to the
unconditional distribution of the auxiliary sequence. Consequently, the
rate-carrying coordinates are those that are reliable from the encoder's
view of $X^N$ but remain uncertain from the decoder's side information
$Y^N$.

The index sets are chosen from a design set generated by the optimized
model. We draw
$(X^N,Y^N,U^N)\sim P_{X^N,Y^N}P_{U^N|X^N}^{\phi^\ast}$,
form the polarized auxiliary sequence $\tilde U^N=U^NG_N$, and evaluate
the learned conditional probabilities under teacher forcing, as in the lossy construction. For each
coordinate, define
\begin{equation}
    \widehat h_i^X
    \triangleq
    \widehat{\mathsf{H}}_{\theta^{X,\ast}}^{\phi^\ast}
    (\tilde U_i|\tilde U^{i-1},X^N),
    \qquad
    \widehat h_i^Y
    \triangleq
    \widehat{\mathsf{H}}_{\theta^{Y,\ast}}^{\phi^\ast}
    (\tilde U_i|\tilde U^{i-1},Y^N),
    \qquad i=1,\ldots,N,
\end{equation}
where the empirical averages are taken over the design set. Let
$q=|\mathcal U|$, choose a common threshold parameter $0<\delta<1/2$,
and define
\[
\begin{aligned}
    \mathcal F_y
    &\triangleq \{i\in[N]:\widehat h_i^Y>\delta\log q\},\\
    \mathcal F_x
    &\triangleq \mathcal F_y\cap
    \{i\in[N]:\widehat h_i^X\ge(1-\delta)\log q\}.
\end{aligned}
\]
The source-frozen coordinates are $\mathcal F_x$, and
the transmitted coordinates are $\mathcal A=\mathcal F_y\setminus\mathcal F_x$.
The sets $\mathcal F_x$, $\mathcal A$, and $[N]\setminus\mathcal F_y$
form a disjoint partition: frozen coordinates are shared, low-entropy
coordinates outside $\mathcal F_y$ are reconstructed using side
information, and every remaining coordinate is transmitted. With true
entropies and $\delta=\delta_N$, this partition agrees with that of
Theorem~\ref{thm:wz_faim_achievability} for sufficiently large $N$ under
its assumptions.
Before encoding, the terminals also generate shared symbols
$s_{\mathcal F_x}\in\mathcal U^{|\mathcal F_x|}$ from a common seed,
independently and uniformly over $\mathcal U$ and independently of the
source and side information.

Encoding and decoding then proceed as in the lossy construction, with the two conditional NPDs in place of the conditional and unconditional pair. For source block $x^N$, the encoder forms the
source-dependent embedding $E_{\theta_{E_X}^\ast}(x^N)$ and runs the
learned SC recursion with the conditional distributions given $X^N$,
$P_{\tilde U_i|\tilde U^{i-1},X^N}^{\theta^{X,\ast}}(\cdot|\tilde u^{i-1},x^N)$,
fixing $\tilde u_i=s_i$ for $i\in\mathcal F_x$ and sampling
$\tilde u_i$ from this distribution for every $i\notin\mathcal F_x$.
It transmits $\tilde u_{\mathcal A}$ and leaves the coordinates outside $\mathcal F_y$ untransmitted, so that the rate is
$R=|\mathcal A|\log|\mathcal U|/N$ bits per source symbol.
Given $y^N$
and $\tilde u_{\mathcal A}$, the decoder forms the side-information embedding
$E_{\theta_{E_Y}^\ast}(y^N)$, regenerates $s_{\mathcal F_x}$, and performs
SC/SCL decoding with the learned distributions given $Y^N$,
$P_{\tilde U_i|\tilde U^{i-1},Y^N}^{\theta^{Y,\ast}}(\cdot|\hat{\tilde u}^{i-1},y^N)$,
where coordinates in $\mathcal F_x$ are fixed to the shared symbols,
coordinates in $\mathcal A$ are fixed to the received symbols, and
coordinates outside $\mathcal F_y$ are estimated.
The resulting sequence $\hat{\tilde u}^N$ is passed through
the inverse polar transform and the reconstruction mapping,
\begin{equation}
    \hat u^N=\hat{\tilde u}^N G_N^{-1},
    \qquad
    \hat x^N=g_{\psi^\ast}(\hat u^N,y^N).
\end{equation}

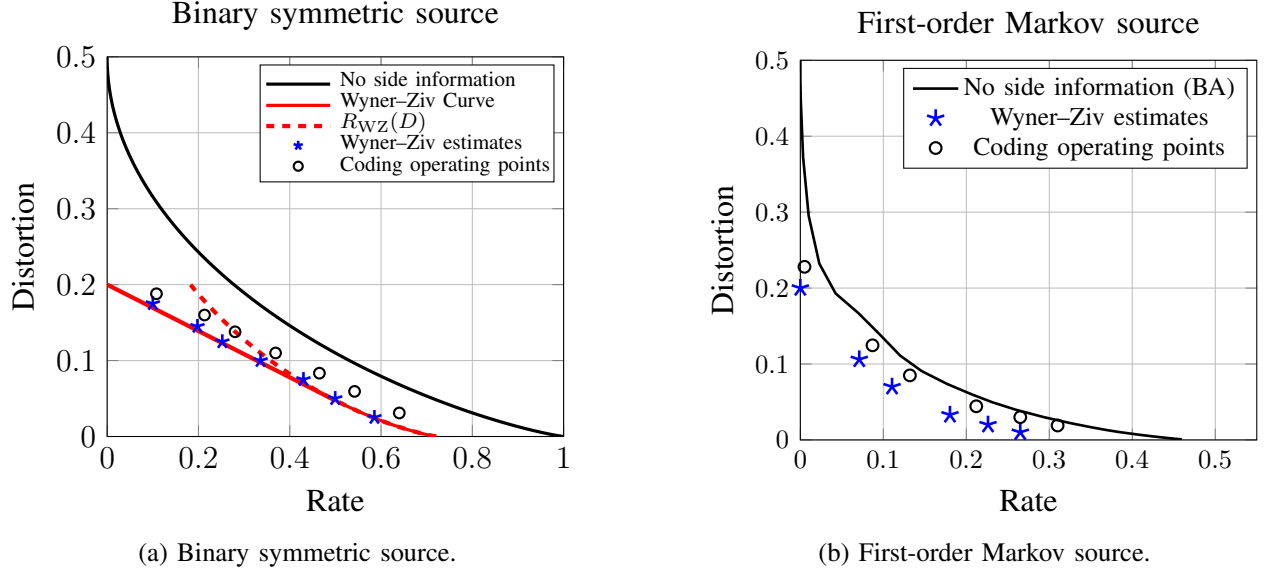
\begin{figure}[!t]
    \centering

    \begin{subfigure}{0.47\textwidth}
        \centering

\pgfmathdeclarefunction{hbin}{1}{%
  \pgfmathparse{(#1<=0 || #1>=1) ? 0 : -(#1*ln(#1) + (1-#1)*ln(1-#1))/ln(2)}%
}

\begin{tikzpicture}
\begin{axis}[
    width=3in,
    height=2.6in,
    grid=both,
    xlabel={Rate},
    ylabel={Distortion},
    title={Binary symmetric source},
    xmin=0, xmax=1,
    ymin=0, ymax=0.5
]

\addplot[
    very thick,
    black,
    domain=0.0001:0.4999,
    samples=800
]
({1 - hbin(x)}, {x});

\addplot[
    ultra thick,
    red,
    dashed,
    domain=0.0001:0.1999,
    samples=800
]
({hbin(0.2*(1-x) + 0.8*x) - hbin(x)}, {x});

\addplot[
    very thick,
    red,
    domain=0.0001:0.04735,
    samples=400
]
({hbin(0.2*(1-x) + 0.8*x) - hbin(x)}, {x});

\addplot[
    ultra thick,
    red,
    domain=0.04735:0.2,
    samples=200
]
({(0.50020/(0.04735-0.2))*(x-0.2)}, {x});

\addplot[
    only marks,
    blue,
    mark=star,
    mark size=2.8,
    line width=1.0,
    mark options={fill=yellow!90!white, draw=blue}
]
coordinates {
    (0.1000, 0.1750)
    (0.1978, 0.1450)
    (0.2523, 0.1248)
    (0.3360, 0.1000)
    (0.4300, 0.0750)
    (0.4996, 0.0500)
    (0.5855, 0.0251)
};

\addplot[
    only marks,
    mark=o,
    mark size=2.0,
    line width=0.8,
    mark options={fill=yellow!90!white, draw=black}
]
coordinates {
    (0.2802, 0.1379)
    (0.5420, 0.0594)
    (0.4650, 0.0834)
    (0.2133, 0.1600)
    (0.1080, 0.1880)
    (0.3690, 0.1100)
    (0.6400, 0.0310)
};

\node[
    anchor=north east,
    draw=black,
    fill=white,
    inner sep=1.5pt,
    font=\scriptsize
] at (rel axis cs:0.98,0.98) {%
\begin{tabular}{@{}c@{\hspace{3pt}}l@{}}

\tikz{\draw[black,very thick] (0,0) -- (0.14,0);} &
No side information \\[-1pt]

\tikz{\draw[red,very thick] (0,0) -- (0.14,0);} &
Wyner--Ziv Curve \\[-1pt]

\tikz{\draw[red,ultra thick,dashed] (0,0) -- (0.14,0);} &
$R_{\mathrm{WZ}}(D)$ \\[-1pt]

\tikz{\draw[draw=blue,fill=yellow!90!white,line width=1.0]
plot[only marks,mark=star,mark size=1.8] coordinates {(0.11,0)};} &
Wyner--Ziv estimates \\[-1pt]

\tikz{\draw[draw=black,fill=yellow!90!white,line width=0.8]
plot[only marks,mark=o,mark size=1.7] coordinates {(0.11,0)};} &
Coding operating points

\end{tabular}%
};

\end{axis}
\end{tikzpicture}
        \caption{Binary symmetric source.}
        \label{fig:WZ_BSC}
    \end{subfigure}
    \hspace{0.015\textwidth}
    \begin{subfigure}{0.47\textwidth}
        \centering

\begin{tikzpicture}
\begin{axis}[
    width=3in,
    height=2.6in,
    grid=both,
    xlabel={Rate},
    ylabel={Distortion},
    xlabel style={font=\normalsize},
    ylabel style={font=\normalsize},
    xmin=0, xmax=0.55,
    ymin=0, ymax=0.5,
    title={First-order Markov source},
    title style={font=\normalsize},
    tick label style={font=\small},
    legend style={
        at={(0.98, 0.98)},
        anchor=north east,
        nodes={scale=0.8, transform shape}
    }
]

\addplot[
    thick,
    black,
    line width=1.0
]
coordinates {
    (4.596849e-01, 7.561174e-04)
    (4.534614e-01, 1.358285e-03)
    (4.438267e-01, 2.378787e-03)
    (4.294139e-01, 4.060263e-03)
    (4.092112e-01, 6.682294e-03)
    (3.835689e-01, 1.047899e-02)
    (3.682486e-01, 1.315574e-02)
    (3.502609e-01, 1.674589e-02)
    (3.255332e-01, 2.211432e-02)
    (2.936764e-01, 2.994642e-02)
    (2.600542e-01, 3.963778e-02)
    (2.313661e-01, 4.959728e-02)
    (2.059031e-01, 6.028679e-02)
    (1.765300e-01, 7.389809e-02)
    (1.461889e-01, 9.066529e-02)
    (1.201853e-01, 1.111429e-01)
    (9.653424e-02, 1.379909e-01)
    (7.000408e-02, 1.669979e-01)
    (4.254666e-02, 1.928439e-01)
    (2.286880e-02, 2.322570e-01)
    (1.002547e-02, 2.952895e-01)
    (3.344648e-03, 3.721938e-01)
    (4.927404e-04, 4.480390e-01)
    (2.482829e-05, 4.772802e-01)
    (0.0, 5.000000e-01)
};
\addlegendentry{No side information (BA)}

\addplot[
    only marks,
    blue,
    mark=star,
    mark size=3.4,
    line width=1.0,
    mark options={fill=yellow!90!white, draw=blue}
]
coordinates {
    (0.265, 0.01)
    (0.22615, 0.02)
    (0.1803, 0.033)
    (0.1105, 0.07)
    (0.071, 0.106)
    (0.0, 0.2)
};
\addlegendentry{Wyner--Ziv estimates}

\addplot[
    only marks,
    mark=o,
    mark size=2.2,
    line width=0.8,
    mark options={fill=yellow!90!white, draw=black}
]
coordinates {
    (0.310, 0.0190)
    (0.265, 0.0300)
    (0.212, 0.0443)
    (0.132, 0.0850)
    (0.087, 0.1246)
    (0.005, 0.2280)
};
\addlegendentry{Coding operating points}

\end{axis}
\end{tikzpicture}
        \caption{First-order Markov source.}
        \label{fig:WZ_Markov}
    \end{subfigure}

    \caption{Rate--distortion performance for Wyner--Ziv coding with decoder side information generated through a binary symmetric channel with crossover probability $0.2$. Blue stars indicate estimated Wyner--Ziv $(R,D)$ pairs obtained by alternating optimization of the conditional NPDs, auxiliary test channel, and reconstruction model, and circular markers indicate the corresponding coding operating points obtained using SCL decoding with block length $N=2^{13}$ and list size $L=16$. In (a), the black curve is the rate--distortion function without side information, the dashed red curve is the branch $R_{\mathrm{WZ}}(D)$ before time sharing, and the solid red curve is the Wyner--Ziv rate--distortion curve. In (b), the black curve is the Blahut--Arimoto estimate of the rate--distortion curve without side information.}

    \label{fig:WZ_combined}
\end{figure}

\subsection{Experiments}

We evaluate the proposed neural--polar Wyner--Ziv construction in two binary-source settings. The results are reported in Fig.~\ref{fig:WZ_combined}.
We first consider a binary symmetric source with decoder side information generated according to $Y = X \oplus Z$, where $Z \sim \mathrm{Bernoulli}(p)$ is independent of $X$ and $p=0.2$.
For this model, the Wyner--Ziv rate--distortion function is known explicitly and therefore provides a benchmark for the proposed estimation procedure.
Following~\cite{korada2010polar}, the dashed red curve in Fig.~\ref{fig:WZ_BSC} is labeled $R_{\mathrm{WZ}}(D)=h_2(p(1-D)+(1-p)D)-h_2(D)$, $0\le D\le p$, denoting the achievable branch before time sharing, where $h_2$ is the binary entropy function. The solid red Wyner--Ziv curve is its lower convex envelope with the zero-rate point $(R,D)=(0,p)$, obtained by time sharing with reconstruction from the side information alone.
As shown in Fig.~\ref{fig:WZ_BSC}, the values obtained by jointly optimizing the auxiliary test channel and reconstruction model, together with the corresponding conditional NPDs, closely track the analytical Wyner--Ziv curve.
The figure also includes the operating points of the resulting polar coding scheme, obtained using the learned test channel, conditional NPDs, and reconstruction mapping.
We further consider a first-order Markov source with the same $\mathrm{BSC}(0.2)$ side-information mechanism, which assesses the method for a source with memory.
Fig.~\ref{fig:WZ_Markov} compares the Wyner--Ziv estimates with a no-side-information Blahut--Arimoto benchmark. The lower rates of the Wyner--Ziv points therefore quantify the gain provided by the decoder side information.
The experiment further demonstrates that the same learning and polar-coding procedure applies beyond the memoryless benchmark.

Finally, we evaluate the Wyner--Ziv estimation procedure on the real
\texttt{E. coli} sequence considered in the preceding lossless and lossy
experiments. Decoder-only side information is generated by passing each source
symbol through a $q$-ary symmetric channel with crossover probability
$p_{\mathrm{side}}=0.2$:
\begin{equation}
    P_{Y|X}(y|x)
    =
    \begin{cases}
        1-p_{\mathrm{side}}, & y=x,\\
        \dfrac{p_{\mathrm{side}}}{q-1}, & y\ne x.
    \end{cases}
\end{equation}
We use the symbolwise Hamming distortion
$d(x,\hat x)=\mathbf 1\{x\ne\hat x\}$. Fig.~\ref{fig:WZ_EColi}
compares the NPD-estimated rate--distortion pairs obtained without side
information in Section~\ref{sec: RD} with those obtained by the proposed
Wyner--Ziv procedure using decoder side information, together with the
corresponding coding operating points. The side information substantially
reduces the rate required at a given distortion.

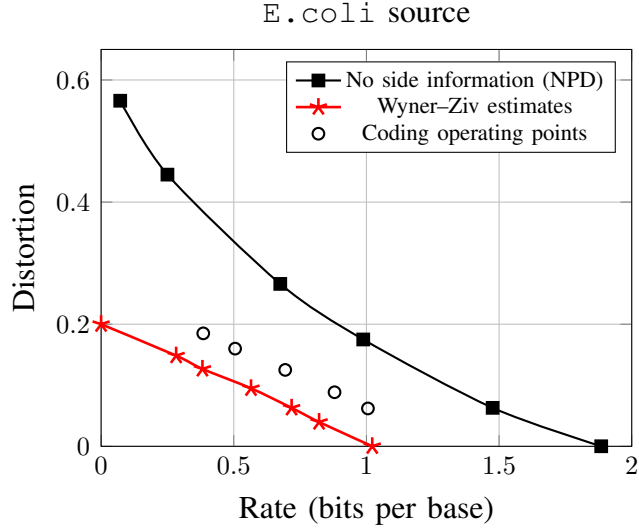
\begin{figure}[t]
    \centering
    \resizebox{0.48\textwidth}{!}{\begin{tikzpicture}
\begin{axis}[
    width=3.4in,
    height=2.7in,
    grid=both,
    xlabel={Rate (bits per base)},
    ylabel={Distortion},
    xmin=0, xmax=2.0,
    ymin=0, ymax=0.65,
    title={\texttt{E.coli} source},
    title style={font=\normalsize},
    tick label style={font=\small},
    legend style={
        at={(0.97,0.97)},
        anchor=north east,
        nodes={scale=0.72, transform shape}
    }
]

\addplot[
    black,
    thick,
    smooth,
    mark=square*,
    mark size=2.2,
    mark options={fill=black, draw=black}
]
coordinates {
    (0.072, 0.566)
    (0.250, 0.445)
    (0.6753, 0.266)
    (0.988, 0.175)
    (1.4758, 0.063)
    (1.885, 0.0001)
};
\addlegendentry{No side information (NPD)}

\addplot[
    red,
    thick,
    smooth,
    mark=star,
    mark size=3.4,
    line width=1.0,
    mark options={fill=yellow!90!white, draw=red}
]
coordinates {
    (0, 0.2)
    (0.282804, 0.1485)
    (0.381924, 0.12677)
    (0.565, 0.095)
    (0.718, 0.063)
    (0.821574, 0.0398)
    (1.021976, 0.0001)
};
\addlegendentry{Wyner--Ziv estimates}

\addplot[
    only marks,
    mark=o,
    mark size=2.2,
    line width=0.8,
    mark options={fill=yellow!90!white, draw=black}
]
coordinates {
    (0.385000, 0.185000)
    (0.505000, 0.160000)
    (0.693907, 0.125156)
    (0.879849, 0.088556)
    (1.005722, 0.062021)
};
\addlegendentry{Coding operating points}
\end{axis}
\end{tikzpicture}}
    \caption{Rate--distortion performance on the
    \texttt{E. coli} source. Black squares show the estimates obtained without
    side information in Section~\ref{sec: RD}. Red stars show the
    Wyner--Ziv estimates when decoder-only side information is generated
    through a quaternary symmetric channel with
    $p_{\mathrm{side}}=0.2$, and circular markers show the corresponding
    coding operating points obtained using SCL decoding at block length
    $N=2^n$ with $n=13$ and list size $L=16$.}
    \label{fig:WZ_EColi}
\end{figure}

The corresponding implementation parameters for the Wyner--Ziv experiments, including the conditional NPDs, auxiliary test channel, reconstruction mapping, and alternating optimization schedule, are reported in Table~\ref{tab:wz_reproducibility} in Appendix~\ref{app:implementation_details}.

\section{Conclusions and Future Work}\label{sec:conc}
This paper introduced a data-driven polar-coding framework for source coding based on neural polar decoders. The main idea is to replace the explicit successive-cancellation update rules of classical polar constructions with trainable modules that estimate the relevant polarized conditional distributions directly from samples. This yields a unified approach to lossless compression, lossy compression, and Wyner--Ziv coding, while avoiding the explicit quadratic alphabet enumeration and cubic state-space enumeration that limit analytical SC and SCT decoders for large alphabets and sources with memory.

The proposed schemes retain the recursive polar structure and use NPDs to identify the polarized coordinates required for compression and reconstruction. In the lossy setting, NPDs were combined with a neural test channel to optimize rate--distortion operating points from data, and the Wyner--Ziv extension incorporated decoder side information through learned conditional polarized distributions and a learned reconstruction mapping, supported by conditional polarization and nested polar coding results for FAIM processes. The numerical results show that, when analytical benchmarks are available, the learned constructions closely match them, and that the same methodology remains applicable in settings where explicit source models or analytical decoding rules are unavailable.

An important direction for future work is extending the NPD framework to multiterminal source coding. A representative example is Slepian--Wolf coding, which requires learning coupled synthetic conditional distributions and constructing encoder-specific transmitted sets and a joint decoder for target rate tuples, while each encoder uses only its local observation.

\begin{appendices}

\section{Proof of Lemma \ref{lem:memoryless_sc}}\label{app:sc_alphbet}
This appendix proves the non-binary SC update formulas in
Lemma~\ref{lem:memoryless_sc}. The main idea is to use the polar-kernel
decomposition and source memorylessness to derive the probability-domain
recursions for odd and even indices, and then convert these identities into
the LLR-domain operations $f_k^{+}$ and $f_k^{-}$.

\begin{proof}
The derivation follows that of Ar\i{}kan in~\cite{arikan2009channel}, adapted
to source coding over $\mathbb{GF}(q)$ without channel-output variables and
to the finite-field kernel in~\eqref{eq:kernel2}. All sums below are over the
corresponding symbols in $\mathcal{X}$. The recursive polar transform
generated by \eqref{eq:kernel2} gives
\begin{equation}
    x_1^NG_N=u_{1,o}^{2N}\ominus_q \alpha u_{1,e}^{2N},
    \qquad
    x_{N+1}^{2N}G_N=u_{1,e}^{2N}.
\end{equation}
Since the source is memoryless, the two length-$N$ input blocks are independent, and hence
\begin{equation}
    P_{2N}(u^{2N})
    =P_N(u_{1,o}^{2N}\ominus_q \alpha u_{1,e}^{2N})P_N(u_{1,e}^{2N}).
    \label{eq:appendix_memoryless_factorization}
\end{equation}
Therefore,
\begin{align}
    P_{2N}^{(2i-1)}(u^{2i-1})
    &= \sum_{u_{2i}^{2N}} P_{2N}(u^{2N}) \nonumber\\
    &= \sum_{u_{2i}}\sum_{u_{2i+1,e}^{2N}}
       P_N(u_{1,e}^{2N})
       \sum_{u_{2i+1,o}^{2N}}
       P_N(u_{1,o}^{2N}\ominus_q \alpha u_{1,e}^{2N}) \nonumber\\
    &= \sum_{u_{2i}}
       P_N^{(i)}(u_{1,e}^{2i})
       P_N^{(i)}(u_{1,o}^{2i}\ominus_q \alpha u_{1,e}^{2i}),\label{eq:memoryless_prob_odd}
\end{align}
where the last equality follows from the definition of $P_N^{(i)}$ in \eqref{eq: pn^i}. In particular, for fixed $u_{1,e}^{2N}$, summing over $u_{2i+1,o}^{2N}\in\mathcal{X}^{N-i}$ ensures that the last $N-i$ coordinates of $u_{1,o}^{2N}\ominus_q \alpha u_{1,e}^{2N}$ also range over $\mathcal{X}^{N-i}$.

Similarly, if $u_{2i}$ is also fixed, then
\begin{align}
    P_{2N}^{(2i)}(u^{2i})
    &= \sum_{u_{2i+1}^{2N}} P_{2N}(u^{2N}) \nonumber\\
    &= \sum_{u_{2i+1,e}^{2N}}
       P_N(u_{1,e}^{2N})
       \sum_{u_{2i+1,o}^{2N}}
       P_N(u_{1,o}^{2N}\ominus_q \alpha u_{1,e}^{2N}) \nonumber\\
    &= P_N^{(i)}(u_{1,e}^{2i})
       P_N^{(i)}(u_{1,o}^{2i}\ominus_q \alpha u_{1,e}^{2i}).\label{eq:memoryless_prob_even}
\end{align}
Equations~\eqref{eq:memoryless_prob_odd} and~\eqref{eq:memoryless_prob_even} are the probability-domain recursions for odd and even indices, respectively. We next use them to derive the log-domain updates in Lemma~\ref{lem:memoryless_sc}.

Fix $k\in\mathcal{X}$, and define
\begin{equation}
    a\triangleq \hat{u}_{1,o}^{2i-2}\ominus_q \alpha\hat{u}_{1,e}^{2i-2},
    \qquad
    b\triangleq \hat{u}_{1,e}^{2i-2}.
\end{equation}
We first prove \eqref{eq:cn_memoryless}. By the definition of the LLRs and \eqref{eq:memoryless_prob_odd},
    \begin{align}
        &L_{2N}^{(2i-1,k)}(\hat{u}^{2i-2})\nn\\
        &=\ln\frac{P_{2N}^{(2i-1)}\left(\hat{u}^{2i-2},u_{2i-1}=k\right)}{P_{2N}^{(2i-1)}\left(\hat{u}^{2i-2},u_{2i-1}=0\right)}\nn\\
        &=\ln\frac{\sum_{u}P_N^{(i)}(b,u)P_N^{(i)}(a,k\ominus_q \alpha u)}{\sum_{u}P_N^{(i)}(b,u)P_N^{(i)}(a,0\ominus_q \alpha u)}\nn\\
        &=\ln\left(\sum_{u}\frac{P_N^{(i)}(b,u)}{P_N^{(i)}(b,0)}\cdot\frac{P_N^{(i)}(a,k\ominus_q \alpha u)}{P_N^{(i)}(a,0)}\right)\nn\\
        &\quad-\ln\left(\sum_{u}\frac{P_N^{(i)}(b,u)}{P_N^{(i)}(b,0)}\cdot\frac{P_N^{(i)}(a,0\ominus_q \alpha u)}{P_N^{(i)}(a,0)}\right)\nn\\
        &=\ln\left(\sum_{u}e^{L_N^{(i,u)}(b)}e^{L_{N}^{(i,k\ominus_q \alpha u)}(a)}\right)
        -\ln\left(\sum_{u}e^{L_N^{(i,u)}(b)}e^{L_{N}^{(i,0\ominus_q \alpha u)}(a)}\right)\nn\\
        &=f^{+}_k\left(\textbf{L}_{N}^{(i)}(a),\textbf{L}_{N}^{(i)}(b)\right),
    \end{align}
which is precisely \eqref{eq:cn_memoryless} by the definition of $f_k^+$ in \eqref{eq:sc1}.

We next prove \eqref{eq:bn_memoryless}. Using \eqref{eq:memoryless_prob_even},
    \begin{align}
        &L_{2N}^{(2i,k)}(\hat{u}^{2i-1})\nn\\
        &=\ln\left(\frac{P_{2N}^{(2i)}\left(\hat{u}^{2i-1},u_{2i}=k\right)}{P_{2N}^{(2i)}\left(\hat{u}^{2i-1},u_{2i}=0\right)}\right)\nn\\
        &=\ln\frac{P_N^{(i)}(b,k)P_N^{(i)}(a,\hat{u}_{2i-1}\ominus_q \alpha k)}{P_N^{(i)}(b,0)P_N^{(i)}(a,\hat{u}_{2i-1}\ominus_q \alpha 0)}\nn\\
        &=L_{N}^{(i,k)}(b)
        +\ln\left(\frac{P_N^{(i)}(a,\hat{u}_{2i-1}\ominus_q \alpha k)}{P_N^{(i)}(a,\hat{u}_{2i-1})}\right)\nn\\
        &=L_{N}^{(i,k)}(b)+L_{N}^{(i,\hat{u}_{2i-1}\ominus_q \alpha k)}(a)
        -L_{N}^{(i,\hat{u}_{2i-1})}(a)\nn\\
        &=f^{-}_k\left(\textbf{L}_{N}^{(i)}(a),\textbf{L}_{N}^{(i)}(b),\hat{u}_{2i-1}\right),
    \end{align}
which is \eqref{eq:bn_memoryless} by the definition of $f_k^-$ in
\eqref{eq:sc2}.
\end{proof}

\section{Proof of Lemma \ref{lem:sct}} \label{app:sct_alphabet}
This appendix proves the SCT recursions in Lemma~\ref{lem:sct}. The main
idea is to split the finite-state source trellis at its midpoint, sum over
the intermediate state, and combine the two length-$N$ components according
to the polar kernel. Marginalizing the undecoded symbols gives the
probability-domain recursions for odd and even indices, whose logarithms yield the stated
log-domain updates.

\begin{proof}
The derivation follows the recursive decomposition used for trellis-based
polar constructions in~\cite{wang2015construction}, adapted here to source
coding over $\mathbb{GF}(q)$ without channel-output variables and to the
finite-field kernel in~\eqref{eq:kernel2}. All sums below are
over the corresponding alphabet. For $N=1$, $U_1=X_1$, so
$P_1^{(1)}(u_1,s_1|s_0)=P_{X_1,S_1|S_0}(u_1,s_1|s_0)$, which gives
\eqref{eq:sct_log_base}, including zero-probability transitions under the
stated logarithm convention. For the recursive step, by definition,
\begin{align}
P^{(2i-1)}_{2N}\left(u_1^{2i-1}, s_{2N} \vert s_0\right)&= \sum_{u^{2N}_{2i}}
P_{U_1^{2N}, S_{2N} \vert S_0}\left(u_1^{2N}, s_{2N} \vert s_0\right) \notag \\
&= \sum_{u^{2N}_{2i}}
P_{X_1^{2N}, S_{2N} \vert S_0}\left(u_1^{2N}G^{-1}_{2N}, s_{2N} \vert s_0\right). \label{eq:app_sct_start}
\end{align}
Using the Markov property of the finite-state source, the joint probability in \eqref{eq:app_sct_start} can be decomposed at time $N$ as
\begin{align}
&P_{X_1^{2N}, S_{2N} \vert S_0}\left(x_1^{2N}, s_{2N} \vert s_0\right) \nn\\
&=\sum_{s_N'}P_{X_1^N, S_N \vert S_0}\left(x_1^N, s_N' \vert s_0\right)
P_{X_{N+1}^{2N}, S_{2N} \vert S_N}\left(x_{N+1}^{2N},s_{2N}\vert s_N'\right). \label{eq:app_sct_markov}
\end{align}
The recursive polar transform generated by \eqref{eq:kernel2} gives
\begin{equation}
    x_1^NG_N=u_{1,o}^{2N}\ominus_q \alpha u_{1,e}^{2N},
    \qquad
    x_{N+1}^{2N}G_N=u_{1,e}^{2N}.
\end{equation}
Time homogeneity ensures that the two length-$N$ segments, conditioned on
their respective initial states, are governed by the same transition law
and hence the same family of metrics $P_N^{(i)}$.
Substituting these identities into \eqref{eq:app_sct_start}--\eqref{eq:app_sct_markov} yields
\begin{align}
&P^{(2i-1)}_{2N}\left(u_1^{2i-1}, s_{2N} \vert s_0\right) \nn\\
&=\sum_{s_N'}\sum_{u^{2N}_{2i}}
P_{U_{1,o}^{2N}\ominus_q \alpha U_{1,e}^{2N}, S_N \vert S_0}
\left(u_{1,o}^{2N}\ominus_q \alpha u_{1,e}^{2N},s_N'\vert s_0\right)
P_{U_{1,e}^{2N},S_{2N}\vert S_N}\left(u_{1,e}^{2N},s_{2N}\vert s_N'\right) \notag \\
&=\sum_{s_N'}\sum_{u_{2i}}
P_N^{(i)}\left(u_{1,o}^{2i}\ominus_q \alpha u_{1,e}^{2i},s_N'\vert s_0\right)
P_N^{(i)}\left(u_{1,e}^{2i},s_{2N}\vert s_N'\right),\label{eq:sct_prob_odd}
\end{align}
where the last equality follows from the definition of $P_N^{(i)}$ in \eqref{eq:p_N}. In particular, for fixed $u_{1,e}^{2N}$, summing over $u_{2i+1,o}^{2N}\in\mathcal{X}^{N-i}$ ensures that the last $N-i$ coordinates of $u_{1,o}^{2N}\ominus_q \alpha u_{1,e}^{2N}$ also range over $\mathcal{X}^{N-i}$. This establishes the probability-domain recursion for odd indices.

For the recursion at even indices, $u_{2i}$ is fixed rather than summed over. The same decomposition therefore gives
\begin{align}
P^{(2i)}_{2N}\left(u_1^{2i}, s_{2N} \vert s_0\right)
&= \sum_{u^{2N}_{2i+1}}
P_{U_1^{2N}, S_{2N} \vert S_0}\left(u_1^{2N}, s_{2N} \vert s_0\right) \notag \\
&=\sum_{s_N'}
P_N^{(i)}\left(u_{1,o}^{2i}\ominus_q \alpha u_{1,e}^{2i},s_N'\vert s_0\right)
P_N^{(i)}\left(u_{1,e}^{2i},s_{2N}\vert s_N'\right).\label{eq:sct_prob_even}
\end{align}
Equations~\eqref{eq:sct_prob_odd} and~\eqref{eq:sct_prob_even} are the probability-domain SCT recursions. Finally, using
\[
P_N^{(i)}(u^i,s_N|s_0)=\exp\left(L_N^{(i)}(u^i,s_N|s_0)\right)
\]
and taking the natural logarithm of both equations gives
\begin{align*}
L_{2N}^{(2i-1)}(u_1^{2i-1},s_{2N}|s_0)
&=\ln\Bigg(\sum_{s_N'}\sum_{u_{2i}}
e^{L_N^{(i)}(u_{1,o}^{2i}\ominus_q\alpha u_{1,e}^{2i},s_N'|s_0)}
e^{L_N^{(i)}(u_{1,e}^{2i},s_{2N}|s_N')}\Bigg),\\
L_{2N}^{(2i)}(u_1^{2i},s_{2N}|s_0)
&=\ln\Bigg(\sum_{s_N'}
e^{L_N^{(i)}(u_{1,o}^{2i}\ominus_q\alpha u_{1,e}^{2i},s_N'|s_0)}
e^{L_N^{(i)}(u_{1,e}^{2i},s_{2N}|s_N')}\Bigg),
\end{align*}
which are precisely \eqref{eq:sct_log_odd} and \eqref{eq:sct_log_even}. 
\end{proof}

\section{Proof of Theorem \ref{th:pg_rd}}\label{app:PG}
This appendix proves the exact score-function gradient in
Theorem~\ref{th:pg_rd}. The main idea is to differentiate the mutual-information
and distortion terms separately. The score-function identity directly gives
the distortion gradient, while the mutual-information derivative is written
using information density. The remaining derivative of the information
density has zero expectation by normalization of the conditional and marginal
distributions.

\begin{proof}
Let
\begin{equation}
    J_\lambda(\phi)
    \triangleq
    \frac{1}{N} \mathsf{I}_\phi(X^N;\hat X^N)
    + \lambda
    \mathbb{E}_{P_{X^N,\hat X^N}^{\phi}}
    \left[d_N(X^N,\hat X^N)\right],
\end{equation}
where
$P_{X^N,\hat X^N}^{\phi}=P_{X^N}P_{\hat X^N|X^N}^{\phi}$ and
$U^N=\hat X^NG_N$. Under the assumptions of Theorem~\ref{th:pg_rd},
all derivatives below are well defined, and differentiation may be
interchanged with the finite summations. Sums over $x^N$ are restricted
to the support of $P_{X^N}$. Since $P_{X^N}$ is independent of $\phi$, we use the
score-function identity
\begin{equation}
    \nabla_\phi P_{\hat X^N|X^N}^{\phi}(\hat x^N|x^N)
    =
    P_{\hat X^N|X^N}^{\phi}(\hat x^N|x^N)
    \nabla_\phi\ln P_{\hat X^N|X^N}^{\phi}(\hat x^N|x^N),
    \label{eq:app_rd_score_identity}
\end{equation}
which gives
\begin{align}
&\nabla_\phi
\mathbb{E}_{P_{X^N,\hat X^N}^{\phi}}
\left[d_N(X^N,\hat X^N)\right] \nn\\
&= \sum_{x^N,\hat x^N}
P_{X^N}(x^N)
\nabla_\phi P_{\hat X^N|X^N}^{\phi}(\hat x^N|x^N)
 d_N(x^N,\hat x^N) \nonumber\\
&= \mathbb{E}_{P_{X^N,\hat X^N}^{\phi}}
\left[
\nabla_\phi\ln P_{\hat X^N|X^N}^{\phi}(\hat X^N|X^N)
 d_N(X^N,\hat X^N)
\right].
\label{eq:app_rd_distortion_grad}
\end{align}
Next, let $\imath_\phi(x^N,\hat x^N)$ be the information density in
\eqref{eq:rd_information_density}, namely
\begin{equation}
    \imath_\phi(x^N,\hat x^N)
    =
    \log\frac{P_{U^N|X^N}^{\phi}(u^N|x^N)}{P_{U^N}^{\phi}(u^N)},
    \qquad u^N=\hat x^NG_N .
\end{equation}
Since the mapping $\hat X^N\mapsto U^N$ is bijective,
$\mathsf{I}_\phi(X^N;\hat X^N)=\mathsf{I}_\phi(X^N;U^N)$, and therefore
\begin{align}
&\nabla_\phi \mathsf{I}_\phi(X^N;\hat X^N) \nn\\
&=\nabla_\phi
\sum_{x^N,\hat x^N}P_{X^N}(x^N)
P_{\hat X^N|X^N}^{\phi}(\hat x^N|x^N)
\imath_\phi(x^N,\hat x^N) \nonumber\\
&=\mathbb{E}_{P_{X^N,\hat X^N}^{\phi}}
\left[
\nabla_\phi\ln P_{\hat X^N|X^N}^{\phi}(\hat X^N|X^N)
\imath_\phi(X^N,\hat X^N)
\right] +
\mathbb{E}_{P_{X^N,\hat X^N}^{\phi}}
\left[\nabla_\phi\imath_\phi(X^N,\hat X^N)\right].
\label{eq:app_rd_mi_grad_decomp}
\end{align}
It remains to show that the second expectation in
\eqref{eq:app_rd_mi_grad_decomp} is zero. For the conditional term,
\begin{align}
&\mathbb{E}
\left[\nabla_\phi\log P_{U^N|X^N}^{\phi}(U^N|X^N)\right] \nn\\
&=\sum_{x^N}P_{X^N}(x^N)
\sum_{u^N}P_{U^N|X^N}^{\phi}(u^N|x^N)
\nabla_\phi\log P_{U^N|X^N}^{\phi}(u^N|x^N) \nonumber\\
&=\frac{1}{\ln 2}\sum_{x^N}P_{X^N}(x^N)
\nabla_\phi\sum_{u^N}P_{U^N|X^N}^{\phi}(u^N|x^N)=0.
\end{align}
Similarly,
\begin{equation}
\mathbb{E}_{U^N\sim P_{U^N}^{\phi}}
\left[\nabla_\phi\log P_{U^N}^{\phi}(U^N)\right]
=\frac{1}{\ln 2}\nabla_\phi\sum_{u^N}P_{U^N}^{\phi}(u^N)=0.
\end{equation}
Consequently,
\begin{equation}
    \mathbb{E}_{P_{X^N,\hat X^N}^{\phi}}
    \left[\nabla_\phi\imath_\phi(X^N,\hat X^N)\right]=0.
\end{equation}
Combining this identity with \eqref{eq:app_rd_distortion_grad},
\eqref{eq:app_rd_mi_grad_decomp}, and the definition of
$J_\lambda(\phi)$ yields
\begin{align}
    \nabla_\phi J_\lambda(\phi)
    =
    \mathbb{E}_{P_{X^N,\hat X^N}^{\phi}}
    \left[
        \nabla_\phi\ln P_{\hat X^N|X^N}^{\phi}(\hat X^N|X^N)
        \left(
            \frac{1}{N}\imath_\phi(X^N,\hat X^N)
            +\lambda d_N(X^N,\hat X^N)
        \right)
    \right],
\end{align}
which is \eqref{eq:rd_exact_gradient}.
\end{proof}

\section{Proof of Lemma \ref{lem:wz_polarization}}
\label{app:wz_polarization}
This appendix proves the fast polarization, nesting, and rate statements in
Lemma~\ref{lem:wz_polarization}. The main idea is to apply conditional FAIM
polarization separately with observations $X^N$ and $Y^N$. The block Markov
relation $U^N-X^N-Y^N$ then orders the coordinatewise conditional entropies,
which yields nested fast high-entropy sets. Their asymptotic size difference
is identified with the conditional mutual-information rate.

\begin{proof}
Let $N=2^n$, let $Z$ denote either $X$ or $Y$, and fix
$\beta\in(0,1/2)$. By assumption, there exists a finite state process
$\{S_t^Z\}_{t\in\mathbb Z}$ such that
$\{(U_t,Z_t,S_t^Z)\}_{t\in\mathbb Z}$ is a FAIM process. Applying
\cite[Cor.~8, Th.~13, App.~B]{shuval2019fast} with $U_t$ as the
$q$-ary source and $Z_t$ as the observation gives conditional polarization
and fast polarization to the high-entropy endpoint. The extension in
\cite[App.~B]{shuval2019fast} uses modulo-$q$ operations and the normalized
$q$-ary conditional entropy. More explicitly, let $z_{i,N}^{Z}$ denote the
corresponding $q$-ary conditional Bhattacharyya parameter, choose $\beta'$
such that $\beta<\beta'<1/2$, and let
$\delta_N'=2^{-N^{\beta'}}$. The non-binary comparison bound in
\cite[Lemma~17]{shuval2019fast} gives
\begin{equation}
    (z_{i,N}^{Z})^2
    \le h_{i,N}^{Z}
    \le \log_q\!\left(1+(q-1)z_{i,N}^{Z}\right).
    \label{eq:app_wz_qary_entropy_bhattacharyya}
\end{equation}
Since $\beta'>\beta$, for every fixed $\epsilon\in(0,1/2)$ and all
sufficiently large $N$, we have $2\delta_N'\le\delta_N<\epsilon$.
Thus $(1-\delta_N')^2\ge1-\delta_N$. To prove the first inclusion below,
take any index $i$ such that $z_{i,N}^{Z}>1-\delta_N'$. Since both sides
are nonnegative, squaring preserves the strict inequality. Combining this
fact with the lower bound in
\eqref{eq:app_wz_qary_entropy_bhattacharyya} gives
\[
    h_{i,N}^{Z}
    \ge (z_{i,N}^{Z})^2
    >(1-\delta_N')^2
    \ge 1-\delta_N,
\]
so $i\in\mathcal H_{U|Z}^{(N)}$ by
\eqref{eq:wz_high_entropy_set}. For the second inclusion, if
$i\in\mathcal H_{U|Z}^{(N)}$, then
$h_{i,N}^{Z}\ge1-\delta_N>1-\epsilon$ because
$\delta_N<\epsilon$. Therefore,
\begin{equation}
    \{i:z_{i,N}^{Z}>1-\delta_N'\}
    \subseteq
    \mathcal H_{U|Z}^{(N)}
    \subseteq
    \{i:h_{i,N}^{Z}>1-\epsilon\}.
    \label{eq:app_wz_fast_high_sandwich}
\end{equation}
The left and right sets in \eqref{eq:app_wz_fast_high_sandwich} both have
limiting density $\overline{\mathsf H}(U|Z)/\log q$ by the fast-to-one and
fixed-threshold polarization results, respectively~\cite{shuval2019fast}.
The squeeze theorem,
followed by multiplication by $\log q$, therefore proves
\eqref{eq:wz_cor_high_limit} for both $Z=X$ and $Z=Y$.

We next establish the coordinatewise ordering.
Since $\tilde U^N$ is a deterministic function of $U^N$, the block Markov
relation in \eqref{eq:wz_markov} implies that
$\tilde U^N-X^N-Y^N$ also forms a Markov chain. By the chain rule,
$0=\mathsf I(\tilde U^N;Y^N|X^N)
=\sum_{i=1}^{N}\mathsf I(\tilde U_i;Y^N|\tilde U^{i-1},X^N)$.
Since every summand is nonnegative, each vanishes, giving
\begin{equation}
    \mathsf H(\tilde U_i|\tilde U^{i-1},X^N,Y^N)
    =
    \mathsf H(\tilde U_i|\tilde U^{i-1},X^N).
    \label{eq:app_wz_conditional_independence}
\end{equation}
Using \eqref{eq:app_wz_conditional_independence}, we obtain
\begin{align}
    h_{i,N}^{Y}-h_{i,N}^{X}
    &=\frac{1}{\log q}\Bigl[
    \mathsf H(\tilde U_i|\tilde U^{i-1},Y^N)
    -
    \mathsf H(\tilde U_i|\tilde U^{i-1},X^N,Y^N)\Bigr]
    \nonumber\\
    &=
    \frac{\mathsf I(\tilde U_i;X^N|\tilde U^{i-1},Y^N)}{\log q}
    \ge 0.
    \label{eq:app_wz_entropy_ordering}
\end{align}
This proves \eqref{eq:wz_cor_entropy_ordering}. The set inclusion follows
directly: if
$i\in\mathcal H_{U|X}^{(N)}$, then
$h_{i,N}^{Y}\ge h_{i,N}^{X}\ge1-\delta_N$, and hence
$i\in\mathcal H_{U|Y}^{(N)}$. This establishes
\eqref{eq:wz_cor_set_nesting}.

By this inclusion and the definition of $\mathcal A_N$,
\begin{align}
    \frac{|\mathcal A_N|}{N}\log q
    &=
    \frac{|\mathcal H_{U|Y}^{(N)}|}{N}\log q
    -
    \frac{|\mathcal H_{U|X}^{(N)}|}{N}\log q.
    \label{eq:app_wz_set_difference}
\end{align}
Taking $N\to\infty$ in \eqref{eq:app_wz_set_difference} and invoking
\eqref{eq:wz_cor_high_limit} yields
\begin{equation}
    \lim_{N\to\infty}\frac{|\mathcal A_N|}{N}\log q
    =
    \overline{\mathsf H}(U|Y)-\overline{\mathsf H}(U|X).
    \label{eq:app_wz_entropy_rate_difference}
\end{equation}
Finally, the block Markov relation gives, for every $N$,
\begin{align}
    \mathsf I(X^N;U^N|Y^N)
    &=
    \mathsf H(U^N|Y^N)-\mathsf H(U^N|X^N,Y^N)
    \nonumber\\
    &=
    \mathsf H(U^N|Y^N)-\mathsf H(U^N|X^N).
    \label{eq:app_wz_cmi_entropy_difference}
\end{align}
After dividing \eqref{eq:app_wz_cmi_entropy_difference} by $N$ and taking
the limit, the right-hand side of
\eqref{eq:app_wz_entropy_rate_difference} equals
$\overline{\mathsf I}(X;U|Y)$. This proves
\eqref{eq:wz_cor_rate_limit}.
\end{proof}

\section{Proof of Theorem \ref{thm:wz_faim_achievability}}
\label{app:wz_faim_achievability}
This appendix proves the operational Wyner--Ziv statement in
Theorem~\ref{thm:wz_faim_achievability}. The proof proceeds in five stages.
First, we define a randomized successive encoder and show that its induced
law approaches the prescribed target law in total variation. Second, we
bound the SC decoding error under the target law and transfer the bound to
the induced law. Third, fast polarization determines the asymptotic rate.
Fourth, total variation and the decoding-error bound transfer the target
distortion guarantee to the actual reconstruction. Finally, a selection
argument removes all encoder randomness and yields a deterministic code.

\begin{proof}
\emph{Randomized construction and target-law approximation.}
For each $i\in[N]$, define the conditional kernel
\begin{equation}
    q_i(a|\tilde u^{i-1},x^N)
    \triangleq
    P_{\tilde U_i|\tilde U^{i-1},X^N}
    (a|\tilde u^{i-1},x^N),
    \qquad a\in\mathbb{GF}(q).
    \label{eq:wz_target_kernel}
\end{equation}
The encoder generates $\tilde U^N$ successively. For
$i\in\mathcal F_N$, it sets $\tilde U_i=C_i$, where the variables
$\{C_i:i\in\mathcal F_N\}$ are independent and uniform on
$\mathbb{GF}(q)$ and are initially available to both terminals. For
$i\notin\mathcal F_N$, it
draws
\begin{equation}
    \tilde U_i
    \sim q_i(\cdot|\tilde U^{i-1},X^N).
    \label{eq:wz_successive_sampling}
\end{equation}
The encoder transmits $\tilde U_{\mathcal I_N}$.

Let $P_N$ denote the target law of $(X^N,Y^N,\tilde U^N)$, and let $Q_N$
denote the law induced by the randomized encoder. The block Markov relation
ensures that the target conditional law of $\tilde U^N$ given
$(X^N,Y^N)$ equals its conditional law given $X^N$. Hence, the sequential
factorizations of $P_N$ and $Q_N$ differ only on $\mathcal F_N$. The
telescoping bound for sequential distributions keeps
the prefix distributed according to $P_N$ at each replacement and gives
\begin{align}
    \|P_N-Q_N\|_{\mathrm{TV}}
    \le
    \sum_{i\in\mathcal F_N}
    \mathbb E_{P_N}\Bigl[
    \bigl\|
    P_{\tilde U_i|\tilde U^{i-1},X^N}
    (\cdot|\tilde U^{i-1},X^N)
    -\operatorname{Unif}(\mathbb{GF}(q))
    \bigr\|_{\mathrm{TV}}
    \Bigr].
    \label{eq:wz_tv_telescope}
\end{align}
Let $A$ take values in $\mathbb{GF}(q)$ and let $B$ be an arbitrary
finite-alphabet observation. First, Pinsker's inequality applied separately
for each $b$ with $P_B(b)>0$ gives, with KL divergence measured in bits,
\begin{equation}
    \left\|P_{A|B}(\cdot|b)
    -\operatorname{Unif}(\mathbb{GF}(q))\right\|_{\mathrm{TV}}
    \le
    \sqrt{\frac{\ln 2}{2}
    D_{\mathrm{KL}}\!\left(
    P_{A|B}(\cdot|b)
    \middle\|
    \operatorname{Unif}(\mathbb{GF}(q))
    \right)}.
    \label{eq:wz_pinsker_pointwise}
\end{equation}
Second, taking the expectation over $B$ in
\eqref{eq:wz_pinsker_pointwise} and applying Jensen's inequality to the
concave square-root function yields
\begin{align}
    &\mathbb E\left[
    \left\|P_{A|B}(\cdot|B)
    -\operatorname{Unif}(\mathbb{GF}(q))\right\|_{\mathrm{TV}}
    \right]
    \nonumber\\
    &\qquad\le
    \sqrt{\frac{\ln 2}{2}
    \mathbb E\left[
    D_{\mathrm{KL}}\!\left(
    P_{A|B}(\cdot|B)
    \middle\|
    \operatorname{Unif}(\mathbb{GF}(q))
    \right)
    \right]}.
    \label{eq:wz_jensen_average}
\end{align}
The average divergence from the uniform distribution satisfies
\begin{equation}
    \mathbb E\left[
    D_{\mathrm{KL}}\!\left(
    P_{A|B}(\cdot|B)
    \middle\|
    \operatorname{Unif}(\mathbb{GF}(q))
    \right)
    \right]
    =\log q-\mathsf H(A|B).
    \label{eq:wz_uniform_divergence_entropy}
\end{equation}
Combining \eqref{eq:wz_jensen_average} and
\eqref{eq:wz_uniform_divergence_entropy} gives
\begin{equation}
    \mathbb E\left[
    \left\|P_{A|B}(\cdot|B)
    -\operatorname{Unif}(\mathbb{GF}(q))\right\|_{\mathrm{TV}}
    \right]
    \le
    \sqrt{\frac{\ln 2}{2}
    \left(\log q-\mathsf H(A|B)\right)}.
    \label{eq:wz_entropy_uniform_bound}
\end{equation}
Using~\eqref{eq:wz_ach_frozen_set} in
\eqref{eq:wz_entropy_uniform_bound}, for $i\in\mathcal F_N$,
\begin{equation}
    \mathbb E_{P_N}\Bigl[
    \bigl\|
    P_{\tilde U_i|\tilde U^{i-1},X^N}
    -\operatorname{Unif}(\mathbb{GF}(q))
    \bigr\|_{\mathrm{TV}}
    \Bigr]
    \le \sqrt{\frac{\delta_N\ln q}{2}},
\end{equation}
where the conditioning arguments are suppressed for readability. Thus,
\begin{equation}
    \|P_N-Q_N\|_{\mathrm{TV}}
    \le N\sqrt{\frac{\delta_N\ln q}{2}}
    \xrightarrow[N\to\infty]{}0.
    \label{eq:wz_tv_convergence}
\end{equation}

\emph{Decoding reliability.}
The decoder is given $Y^N$, $\tilde U_{\mathcal I_N}$, and
$C_{\mathcal F_N}$. It reconstructs
$\tilde U^N$ successively. At indices in
$\mathcal F_N\cup\mathcal I_N$, it inserts the known symbol. At an index
$i\in\mathcal L_N$, it uses the maximum a posteriori (MAP) rule
\begin{equation}
    \hat{\tilde U}_i
    =
    \arg\max_{a\in\mathbb{GF}(q)}
    P_{\tilde U_i|\tilde U^{i-1},Y^N}
    (a|\hat{\tilde U}^{i-1},Y^N).
    \label{eq:wz_sc_map}
\end{equation}
Ties are resolved according to an arbitrary fixed rule.
The decoder can equivalently be viewed as a deterministic function of
$(Y^N,\tilde U_{\mathcal F_N},\tilde U_{\mathcal I_N})$. Under $Q_N$ these
known coordinates equal
$(C_{\mathcal F_N},\tilde U_{\mathcal I_N})$; under $P_N$ we analyze
the same decoder supplied with the true coordinates. For $i\in\mathcal L_N$,
let $E_i$ be the event that the first decoding error occurs at $i$. On
$E_i$, all preceding decisions are correct, so the MAP rule in
\eqref{eq:wz_sc_map} is evaluated with the true prefix
$\tilde U^{i-1}$. The events $\{E_i:i\in\mathcal L_N\}$ are disjoint and
their union is the block-error event. Therefore,
\begin{equation}
    P_N(\hat{\tilde U}^N\ne\tilde U^N)
    \le
    \sum_{i\in\mathcal L_N}
    \mathcal P_{\mathrm e}
    (\tilde U_i|\tilde U^{i-1},Y^N),
    \label{eq:wz_sc_union}
\end{equation}
where
$\mathcal P_{\mathrm e}(A|B)
\triangleq\sum_b P_B(b)\bigl[1-\max_a P_{A|B}(a|b)\bigr]$
is the MAP error probability averaged over the conditioning variable $B$. The Hellman--Raviv equivocation bound
\cite{hellman1970probability} gives, for finite-alphabet $A$,
$\mathcal P_{\mathrm e}(A|B)\le\mathsf H(A|B)/2$. Consequently,
\begin{equation}
    P_N(\hat{\tilde U}^N\ne\tilde U^N)
    \le \frac{N\delta_N\log q}{2}
    \xrightarrow[N\to\infty]{}0.
    \label{eq:wz_target_error}
\end{equation}
The block-error event is defined by the same deterministic decoding rule
under both laws. Since probabilities of a fixed event differ by at most
total variation,
\begin{equation}
    Q_N(\hat{\tilde U}^N\ne\tilde U^N)
    \le
    \frac{N\delta_N\log q}{2}+\|P_N-Q_N\|_{\mathrm{TV}}
    \xrightarrow[N\to\infty]{}0.
    \label{eq:wz_induced_error}
\end{equation}
The decoder then forms
$\hat U^N=\hat{\tilde U}^NG_N^{-1}$ and
$\hat X^N=g_N(\hat U^N,Y^N)$. Since $G_N$ is invertible,
$\hat U^N\ne U^N$ if and only if
$\hat{\tilde U}^N\ne\tilde U^N$.

\emph{Rate analysis.}
By
\eqref{eq:wz_cor_entropy_ordering}, $h_{i,N}^{X}\le h_{i,N}^{Y}$ for every
$i$. Since $\delta_N<1/2$ for all sufficiently large $N$,
$\mathcal F_N\cap\mathcal L_N=\varnothing$, and therefore
\begin{equation}
    R_N
    =\left(1-\frac{|\mathcal F_N|}{N}
    -\frac{|\mathcal L_N|}{N}\right)\log q.
    \label{eq:wz_rate_partition}
\end{equation}
We use the same Bhattacharyya-to-entropy sandwich argument as in
Appendix~\ref{app:wz_polarization},
\eqref{eq:app_wz_qary_entropy_bhattacharyya}--\eqref{eq:app_wz_fast_high_sandwich},
now at both the high-entropy and low-entropy endpoints.
Choose $\beta'$ such that $\beta<\beta'<1/2$. The fast polarization results
in~\cite{shuval2019fast} can be transferred to the
normalized entropy sets using the same comparison bound. To see this
explicitly, let $z_{i,N}^{Z}$ denote the $q$-ary Bhattacharyya parameter
associated with $h_{i,N}^{Z}$ and set
$\delta_N'=2^{-N^{\beta'}}$. Since
\begin{equation}
    (z_{i,N}^{Z})^2\le h_{i,N}^{Z}
    \le\log_q\!\left(1+(q-1)z_{i,N}^{Z}\right),
    \label{eq:wz_qary_entropy_bhattacharyya}
\end{equation}
for all sufficiently large $N$,
\begin{align}
    \{i:z_{i,N}^{X}>1-\delta_N'\}
    &\subseteq \mathcal F_N
    \subseteq \{i:h_{i,N}^{X}>1-\epsilon\},
    \label{eq:wz_frozen_sandwich}\\
    \{i:z_{i,N}^{Y}<\delta_N'\}
    &\subseteq \mathcal L_N
    \subseteq \{i:h_{i,N}^{Y}<\epsilon\},
    \label{eq:wz_low_sandwich}
\end{align}
for any fixed $\epsilon\in(0,1/2)$. Indeed,
$2\delta_N'<\delta_N$ and $\delta_N<\epsilon$ eventually because
$\beta'>\beta$. Thus, $z_{i,N}^{X}>1-\delta_N'$ implies
$h_{i,N}^{X}\ge(z_{i,N}^{X})^2>1-2\delta_N'>1-\delta_N$, while
$z_{i,N}^{Y}<\delta_N'$ implies
$h_{i,N}^{Y}\le\log_q(1+(q-1)\delta_N')<\delta_N$. The last inequality
holds eventually because $\beta'>\beta$. The reverse inclusions into the
fixed-threshold entropy sets follow from $\delta_N<\epsilon$. The
fixed-threshold entropy densities and the fast-to-zero Bhattacharyya
density follow from Theorem~7 and Corollary~8 of~\cite{shuval2019fast},
whereas Theorem~13 gives the fast-to-one Bhattacharyya density. Thus, the
left and right sets in each of
\eqref{eq:wz_frozen_sandwich}--\eqref{eq:wz_low_sandwich} have the same
respective limiting density. Hence, the
sandwich theorem gives
\begin{align}
    \frac{|\mathcal F_N|}{N}\log q
    &\longrightarrow \overline{\mathsf H}(U|X),
    \label{eq:wz_frozen_density}\\
    \frac{|\mathcal L_N|}{N}\log q
    &\longrightarrow \log q-\overline{\mathsf H}(U|Y).
    \label{eq:wz_low_density}
\end{align}
It follows from \eqref{eq:wz_rate_partition} that
\begin{equation}
    \lim_{N\to\infty}R_N
    =\overline{\mathsf H}(U|Y)-\overline{\mathsf H}(U|X)
    =\overline{\mathsf I}(X;U|Y),
    \label{eq:wz_rate_entropy_difference}
\end{equation}
where the last equality follows from the block Markov relation, as shown in
\eqref{eq:app_wz_cmi_entropy_difference}. This proves
\eqref{eq:wz_ach_rate}.

\emph{Distortion analysis.}
Recall that $d_{\max}=\max_{x,\hat x\in\mathcal X}d(x,\hat x)<\infty$.
Since $0\le d_N\le d_{\max}$,
\begin{align}
    \mathbb E_{Q_N}\left[
    d_N\bigl(X^N,g_N(U^N,Y^N)\bigr)
    \right]
    \le{}&
    \mathbb E_{P_N}\left[
    d_N\bigl(X^N,g_N(U^N,Y^N)\bigr)
    \right]
    \nonumber\\
    &+d_{\max}\|P_N-Q_N\|_{\mathrm{TV}}.
    \label{eq:wz_ideal_distortion_tv}
\end{align}
On the event $\hat U^N=U^N$, the actual and ideal reconstructions coincide.
Thus,
\begin{align}
    \mathbb E_{Q_N}[d_N(X^N,\hat X^N)]
    \le{}&
    \mathbb E_{P_N}\left[
    d_N\bigl(X^N,g_N(U^N,Y^N)\bigr)
    \right]
    \nonumber\\
    &+d_{\max}\|P_N-Q_N\|_{\mathrm{TV}}
    +d_{\max}Q_N(\hat U^N\ne U^N).
    \label{eq:wz_final_distortion_bound}
\end{align}
Taking the limit superior in~\eqref{eq:wz_final_distortion_bound} and using
the target-distortion assumption~\eqref{eq:wz_target_distortion}, together
with~\eqref{eq:wz_tv_convergence} and~\eqref{eq:wz_induced_error}, gives
\begin{equation}
    \limsup_{N\to\infty}
    \mathbb E_{Q_N}[d_N(X^N,\hat X^N)]
    \le D.
\end{equation}

\emph{Derandomization.}
Let $W_N$ collect the frozen vector $C_{\mathcal F_N}$ and independent
uniform random seeds implementing the encoder's conditional sampling,
all independent of $(X^N,Y^N)$. Fixing $W_N=w$ makes the encoder a
deterministic function of $X^N$.

For a fixed realization $w$, define
\begin{align}
    e_N(w)
    &\triangleq
    \Pr(\hat U^N\ne U^N\mid W_N=w),\\
    \mathrm D_N(w)
    &\triangleq
    \mathbb E[d_N(X^N,\hat X^N)\mid W_N=w].
\end{align}
Thus, both quantities are averaged over the source and side-information law
with the code fixed. The law of total expectation shows that averaging over
$W_N$ recovers the randomized construction:
\begin{align}
    \bar e_N
    \triangleq \mathbb E[e_N(W_N)]
    &=Q_N(\hat U^N\ne U^N),\\
    \overline{\mathrm D}_N
    \triangleq \mathbb E[\mathrm D_N(W_N)]
    &=\mathbb E_{Q_N}[d_N(X^N,\hat X^N)].
\end{align}
Hence, \eqref{eq:wz_induced_error} and the preceding distortion bound give
\begin{equation}
    \bar e_N\longrightarrow 0,
    \qquad
    \limsup_{N\to\infty}\overline{\mathrm D}_N\le D.
\end{equation}
Define the set of realizations with small error by
\begin{equation}
    \mathcal G_N
    \triangleq
    \left\{w: e_N(w)\le \sqrt{\bar e_N}\right\}.
\end{equation}
If $\bar e_N>0$, Markov's inequality gives
\begin{equation}
\begin{aligned}
    \Pr(W_N\notin\mathcal G_N)
    &=\Pr\bigl(e_N(W_N)>\sqrt{\bar e_N}\bigr)\\
    &\le \frac{\mathbb E[e_N(W_N)]}{\sqrt{\bar e_N}}
    =\sqrt{\bar e_N}.
\end{aligned}
\end{equation}
If $\bar e_N=0$, nonnegativity implies $e_N(W_N)=0$ almost surely. Hence,
in both cases,
$\Pr(W_N\in\mathcal G_N)\ge1-\sqrt{\bar e_N}$, which is positive for all
sufficiently large $N$.
Since the distortion is nonnegative,
\begin{equation}
\begin{aligned}
    \mathbb E[\mathrm D_N(W_N)\mid W_N\in\mathcal G_N]
    &=
    \frac{\mathbb E[\mathrm D_N(W_N)
    \mathbf 1\{W_N\in\mathcal G_N\}]}
    {\Pr(W_N\in\mathcal G_N)}\\
    &\le
    \frac{\overline{\mathrm D}_N}
    {\Pr(W_N\in\mathcal G_N)}\\
    &\le
    \frac{\overline{\mathrm D}_N}{1-\sqrt{\bar e_N}}.
\end{aligned}
\end{equation}
Thus, some $w_N\in\mathcal G_N$ satisfies
\begin{equation}
    \mathrm D_N(w_N)
    \le
    \frac{\overline{\mathrm D}_N}{1-\sqrt{\bar e_N}}.
\end{equation}
For the deterministic code indexed by this realization,
\begin{align}
    e_N(w_N)
    &\le \sqrt{\bar e_N}\longrightarrow 0,\\
    \limsup_{N\to\infty}\mathrm D_N(w_N)
    &\le D.
\end{align}
The second limit follows because
$(1-\sqrt{\bar e_N})^{-1}\to1$ and
$\limsup_N\overline{\mathrm D}_N\le D$. The fixed frozen vector is part of
the code specification, so fixing $w_N$ does not alter the message set or
its rate. This proves
\eqref{eq:wz_ach_distortion}.
\end{proof}

\section{Proof of Theorem \ref{th:pg_wz}}\label{app:PG_WZ}
This appendix extends the proof of Theorem~\ref{th:pg_rd} to the
Wyner--Ziv objective. The additional step is to account for the
test-channel dependence of the conditional distribution given the side
information.

\begin{proof}
All expectations below are with respect to
$P_{X^N,Y^N}P_{U^N|X^N}^{\phi}P_{\hat X^N|U^N,Y^N}^{\psi}$, with
$P_{X^N,Y^N}$ and $\psi$ fixed. Under the assumptions of
Theorem~\ref{th:pg_wz}, the softmax test-channel probabilities are
strictly positive and differentiable, and differentiation may be
interchanged with the finite summations. For each $y^N$ in the support
of $P_{Y^N}$, the induced conditional distribution is
\begin{equation}
    P_{U^N|Y^N}^{\phi}(u^N|y^N)
    =
    \sum_{x^N}
    P_{X^N|Y^N}(x^N|y^N)
    P_{U^N|X^N}^{\phi}(u^N|x^N),
\end{equation}
which is also strictly positive and differentiable. The Markov relation
$U^N-X^N-Y^N$ gives
$\mathsf I_\phi(X^N;U^N|Y^N)=\mathbb E[\imath_\phi(X^N,Y^N,U^N)]$.
Applying the same product-rule and score-function argument as in
\eqref{eq:app_rd_mi_grad_decomp}, we obtain
\begin{align}
    \nabla_\phi\mathsf I_\phi(X^N;U^N|Y^N)
    &=
    \mathbb E\left[
        \nabla_\phi\ln P_{U^N|X^N}^{\phi}(U^N|X^N)
        \imath_\phi(X^N,Y^N,U^N)
    \right] \nonumber\\
    &\quad+
    \mathbb E\left[
        \nabla_\phi\imath_\phi(X^N,Y^N,U^N)
    \right].
    \label{eq:app_wz_cmi_grad_decomp}
\end{align}
For $Z=X$ or $Z=Y$, normalization of the corresponding conditional
distribution yields
\begin{align}
    &\mathbb E\left[
        \nabla_\phi\log P_{U^N|Z^N}^{\phi}(U^N|Z^N)
    \right] \nonumber\\
    &\quad=
    \frac{1}{\ln 2}\sum_{z^N}P_{Z^N}(z^N)
    \nabla_\phi\sum_{u^N}P_{U^N|Z^N}^{\phi}(u^N|z^N)
    =0,
\end{align}
where the sum over $z^N$ is restricted to the support of $P_{Z^N}$.
Taking the difference of these identities shows that
$\mathbb E[\nabla_\phi\imath_\phi(X^N,Y^N,U^N)]=0$.
Finally, $P_{\hat X^N|U^N,Y^N}^{\psi}$ does not depend on $\phi$ when the reconstruction
parameters are fixed. Applying the score-function identity as in
\eqref{eq:app_rd_distortion_grad} to the joint distribution therefore yields
\[
    \nabla_\phi\mathbb E[d_N(X^N,\hat X^N)]
    =\mathbb E\left[
        \nabla_\phi\ln P_{U^N|X^N}^{\phi}(U^N|X^N)
        d_N(X^N,\hat X^N)
    \right].
\]
Combining this identity
with \eqref{eq:app_wz_cmi_grad_decomp} and the definition of
$J_\lambda(\phi,\psi)$ proves \eqref{eq:wz_exact_phi_gradient}.
\end{proof}

\section{Implementation Details for Experiments}\label{app:implementation_details}

This appendix provides the implementation details used to reproduce the numerical experiments. It first gives the recursive loss computation of the neural polar decoder (NPD) and specifies the NPD architecture and training parameters for the lossless experiments, including the preprocessing and bits per base (bpb) accounting used for the \texttt{E. coli} benchmark. It then reports the corresponding parameters for the lossy source-coding and Wyner--Ziv experiments, together with the complete Wyner--Ziv optimization procedure. The preprocessing rules are stated explicitly in prose, while the neural-network dimensions, optimization parameters, and alternating update schedules are collected in tables.

\subsection{Recursive Loss Computation}
Following~\cite{aharoni2026optimized}, Algorithm~\ref{alg:training-nsc} gives the recursive procedure $\mathsf{Loss}(\mathbf{e},\mathbf{u},\mathsf{L})$ called by Algorithm~\ref{alg:memory}, which evaluates the per-block loss $\cL(x_1^N)$ in \eqref{eq:npd_full_loss} along the SC decoding tree. Its inputs are the embeddings $\mathbf{e}$ at the current depth, the polarized symbols $\mathbf{u}$ of the corresponding subtree, and the loss accumulator $\mathsf{L}$, and it returns the updated accumulator together with the symbols $\mathbf{v}$ at the current depth. At each internal node, the check-node network $F_\theta$ and the bit-node network $G_\theta$ produce the embeddings of the two children, the recursion is applied to each child, and the returned symbols are combined through the kernel to obtain the symbols at the current depth, which supply the auxiliary CE terms of \eqref{eq:npd_full_loss}. The same recursion, with the appropriate initial embeddings, is used for the conditional NPDs of Sections~\ref{sec: RD} and~\ref{sec:WZ}.

\begin{algorithm}[H]
    \caption{ $\mathsf{Loss}(\mathbf{e},\mathbf{u}, \mathsf{L})$}
    \label{alg:training-nsc}
    \begin{algorithmic}
    \State $N=\f{\mathsf{dim}}{\mathbf{u}}$
    \If{$N=1$}
        \State $\mathsf{L} = \mathsf{L} +L^\theta_\mathsf{ce}(\text{softmax}(\mathbf{H_{\theta}(e)}),\mathbf{u})$ \\
        \Return $\mathsf{L},\mathbf{u}$
    \EndIf
    \State Split $\mathbf{e}$ into even and odd indices $\mathbf{e}_e, \mathbf{e}_o$
    \State $\mathbf{e}_{\mathsf{C}} = \f{F_\theta}{\mathbf{e}_e, \mathbf{e}_o}$ \Comment{Check-node}
    \State  $\mathsf{L},\mathbf{v}_1 = \f{\mathbf{\mathsf{Loss}}}{\mathbf{e}_{\mathsf{C}}, \mathbf{u}_{1}^{N/2}, \mathsf{L}}$
    \vspace{0.1em}
    \State $\mathbf{e}_{\mathsf{B}} = \f{G_\theta}{\mathbf{e}_e, \mathbf{e}_o, \mathbf{v}_1}$ \Comment{Bit-node}
    \vspace{0.1em}
    \State $\mathsf{L},\mathbf{v}_2 = \f{\mathbf{\mathsf{Loss}}}{\mathbf{e}_{\mathsf{B}}, \mathbf{u}_{N/2+1}^{N}, \mathsf{L}}$
    \State $\mathbf{v} = [\mathbf{v}_1, \mathbf{v}_2]$
    \State $\mathbf{v} = \mathbf{v} \left(I_{N/2} \otimes G_2\right) R_N$ \Comment{Symbols in current depth}
    \State $\mathsf{L} = \mathsf{L} + \sum_{i=1}^N L^\theta_\mathsf{ce}(\text{softmax}(H_\theta(e_i)), v_i)$ \Comment{Loss in current depth}\\
    \Return $\mathsf{L},\mathbf{v}$
    \end{algorithmic}
\end{algorithm}

\subsection{Lossless Experiments}

The lossless experiments use the same recursive NPD structure for both the synthetic sources and the real-world DNA data. The distinction between these settings is in the size of the neural modules: the synthetic experiments use a smaller network, whereas the \texttt{E. coli} experiment uses a larger embedding and hidden dimension to accommodate the richer empirical distribution. The block length, optimizer, learning rate, and architecture parameters used in these experiments are listed in Table~\ref{tab:lossless_reproducibility}. The modules $E_{\theta_E}$, $F_{\theta_F}$, and $G_{\theta_G}$ are multilayer perceptrons with the number of hidden layers listed in the tables, each of width $k$, rather than the single-hidden-layer networks of Definition~\ref{def:NN_function_class}. For a fixed number $\ell$ of hidden layers, the cost of one module evaluation is $O(kd+(\ell-1)k^2)$ in place of $O(kd)$, which is independent of $N$ and $|\mathcal S|$, so the $O(N\log N)$ scaling in Remark~\ref{rem:npd_complexity} is unaffected.

\subsubsection{DNA Preprocessing and Rate Accounting}
For the \texttt{E. coli} experiment, the preprocessing procedure determines the effective source alphabet, the parsed sequence, and the denominator used when reporting bpb. The Large Canterbury Corpus file is parsed as plain text. FASTA header lines are removed when present, newline characters are stripped, all symbols are converted to uppercase, and only the four canonical DNA bases \texttt{A}, \texttt{C}, \texttt{G}, and \texttt{T} are retained. Ambiguity symbols, including \texttt{N}, are discarded. The retained bases are assigned integer labels using the mapping \texttt{A}$\mapsto 0$, \texttt{C}$\mapsto 1$, \texttt{G}$\mapsto 2$, and \texttt{T}$\mapsto 3$. We use $\mathbb{GF}(4)=\mathbb{GF}(2)(\alpha)$, where $\alpha^2+\alpha+1=0$. The labels $0,1,2,3$ represent the field elements $0,1,\alpha,\alpha+1$, respectively, and the polar kernel in \eqref{eq:kernel2} uses this $\alpha$. The sequence is split chronologically into training and validation portions, with the first $95\%$ used for training and the remaining $5\%$ used for validation. For each validation block, we compute $u^N=x^NG_N$ and use the trained NPD to generate $P_\theta(u_i|u^{i-1})$ successively through the SC graph. These probabilities are supplied to an arithmetic coder to encode $u^N$. The decoder recovers $u^N$ from the arithmetic-coded bitstream and then reconstructs $x^N=u^NG_N^{-1}$. The reported $1.883$ bpb is the actual arithmetic-coded length normalized by the number of retained validation bases. This is a held-out evaluation with a model trained on the first $95\%$ of the sequence.

\begin{table}[H]
\centering
\captionsetup{font=normalsize}
\normalsize
\caption{NPD architecture and training parameters for the lossless experiments. The synthetic-source column applies to Figs.~\ref{fig:alphabet} and~\ref{fig:memory}. The \texttt{E. coli} rate is the final NPD-driven arithmetic-coding rate.}
\label{tab:lossless_reproducibility}
\begin{tabular}{lcc}
\hline
Parameter & Synthetic sources & \texttt{E. coli} \\
\hline
Training block length $N=2^{n_\mathsf{t}}$ & 1024 & 1024 \\
Number of hidden layers in $E_{\theta_E},F_{\theta_F},G_{\theta_G}$ & 2 & 3 \\
Number of hidden layers in $H_{\theta_H}$ & 1 & 1 \\
NPD embedding dimension $d$ & 64 & 128 \\
NPD hidden units $k$ & 128 & 256 \\
Mini-batch size & 20 & 20 \\
Optimizer & Adam & Adam \\
Learning rate & 0.001 & 0.001 \\
\hline
\end{tabular}
\end{table}

\subsection{Lossy Experiments}

For the lossy source-coding experiments, the NPDs are trained jointly with the neural test channel according to the alternating estimation--optimization procedure described in Section~\ref{sec: RD}. Different operating points on the rate--distortion curve are obtained by repeating the training procedure for different values of the Lagrange multiplier $\lambda$, which controls the tradeoff between mutual information and expected distortion. The same architecture and optimization parameters are used for the binary symmetric source and the first-order Markov source, so a single column is sufficient in Table~\ref{tab:lossy_reproducibility}. The table reports the NPD dimensions, the transformer embedding dimension of the test channel, and the number of update steps used for each stage of the alternating procedure. In both the lossy and Wyner--Ziv experiments, the test-channel transformer uses four attention heads and two layers.

\begin{table}[H]
\centering
\captionsetup{font=normalsize}
\normalsize
\caption{Architecture and training parameters for the lossy source-coding experiments. The same implementation parameters are used for the binary symmetric source (BSS) and Markov source models. The optimization-step rows give the alternating schedule within each phase: NPD updates are followed by test-channel updates, and the phase is repeated until convergence.}
\label{tab:lossy_reproducibility}
\begin{tabular}{lc}
\hline
Parameter & BSS and Markov sources \\
\hline
Training block length $N=2^{n_\mathsf{t}}$ & 1024 \\
Number of hidden layers in $E_{\theta_{E_1}},E_{\theta_{E_2}},F_{\theta_F},G_{\theta_G}$ & 2 \\
Number of hidden layers in $H_{\theta_H}$ & 1 \\
NPD embedding dimension $d$ & 128 \\
NPD hidden units $k$ & 256 \\
Test-channel transformer embedding dimension & 128 \\
Mini-batch size & 15 \\
NPD optimization steps per phase & 2 \\
Test-channel optimization steps per phase & 1 \\
Optimizer & Adam \\
Learning rates $(\gamma_\theta,\gamma_\phi)$ & (0.005,0.0025) \\
\hline
\end{tabular}
\end{table}

\subsection{Wyner--Ziv Experiments}

Algorithm~\ref{alg:WZ} summarizes the Wyner--Ziv estimation--optimization procedure of Section~\ref{sec:WZ_opt}. It follows Algorithm~\ref{alg:RD}, with the estimation losses \eqref{eq:wz_loss_y} and \eqref{eq:wz_loss_x}, sampling from the reconstruction model $P_{\hat X^N|U^N,Y^N}^{\psi}$, and joint updates of $(\phi,\psi)$ according to \eqref{eq:wz_phi_update}--\eqref{eq:wz_psi_update}.

\begin{algorithm}[H]
    \caption{Wyner--Ziv optimization via NPDs}
    \label{alg:WZ}
    \textbf{input:}
    Block length $N$, dataset size $M$,
    number of iterations $\mathsf{N_{iters}}$, learning rates
    $\gamma_\theta,\gamma_\phi,\gamma_\psi$, Lagrange multiplier $\lambda$\\
    \textbf{output:} Optimized $\theta^\ast,\phi^\ast,\psi^\ast$
    \algrule
    \begin{algorithmic}
    \State Initialize the weights of
    $\theta=\{\theta^Y,\theta^X\}$, $\phi$, and $\psi$
    \State Generate
    $\mathcal D_{M,N}^{\phi}
    \sim
    P_{X^N,Y^N}P_{U^N|X^N}^{\phi}$
    \State \underline{\textbf{Warm-up:}} repeat the estimation step below with $\phi$ fixed

    \For{$\mathsf{N_{iters}}$ iterations}

        \State \underline{\textbf{Estimation step:}}
        \State Sample $(x^N,y^N,u^N)\sim \mathcal D_{M,N}^{\phi}$
        \State Compute $\tilde u^N=u^NG_N$
        \State Compute
        $\mathbf e_0^{Y}=E_{\theta_{E_Y}}(y^N)$
        and
        $\mathbf e_0^{X}=E_{\theta_{E_X}}(x^N)$
        \State Compute
        $\mathcal L_Y(y^N,u^N)$ using \eqref{eq:wz_loss_y}
        and
        $\mathcal L_X(x^N,u^N)$ using \eqref{eq:wz_loss_x}
        \State Update
        \[
            \theta
            :=
            \theta
            -
            \gamma_\theta
            \nabla_\theta
            \left[
                \mathcal L_Y(y^N,u^N)
                +
                \mathcal L_X(x^N,u^N)
            \right].
        \]
        \State \underline{\textbf{Optimization step:}}
        \State Sample $(x^N,y^N,u^N)\sim \mathcal D_{M,N}^{\phi}$
        \State Compute $\tilde u^N=u^NG_N$
        \State Sample $\hat x^N\sim P_{\hat X^N|U^N,Y^N}^{\psi}(\cdot|u^N,y^N)$
        to form $\mathcal D_{M,N}^{\phi,\psi}$
        \State Compute
        $\hat\imath_{\theta}(x^N,y^N,u^N)$ using
        \eqref{eq:wz_npd_information_density}
        \State Compute
        $\widehat{\nabla_\phi J_\lambda}$ using
        \eqref{eq:wz_empirical_phi_gradient}
        \State Compute
        $\widehat{\nabla_\psi J_\lambda}$ by Monte Carlo averaging of
        \eqref{eq:wz_psi_gradient}
        \State Update
        \[
        \begin{aligned}
            \phi
            &:=
            \phi
            -
            \gamma_\phi
            \widehat{\nabla_\phi J_\lambda},\\
            \psi
            &:=
            \psi
            -
            \gamma_\psi
            \widehat{\nabla_\psi J_\lambda}.
        \end{aligned}
        \]
        \State Generate
        $\mathcal D_{M,N}^{\phi}
        \sim
        P_{X^N,Y^N}P_{U^N|X^N}^{\phi}$

    \EndFor \\

    \Return $\theta^\ast,\phi^\ast,\psi^\ast$
    \end{algorithmic}
\end{algorithm}

The Wyner--Ziv experiments use conditional NPDs, an auxiliary test channel, and a reconstruction network that uses the decoder side information. Table~\ref{tab:wz_reproducibility} lists the parameters for the binary symmetric and Markov sources. The larger dimensions used for the Markov source account for its memory and the conditional dependence induced by the side-information sequence. For the Markov-source experiment with auxiliary alphabet size $|\mathcal U|=4$, we use $\mathbb{GF}(4)=\mathbb{GF}(2)(\alpha)$, where $\alpha^2+\alpha+1=0$. The labels $0,1,2,3$ represent the field elements $0,1,\alpha,\alpha+1$, respectively. The \texttt{E. coli} Wyner--Ziv experiment uses the configuration in the Markov-source column, except that the auxiliary alphabet size is $|\mathcal U|=5$. 

All neural components are trained at block length $N=2^{n_\mathsf{t}}=1024$ ($n_\mathsf{t}=10$). The coding operating points in Figs.~\ref{fig:WZ_combined} and~\ref{fig:WZ_EColi} are evaluated using SCL decoding at $N=2^{13}=8192$ with list size $L=16$. The learned NPD modules, test channel, and reconstruction model are reused without retraining; their shared parameters are applied over the additional SC recursion levels.

\begin{table}[H]
\centering
\captionsetup{font=normalsize}
\normalsize
\caption{Architecture and training parameters for the Wyner--Ziv experiments with decoder side information. Within each phase, the NPD updates are followed by joint test-channel and reconstruction-network updates, and the phases are repeated until convergence.}
\label{tab:wz_reproducibility}
\begin{tabular}{@{}p{0.44\linewidth}
    >{\centering\arraybackslash}p{0.22\linewidth}
    >{\centering\arraybackslash}p{0.22\linewidth}@{}}
\hline
Parameter & BSS with side information & Markov with side information \\
\hline
Training block length $N=2^{n_\mathsf{t}}$ & 1024 & 1024 \\
Auxiliary alphabet size $|\mathcal U|$ & 3 & 4 \\
Hidden layers in NPD modules $E,F,G$ & 2 & 2 \\
Number of hidden layers in $H_{\theta_H}$ & 1 & 1 \\
NPD embedding dimension $d$ & 64 & 128 \\
NPD hidden units $k$ & 128 & 256 \\
Test-channel transformer embedding dimension & 128 & 256 \\
Reconstruction network $g_\psi$ hidden dimension & 128 & 256 \\
Mini-batch size & 15 & 10 \\
NPD optimization steps per phase & 2 & 2 \\
Test-channel/reconstruction optimization steps per phase & 1 & 1 \\
Optimizer & Adam & Adam \\
Learning rates $\gamma_\theta=\gamma_\phi$ & 0.002 & 0.002 \\
Learning rate $\gamma_\psi$ & 0.001 & 0.001 \\
\hline
\end{tabular}
\end{table}
\end{appendices}
\bibliography{ref_verified}
\bibliographystyle{IEEEtran}

\end{document}